\documentclass[ecta,nameyear,draft]{econsocart}

\usepackage{amsmath,blkarray,bm}
\usepackage{amsfonts} 
\usepackage{amsthm} 
\usepackage{amssymb}
\usepackage{tikz}
\usepackage[mathscr]{euscript}
\usepackage{caption}
\usepackage{subcaption}
\usepackage{xcolor}
\usetikzlibrary{positioning}
\usepackage{pgfplots} 
\usepackage{pdfsync} 
\usepackage{xr} 

\usepackage[normalem]{ulem}

\usetikzlibrary{intersections}
\usepgfplotslibrary{fillbetween} 
\usepackage{centernot}
\usepackage{mathtools}
\usepackage{ stmaryrd }
\usetikzlibrary{shapes.geometric, arrows,positioning}
\usetikzlibrary{arrows}
\usepackage{bbm} \usepackage{natbib} \usepackage{setspace} \usepackage{enumitem}
\usepackage[flushleft]{threeparttable}
\usepackage{algorithm}
\usepackage{algorithmic}
\usepackage{multirow}

\newcommand{\sgn}{\operatorname{sgn}}
\newcommand{\supp}{\operatorname{supp}}
\newcommand{\interior}{\operatorname{int}}

\def\argmax{\mathop{\rm arg\,max}}
\def\argmin{\mathop{\rm arg\,min}}

\def\core{\mathop{\rm core}} 
\def\ri{\mathop{\rm ri}}   

\newcommand{\salg}{\mathfrak{F}}
\newcommand{\SY}{\Sigma_Y} 
\newcommand{\SX}{\Sigma_X} 
\newcommand{\E}{\mathbb{E}}
\newcommand{\R}{\mathbb{R}}

\newcommand{\ex}{X} 
\newcommand{\ey}{Y} 
\newcommand{\ez}{Z} 
\newcommand{\eu}{U} 
\newcommand{\er}{R} 

\newcommand{\eX}{{\boldsymbol{X}}}

\newcommand{\cQ}{\mathcal{Q}}
\newcommand{\fQ}{\mathfrak{Q}}
\newcommand{\fq}[1]{\mathfrak{q}_{#1}}

\newcommand{\cP}{\mathcal{P}}
\newcommand{\cX}{\mathcal{X}}  
\newcommand{\cY}{\mathcal{Y}}  
\newcommand{\cU}{\mathcal{U}}  
\newcommand{\cS}{\mathcal{S}}  
\newcommand{\cM}{\mathcal{M}}  
\newcommand{\cH}{\mathcal{H}}  
  
\newcommand{\cC}{\mathcal{C}}  
  
\newcommand{\cA}{\mathcal{A}}  
\newcommand{\cJ}{\mathcal{J}}

\newcommand{\f}{F_\theta} 
\newcommand{\g}{g} 
\newcommand{\G}{G} 
\newcommand{\Qs}[1]{Q_{#1}} 
\newcommand{\qs}[1]{q_{#1}} 
\newcommand{\Ps}{P} 
\newcommand{\ps}[1]{p_{#1}} 
\newcommand{\contf}{\nu_\theta} 
\newcommand{\PP}{\mathbf{P}} 
\newcommand{\crit}{L} 
\newcommand{\score}[2]{s_{#1}(#2)}
\newcommand{\nq}{\rho} 

\newcommand{\dKL}[2]{I(#1||#2)}

\newcommand{\one}{\mathbf{1}}
\newcommand{\dist}{\mathrm{dist}}

\newcommand\independent{\protect\mathpalette{\protect\independenT}{\perp}}
\def\independenT#1#2{\mathrel{\rlap{$#1#2$}\mkern2mu{#1#2}}}

\RequirePackage[colorlinks,citecolor=blue,linkcolor=blue,urlcolor=blue,pagebackref]{hyperref}

\startlocaldefs

\numberwithin{equation}{section}
\theoremstyle{plain}
\newtheorem{theorem}{Theorem}[section] \newtheorem{proposition}{Proposition}[section] \newtheorem{lemma}{Lemma}[section] \newtheorem{corollary}{Corollary}[section]  \newtheorem{assumption}{Assumption}
\newtheorem{assumptionA}{Assumption}[section]  \newtheorem{remark}{Remark}[section] 
\newtheorem{definition}{Definition}[section]
\theoremstyle{definition}
\newtheorem{example}{Example}
\newtheorem{exampleA}{Example}[section] 
\newcommand{\citeposs}[1]{\citeauthor{#1}'s \citeyearpar{#1}}

\endlocaldefs

\pgfplotsset{compat=1.18}

\begin{document}
\begin{frontmatter}

\title{Information Based Inference\\
in Models with Set-Valued Predictions and Misspecification}
\runtitle{Information Based Inference with Misspecification}  

\begin{aug}
%
%
%
\author[id=au1,addressref={add1}]{\fnms{Hiroaki}~\snm{Kaido}\ead[label=e1]{hkaido@bu.edu}}
\author[id=au2,addressref={add2}]{\fnms{Francesca}~\snm{Molinari}\ead[label=e2]{fm72@cornell.edu}}
\address[id=add1]{%
\orgdiv{Department of Economics},
\orgname{Boston University}}

\address[id=add2]{%
\orgdiv{Department of Economics},
\orgname{Cornell University}}
\end{aug}

\support{We thank anonymous referees, Xiaoxia Shi, and audiences at Bonn/Mannheim, Bristol/Warwick, BU, CEMFI, Chicago, Columbia, Cornell, FGV-EESP, JHU, Michigan, Nebraska, NYU, Queen's Mary, Tokyo, Toulouse, UCD, UCL, UCLA, UCSB, UPF, USC, Yale, Wisconsin, ESAM21, AMES23, ESWC25, Workshop on Econometrics and Models of Strategic Interaction, and Chamberlain Seminar, for comments.
U.~Byambadalai, S.~Chen, Q.~Han, L.~Hoderlein, Yan~Liu, Yiqi~Liu, P.~Power, R.~Strong, Y.~Sun provided excellent research assistance. We gratefully acknowledge financial support from NSF grants SES-2018498 (Kaido) and 1824375 (Molinari).}
\begin{abstract}
This paper proposes an information-based inference method for partially identified parameters in incomplete models that is valid both when the model is correctly specified and when it is misspecified. 
Key features of the method are:  (i) it is based on minimizing a suitably defined Kullback-Leibler information criterion that accounts for incompleteness of the model and delivers a non-empty pseudo-true set;
(ii) it is computationally tractable; (iii) its implementation is the same for both correctly and incorrectly specified models; (iv) it exploits all information provided by variation in discrete and continuous covariates; (v) it relies on Rao's score statistic, which is shown to be asymptotically pivotal.
\end{abstract}

\begin{keyword}
\kwd{Misspecification}
\kwd{Partial Identification}
\kwd{Rao's score statistic}
\end{keyword}

\end{frontmatter}
\section{Introduction}
Applied research is rarely based on the empirical evidence alone: exogeneity assumptions, behavioral restrictions, distributional and functional form specifications, etc., are routinely imposed to approximate features of complex social and economic phenomena. 
Early on, \citet[p.~169]{koo:rei50} highlighted the importance of imposing restrictions based on prior knowledge of the phenomenon under analysis and some criteria of simplicity, but argued against choosing these restrictions primarily for the purpose of identifiability of the structure the researcher happens to be interested in. 
Yet, even after embracing \citeauthor{koo:rei50}'s perspective, concerns for misspecification remain.

To answer these concerns, we propose a novel information theoretic inference method in the spirit of \cite{White1982} that is asymptotically valid both when the model is correctly or incorrectly specified and point or partially identified. 
The latter case often results when, based on prior knowledge, a researcher is willing to impose modeling restrictions on some features of the model, but remains agnostic about others.
Once the researcher has taken a stand on the model, we provide inference methods for its features that are robust to misspecification.
Our method is easy to implement and addresses three challenges that arise in misspecified partially identified models: 
(i) the set of observationally equivalent parameters may be spuriously tight or empty; (ii) confidence sets constructed assuming correct model specification may (severely) undercover; (iii) their tightness may be misinterpreted as highly informative data.\footnote{Counterpart challenges arise in point identified models, with various solutions put forward since at least \citet{White1982} \citep[see, e.g.,][and references therein for a recent discussion]{han:lee21}.}
The method applies to a specific but wide class of partially identified models that predict a set of values for the endogeneous variables ($\ey$) given the exogenous observed and unobserved ones ($\ex$ and $\eu$, respectively), yielding a set of conditional distributions for $\ey|\ex$.
Many examples belong to this class, including: games with multiple equilibria; discrete choice models with either interval data on covariates,  counterfactual choice sets, endogeneous explanatory variables, or unobserved heterogeneity in choice sets; dynamic discrete choice models; network formation models; and auctions and school choice models under weak assumptions on behavior \citep[see, e.g.,][]{MolinariHOE}.

We adapt the textbook method for (point identified) models that predict a singleton conditional distribution for $\ey|\ex$, to models that predict a set of distributions for $\ey|\ex$, through three steps that jointly yield our main innovations.
In step one, leveraging a result in \citet{art83}, we characterize the exact set of model predicted distributions consistent with all maintained assumptions, and show that with discrete $\ey$ it is a finite dimensional convex polytope.\footnote{Using only a subset of model implications may yield misleading conclusions \citep{ked:li:mou21,MolinariHOE,ber:mol:mol11}; nonetheless, even in this case our method remains applicable.}
Doing so greatly simplifies our next steps as it allows us to bypass the need to model all possible distribution of $\ey|\ex$ through the use of probability mixtures with infinite dimensional nuisance mixing functions. 
In step two, we define a never-empty pseudo-true set, denoted $\Theta^*$, for the parameter vector $\theta$ characterizing the model.
This is the collection of minimizers of a Kullback-Leibler (KL) information criterion measuring the divergence of the set of model-predicted distributions from the distribution of the observed data.
The set $\Theta^*$ shrinks to the pseudo-true parameter vector in \citet{White1982} if the modeling assumptions are augmented so that the model predicts a single distribution. 
When the model predicts multiple distributions, $\Theta^*$ collects the parameter values that minimize the researcher's ignorance about the true structure \citep{Akaike1973,White1982}, recognizing that this ignorance extends to the selection mechanism that picks an element from the set of model predictions.
As in the point identified case, one may wonder to what extent a pseudo-true set is of substantive interest.
In our view, models are only approximations to the true data generating process (DGP) and hence pseudo-true sets are often what researchers estimate in practice; hence, inference methods robust to misspecification are needed.\label{page:recognize_pseudotrue}

To this end, in the third step we obtain a profiled likelihood function by projecting, with respect to the KL divergence measure, the distribution of the observed data on the set of model implied distributions.
A key advantage, yielded by steps one and two, is that this projection is carried out through a computationally simple convex program, which in our leading examples with discrete outcomes features a strictly convex objective and linear constraints. 
As in the textbook case, the pseudo-true set equals the collection of maximizers of the profiled likelihood function.
We next derive a novel score representation for this function, based on $d_\theta$ estimating (score) equations, with $d_\theta$ the number of model's parameters.
These equations depend on the conditional distribution of $\ey|\ex$, which is unknown and needs to be estimated nonparametrically. 
We leverage classic results in the semiparametric inference literature, specifically \citet{Newey:1994aa}, to establish that an orthogonality property holds.
Provided the convergence rate of the nonparametric conditional density estimator is sufficiently fast ($o_p(n^{-1/4})$), this implies that the limit distribution of the averaged score function is insensitive to estimation of the distribution of the data.
We use this result to construct a Rao's score statistic with asymptotically pivotal limit distribution $\chi^2_{d_\theta}$, which we use to test the hypothesis that a candidate parameter vector belongs to $\Theta^*$.
We invert the test to construct a confidence set and show that it is robust to misspecification: it covers each element of $\Theta^*$ with asymptotic probability at least equal to the nominal level $1-\alpha$, uniformly over a large class of DGPs. 
While the size of $\Theta^*$ may or may not be impacted by the extent to which the model is misspecified (see Section~\ref{subsec:geometry_Theta_star}), relative to $\Theta^*$'s size the volume of the confidence set depends only on the sampling variability of the score statistic.

\noindent\textbf{Related Literature.}
\citet{Chen:2011aa,Chen_2018} put forward inference methods that asymptotically cover the $\arg\max$ of a profiled likelihood function, which by textbook arguments is related to the $\arg\min$ of the KL divergence that we focus on. 
However, their method relies on infinite dimensional nuisance functions to represent each possible distribution of $\ey|\ex$, and its validity is established for correctly specified models only.
Our method completely bypasses the use of infinite dimensional nuisance functions, yielding computational advantages and allowing us to consider broader classes of incomplete models.\footnote{\label{fn:chen_validity}See the discussion following \eqref{eq:q_star} for further comparisons. For a dynamic model where misspecification results when subjective beliefs deviate from rational expectations, \citet{che:han:han21} show that a confidence set built using the results in \citet{Chen_2018} is asymptotically valid.}

Only a few recent papers put forth tools for construction of confidence sets that are valid in the presence of misspecification, covering each element of a pseudo-true identified set with an asymptotic probability at least as large as a prespecified nominal level.
\citet{and:kwo22} show that model misspecification can lead to spuriously tight confidence sets while statistical tests have low power at detecting misspecification. 
They propose a notion of pseudo true set, a specification test, and an asymptotically uniformly valid inference method for partially identified models defined by a finite number of unconditional moment inequalities.
Their inference method aggregates violations of the sample moment conditions relaxed by the minimum amount that guarantees that at least one parameter vector in the parameter space satisfies them.
\citet{Stoye20} studies interval identified scalar parameters with asymptotic normality of the estimators of the endpoints.
He obtains a valid and never-empty confidence interval that is free of tuning parameters and simple to compute.
In contrast, our method applies to models for conditional density functions of outcome variables given discrete and continuous covariates.
Allowing for the latter is important: they are commonplace in practice and may yield substantial identifying information.
While a fine discretization or the use of instrument functions \citep[e.g.,][]{and:shi13} may transform conditional moment inequalities in unconditional ones, doing so may incur computational costs associated with an increase in the number of moment inequalities proportional to the cardinality of the discretized support or the number of instrument functions.
We bypass this problem by evaluating directly the contribution of each observation to the score function.
In our simulations (Section~\ref{sec:monte_carlo}), doing so reduces computational time by 228 to 3,564 times relative to \citet{and:shi13}'s method, depending on how fine a discretization one uses.
Our pseudo-true set and inference method are insensitive to which inequalities one uses to characterize the sharp collection of model implied distributions for $\ey|\ex$.
In comparison, much of the related literature often requires moment selection either for computational tractability or as part of the inference procedure, which may substantially impact the population region that the researcher targets \citep{ked:li:mou21} and the properties of confidence sets \citep[e.g.,][]{and:shi13,bug:can:shi17,kai:mol:sto19}.

\noindent\textbf{Outline.} 
Section \ref{sec:examples} introduces the class of models we study.
Section \ref{sec:inference} provides the notion of pseudo-true set and derives the misspecification robust inference method.
Section \ref{sec:computation} discusses computational aspects of the method.
Section \ref{sec:empirical} provides an empirical illustration revisiting the analysis in \citet{kli:tam16} and Section \ref{sec:monte_carlo} Monte Carlo evidence on the size and power of the test. 
Section \ref{sec:conclusions} concludes.
Appendix \ref{sec:app} provides proofs of our main results. 
The Online Appendix includes auxiliary Lemmas and additional examples.

\section{Notation and Motivating Example}
\label{sec:examples}
Let $\ey\in \cY\subseteq\R^{d_Y}$, $\ex\in \cX\subseteq\R^{d_X}$ and $\eu\in \cU\subseteq\R^{d_U}$ denote, respectively, observable endogenous and exogenous variables, and unobservable variables, with realizations $y,x,u$. 
Let $\Ps_0\in\cP(\cY\times\cX)$ denote the distribution of $(\ey,\ex)$.\footnote{For a space $\cS$ with Borel $\sigma$-algebra $\Sigma_\cS$, $\cP(\cS)$ denotes the set of all Borel probability measures on $(\cS,\Sigma_\cS)$.}
Assume the conditional law $\Ps_0(\cdot|x)$ is absolutely continuous with respect to a $\sigma$-finite measure $\mu$ on $\cY$. 
Let $\ps{0,y|x}$ be the Radon-Nikodym derivative of $\Ps_0(\cdot|x)$ with respect to $\mu$ and $\ps{0}\equiv\{\ps{0,y|x},x\in\cX\}$.
Throughout, for some $\underline{c}>0$ and all $x\in\cX,y\in\mathcal{Y}$, let $\ps{0,y|x}(y|x)\ge\underline{c}$.
We consider a framework where, based on prior knowledge, a researcher is willing to maintain some parametric structural restrictions on the joint behavior of $(\ey,\ex,\eu)$, but is agnostic about other features of the model or whether they continue to hold after a policy intervention.
We denote $\theta\in\Theta\subset\R^{d_\theta}$ the parameters characterizing the model. We let the structure associate with each $(u,x,\theta)$ a set of predicted outcomes through a closed-valued measurable correspondence $\G:\cU\times\cX\times\Theta\mapsto\cY$.\footnote{Given a probability space $(\Omega,\salg,\PP)$ and $\cC$ the family of closed sets in $\R^d$, a correspondence $\G:\Omega\mapsto\cC$ is measurable if, for every compact set $K$ in $\R^d$, $\G^{-1}(K) =\{ \omega \in \Omega :\G(\omega)\cap K\neq \emptyset \} \in \salg$.} 
This nests as special case the textbook model with singleton predictions with $\ey=\g(\eu|\ex;\theta) ~\text{a.s.}$ for $\g:\cU\times\cX\times\Theta\mapsto\cY$ a measurable function.
We assume the family of distributions for the latent variables $\eu$ is known up to finite dimensional parameter vector that is part of $\theta$, and, omitting specific notation for subvectors of $\theta$, we denote $\{\f:\theta\in\Theta\}$ the family of distributions for $\eu$ and assume that $\f$ is independent of $\ex$.\footnote{This can easily be relaxed if the researcher is willing to specify the conditional distribution of $\eu|\ex$.}
While assuming a parametric distribution for the entire vector $\eu$ simplifies our presentation, in some applications the researcher might be unwilling to take a stand on the distribution of some of its components.  
Our method is applicable in these cases too, though it requires a parametric distribution for a subvector of $\eu$, as we show in Online Appendix Example~\ref{example:panel} for the case of panel dynamic discrete choice models, where the researcher may lack prior knowledge of the distribution of the initial condition \citep[e.g.,][]{hon:tam06}.

The next example, used throughout the paper to illustrate results, clarifies notation.
More examples are provided in the Online Appendix and in \citet{MolinariHOE}.\smallskip
\begin{example}[Static entry game]
\label{example:CT}
  Consider a two player entry game as in \citet{tam03}, with each player $i=1,2$ choosing to enter ($\ey_i=1$) or stay out of the market ($\ey_i=0$).
  Let $(\ex_1,\ex_2)$ and $(\eu_1,\eu_2)\sim\f$ be, respectively, observable and unobservable payoff shifters and player's payoffs be
$
    \pi_j=\ey_j(\ex_j\beta_j+\delta_j\ey_{(3-j)}+\eu_j),  j=1,2,
$ 
with $\delta_1\le 0,\delta_2\le 0$ the interaction effects and $(\beta_1,\beta_2,\delta_1,\delta_2)$ part of $\theta$. 
  Let each player enter the market if and only if $\pi_j\geq 0$. 
  Given $\theta\in\Theta$ and $x\in\cX$, the model has multiple pure strategy Nash equilibria (PSNE), depicted in Figure~\ref{fig:G_entry} as a function of $(u_1,u_2)$.
In our notation, the set of PSNE is the measurable correspondence $\G(\cdot|x;\theta)$ \citep[Proposition 3.1]{ber:mol:mol11}, with:
\begingroup
\allowdisplaybreaks
\small{
  \begin{align}
  			\G(\eu|x;\theta)&=\{(0,0)\}~~\text{if}~~\eu\in S_{\{(0,0)\}|x;\theta} \equiv \{u:u_j<-x_j\beta_j,~j=1,2\},\label{eq:S00}\\
			\G(\eu|x;\theta)&=\{(1,1)\}~~\text{if}~~\eu\in S_{\{(1,1)\}|x;\theta} \equiv \{u:u_j\ge-x_j\beta_j-\delta_j, ~j=1,2\},\label{eq:S11}\\
			\G(\eu|x;\theta)&=\{(1,0)\}~~\text{if}~~\eu\in S_{\{(1,0)\}|x;\theta} \equiv \{u: u_1\ge -x_1\beta_1 -\delta_1, u_2<-x_2\beta_2-\delta_2)\}\notag\\
			&\hspace{4.1cm}\bigcup\{u: -x_1\beta_1\le u_1< -x_1\beta_1-\delta_1, u_2<-x_2\beta_2\},\label{eq:S10}\\
			\G(\eu|x;\theta)&=\{(0,1)\}~~\text{if}~~\eu\in S_{\{(0,1)\}|x;\theta} \equiv \{u:u_1< -x_1\beta_1, u_2\ge -x_2\beta_2\} \notag\\
			&\hspace{3.3cm}\bigcup \{u:-x_1\beta_1\le u_1< -x_1\beta_1-\delta_1, u_2\ge -x_2\beta_2-\delta_2\},\label{eq:S01}\\
			\G(\eu|x;\theta)&=\{(1,0),(0,1)\}~~\text{if}~~~\eu\in M_{x;\theta}\equiv \{u: -x_j\beta_j\le u_j<-x_j\beta_j-\delta_j,~j=1,2\}.\label{eq:M}
  \end{align}}
  \endgroup
  \begin{figure}
\begin{center}
\begin{tikzpicture}[thick,scale=0.525, every node/.style={transform shape}]
		\draw[->,dashed,gray!30] (-4,0) -- (13,0) node[below] {$\color{gray}u_1$};
		\draw[->,dashed,gray!30] (2,-2) -- (2,4) node[left] {$\color{gray}u_2$};

		\draw[white] (2,2.5) -- (2,3.7);
		\draw[white] (8,0) -- (12.7,0);

		\draw (2,3.25) node {$\G(\eu|x;\theta)=\{(0,1)\}$};

		\draw (8.5,-1.5) node {$\G(\eu|x;\theta)=\{(1,0)\}$};

		\draw (5,1.25) node {$\G(\eu|x;\theta)=\{(1,0),(0,1)\}$};

		\draw[dashed,gray!30] (13,2.5) -- (2,2.5) node[left] {$\color{gray}-x_2\beta_2-\delta_2$};
		\draw[dashed,gray!30] (8,4) -- (8,0) node[below] {$\color{gray}-x_1\beta_1-\delta_1$};
		\draw[dashed,gray!30] (11,2.5) -- (2,2.5) node[left] {$\color{gray}-x_2\beta_2-\delta_2$};

		\draw[gray!30] (2,0) -- (2,2.5);
		\draw[gray!30] (8,0) -- (8,2.5);
		\draw[gray!30] (2,0) -- (8,0);
		\draw[gray!30] (2,2.5) -- (8,2.5);

		\draw (10.5,3.25) node {$\G(\eu|x;\theta)=\{(1,1)\}$};

		\draw (-1,-1.5) node {$\G(\eu|x;\theta)=\{(0,0)\}$};

		\draw[gray!30] (2,0.2) node[left] {$\color{gray}-x_2\beta_2$};
		\draw[gray!30] (2.2,0) node[below] {$\color{gray}-x_1\beta_1$};		
		\end{tikzpicture}
\caption{\footnotesize{Stylized depiction of $\G(\cdot|x;\theta)$ in Example \ref{example:CT} with $\delta_1< 0,\delta_2< 0$. }}
		\label{fig:G_entry}
		\end{center}
\end{figure}
If one assumes $\delta_1\cdot\delta_2=0$ (a ``principal assumption'' in the econometrics literature on simultaneous equation models with dummy endogeneous variables predating \citeauthor{tam03}, \citeyear{tam03}), the region $M_{x;\theta}$ occurs with probability zero, and $\G(\eu|x;\theta)$ reduces to a measurable function. Doing so, however, removes the simultaneity in player's actions that many applications aim to capture. 
One may consider completing the model through a selection mechanism that picks an outcome in the region of multiplicity, $M_{x;\theta}$.
Doing so is also often undesirable, as various plausible selection mechanisms may lead to notably different conclusions about the nature of firms' competition \citep[e.g.][Table VII]{ber92}, and which firms enters when the market can sustain only one profitable entrant ($\eu\in M_{x;\theta}$) may be impacted by counterfactual interventions.
The specification of $\f$ can be made flexible through the use of mixtures, although our approach requires the mixture to be finite.
\hfill $\square$
\end{example}

\section{Information-Based Inference Robust to Misspecification}
\label{sec:inference}
\subsection{The Set of Model-Implied Density Functions}
\label{subsec:density_set}
One can view $\G(\eu|x;\theta)$ as the collection of its \emph{measurable selections} \citep[][Def.~2.1]{mol:mol18}, i.e., all random vectors $\tilde{\ey}$ such that $\tilde\ey\in\G(\eu|x;\theta)$ a.s. 
Each selection $\tilde\ey$ is a model predicted outcome.
In order to obtain a set-valued analog of a likelihood model, one needs to be able to characterize the distribution of each of these predicted outcomes.
To do so in a computationally feasible manner, we denote by $\cC$ the collection of closed subsets of $\cY$ and define the law of $\G(\eu|x;\theta)$ induced by the model's structure:
\begin{align}
    \contf(A|x)&\equiv\int_\cU \one(\G(u|x;\theta)\subseteq A)d\f(u),~~~\forall A\in\cC.\label{eq:containment}
\end{align}
The \emph{containment functional} $\contf(A|x)$ uniquely determines the distribution of $\G(\eu|x;\theta)$ when it is evaluated at all $A\in\cC$ \citep[p.20]{mol:mol18}.

Given $\theta\in\Theta$, $x\in\cX$,  and $\contf(\cdot|x)$, by \citet[Theorem 2.1]{art83} it is possible to characterize all distributions of measurable selections of $\G(\eu|x;\theta)$ as the set
\begin{align}
\core(\contf(\cdot|x))\equiv\left\{Q\in\cM(\SY,\cX):Q(A|x)\ge\contf(A|x),~\forall A\in\cC\right\},
\label{eq:core}
\end{align}
where $\cM(\SY,\cX)$ is the collection of laws of random variables supported on $\cY$ conditional on $\ex$.
The characterization in \eqref{eq:core} is sharp, in the sense that, up to an ordered coupling \citep[][Chapter 2]{mol:mol18}, given $\tilde\ey\sim Q(\cdot|x)$, 
$\tilde\ey\in\G(\eu|x;\theta)$ a.s. if and only if $Q(A|x)\ge\contf(A|x),$ for all $A\in\cC$, $x$-a.s.
\smallskip

\noindent\textbf{Example 1} (Continued).
Let $\ey_{x;\theta}(u)$ be the unique element of $\G(u|x;\theta)$ if $u\notin M_{x;\theta}$ and (arbitrary) $\ey_{x;\theta}(u)\equiv(0,0)$ if $u\in M_{x;\theta}$. 
\citet[Example 2.6]{mol:mol18} show that all measurable selections of $\G(\cdot|x;\theta)$ in Example \ref{example:CT} can be represented as
\begin{align}
\ey(\eu,\er)=\ey_{x;\theta}(\eu)&\one(\eu\notin M_{x;\theta})+(\er\times(0,1)+(1-\er)\times(1,0))\one(\eu\in M_{x;\theta}),\label{eq:exampleCT_selection_mech}
\end{align}
for a random variable $\er\in\{0,1\}$ with any distribution in $\cP(\{0,1\})$ and unrestricted dependence on $\eu$ given $\ex$.
Each of these distributions is a \emph{selection mechanism} \citep[e.g.,][]{cil:tam09} that assigns to $(1,0)$ and $(0,1)$ the probability that each is played given $\ex=x$ and $\eu\in M_{x;\theta}$.
\citet[Lemma 2.1]{ber:mol:mol11} show that the distribution of each selection in \eqref{eq:exampleCT_selection_mech} belongs to $\core(\contf(\cdot|x))$, and that only those distributions do.
The containment functional of $\G(\cdot|x;\theta)$ satisfies: $\contf(\{(0,0)\}|x)=\f(S_{\{(0,0)\}|x;\theta})$, $\contf(\{(0,1),(1,0)\}|x)=1-\f(S_{\{(0,0)\}|x;\theta})-\f(S_{\{(1,1)\}|x;\theta})$, and similarly for all $A\subseteq\cY=\{(0,0),(1,0),(0,1),(1,1)\}$, where for a given set $B\subset\cU$, $\f(B)=\int_\cU \one(u\in B)d\f(u)$ and the sets $S_{\{y\}|x;\theta},y\in\cY$ are defined in \eqref{eq:S00}-\eqref{eq:S01}.\hfill$\square$
\smallskip

Assume that there are $\sigma$-finite measures $\mu$ on $(\cY,\SY)$ and $\xi$ on $(\cX,\SX)$, a product measure $\zeta\equiv\mu\times\xi$ on $(\cY\times\cX,\SY\times\SX)$, and for all $\theta\in\Theta$, $x\in\cX$, and $Q\in\core(\contf(\cdot|x))$, $Q\ll\mu$ \citep[this requirement is typically unrestrictive; see, e.g.,][p.~2]{White1982}.
Let the set of conditional densities associated with $\core(\contf(\cdot|x))$ be
\vspace{-.3cm}
\begin{align}
\fq{\theta,x}&\equiv\{\qs{y|x}:\qs{y|x}=dQ(\cdot|x)/d\mu,~Q\in\core(\contf(\cdot|x)) \},\label{eq:core_density_x}\\
\fq{\theta}&\equiv\left\{\fq{\theta,x},~x\in\cX \right\}.\label{eq:core_density}
\end{align}

\noindent\textbf{Example 1} (Continued).
Given $\theta\in\Theta$ and denoting $\Delta$ the unit simplex in $\R^4$, the set of all model predicted probability mass functions corresponding to selections of $\G(\cdot|x;\theta)$ is
\begin{multline}
	\fq{\theta}=\Big\{\qs{y|x}\in\Delta:~\qs{y|x}((0,0)|x)=\f(S_{\{(0,0)\}|x;\theta});~~~~~~\qs{y|x}((1,1)|x)=\f(S_{\{(1,1)\}|x;\theta});\\
	  \f(S_{\{(1,0)\}|x;\theta})\le \qs{y|x}((1,0)|x)\leq \f(S_{\{(1,0)\}|x;\theta})+\f(M_{x;\theta}),~~~~x\in\cX\Big\},\label{eq:frak_q_CT}
\end{multline}
with $S_{\{(0,0)\}|x;\theta}$, $S_{\{(1,1)\}|x;\theta}$,  $S_{\{(1,0)\}|x;\theta}$, $M_{x;\theta}$ defined in \eqref{eq:S00}, \eqref{eq:S11}, \eqref{eq:S10}, \eqref{eq:M}.
\hfill $\square$

\subsection{Correct Specification, Misspecification, and Pseudo-True Set}\label{subsec:pseudo_true}
Define a model $\fQ\equiv\left\{\fq{\theta}:\theta\in\Theta\right\}$ as the collection of sets $\fq{\theta}$ across $\theta\in\Theta$.
We propose a generalization of the standard definition of correct specification for models with singleton predictions \citep[e.g.,][Def.~2.5]{White1996} to models with set-valued predictions.

\begin{definition}[Correctly Specified Model \& Misspecified Model]\label{def:correct_misspec}
A model is correctly specified if $\ps{0}\in\fq{\theta}$ for some $\fq{\theta}\in\fQ\equiv\left\{\fq{\vartheta}:\vartheta\in\Theta\right\}$, and misspecified otherwise.
\end{definition}

\begin{remark}
In models that yield a singleton prediction $\ey=\g(\eu|\ex;\theta)$ a.s.,~with $\g:\cU\times\cX\times\Theta\mapsto\cY$ a measurable function, there is a unique implied law for $\g|\ex=x$:
$\Qs{\theta}(A|x)=\int_\cU \one(\g(u|x;\theta)\in A)d\f(u),~\forall A\in\cC
$, with associated conditional density function $\qs{\theta,y|x}=d\Qs{\theta}(\cdot|x)/d\mu$ (compare with \eqref{eq:containment} and \eqref{eq:core_density_x}).
The model is defined as the collection of (singleton) $\qs{\theta,y|x}$ across $\theta\in\Theta$ and $x\in\cX$, $\cQ=\left\{[\qs{\theta,y|x},~x\in\cX]:~\theta\in\Theta\right\}$.
The model is correctly specified if $\ps{0}=\qs{\theta}$ for some $\qs{\theta}\in\cQ$, and misspecified otherwise.
\end{remark}
Given density functions $f$ and $f'$, it has been common to measure their similarity through the \emph{Kullback-Leibler Information Criterion (KLIC)} \citep{White1996}.
Here we extend the standard definition of KLIC to measure divergence from $f$ of a set of density functions $\mathfrak{f}$. 
\begin{definition}[KLIC for set of density functions] 
Let $(\Omega,\salg,\zeta)$ be a measure space.
Let $f:\Omega\mapsto\R_+$ be a measurable function satisfying $\int f d\zeta<\infty$ and $\int_S f\ln f d\zeta<\infty$ where $S=\{\omega\in\Omega:f(\omega)>0\}$.
Let $\mathfrak{f}$ denote a \emph{set} of measurable functions $f':\Omega\mapsto\R_+$ satisfying $\int_S f\ln f' d\zeta<\infty$ and $\dKL{f}{f'}\equiv\int_{S}f\ln\tfrac{f}{f'}d\zeta$.
The Kullback-Leibler divergence measure from $f$ of a set $\mathfrak{f}$ is $\dKL{f}{\mathfrak{f}}\equiv\inf_{f'\in\mathfrak{f}}\dKL{f}{f'}.$
\end{definition}
It follows from \citet[Theorem 2.3]{White1996} that when $\inf_{f'\in\mathfrak{f}}\int_S(f-f')d\zeta\ge 0$, $\dKL{f}{\mathfrak{f}}=0$ if $f\in\mathfrak{f}$, and $\dKL{f}{\mathfrak{f}}>0$ otherwise.
As our approach is based on measuring divergence between conditional density functions, denoted $f(y|x)$ and $f'(y|x)$, in $\dKL{f}{\mathfrak{f}}$ we replace the unconditional KLIC with the conditional one, $\dKL{f}{f'}\equiv\int_{\cY\times\cX}f(y,x)\ln\tfrac{f(y|x)}{f'(y|x)}d\zeta(y,x)$. 
Following \cite{White1982}'s treatment of point-identified models, we let the \emph{pseudo true set}, $\Theta^*(\ps{0})$, be the set of minimizers of the researcher's ignorance about the true structure.
But here one is also ignorant about which selection from the model predicted set is closest to the data.
Hence, minimization occurs with respect to both $\vartheta\in\Theta$, as in the textbook case, and $\qs{}\in\fq{\vartheta}$.
If the model is correctly specified, $\Theta^*(\ps{0})$ equals the sharp identification region, as in point-identified models where the pseudo-true value matches the data-generating one.
\begin{definition}
The \emph{pseudo-true identified set} is given by
\begin{align}
\Theta^*(\ps{0})&\equiv
\Big\{\theta\in\Theta:\dKL{\ps{0}}{\fq{\theta}}=\inf_{\vartheta\in\Theta}\dKL{\ps{0}}{\fq{\vartheta}} \Big\}.\label{eq:def_pseudo_true_set}
\end{align}
\end{definition}
To understand the effect of minimizing KLIC with respect to $\qs{}\in\fq{\vartheta}$, note that
\begin{align}
\dKL{\ps{0}}{\fq{\vartheta}}&=\inf_{\qs{}\in\fq{\vartheta}}\int_{\cY\times\cX}\ps{0}(y,x)\ln\tfrac{\ps{0,y|x}(y|x)}{\qs{y|x}(y|x)}d\zeta(y,x)\notag\\
&=\int_\cX\ps{0,x}(x)\inf_{\qs{y|x}\in\fq{\vartheta,x}}\int_\cY\ps{0,y|x}(y|x)\ln\tfrac{\ps{0,y|x}(y|x)}{\qs{y|x}(y|x)}d\mu(y)d\xi(x),\label{eq:profiling}
\end{align}
with $\fq{\vartheta,x}$ defined in \eqref{eq:core_density}.
No unknown selection mechanisms are used in \eqref{eq:profiling} to formalize all ways in which a measurable selection could be picked from $\G(\cdot|x;\theta)$ and all associated likelihoods obtained. 
Rather, \eqref{eq:profiling} relies on a convex program, with strictly convex objective and convex constraints (a finite dimensional convex program with linear constraints if $\cY$ is finite).
It delivers the density function in $\fq{\vartheta,x}$ closest with respect to KLIC to $\ps{0,y|x}$,
\begin{align} 
\qs{\vartheta,y|x}^*=\arg\inf_{\qs{y|x}\in\fq{\vartheta,x}}\int_\cY\ps{0,y|x}(y|x)\ln\tfrac{\ps{0,y|x}(y|x)}{\qs{y|x}(y|x)}d\mu(y), \label{eq:q_star}
\end{align}
which can be calculated analytically or numerically.\footnote{Existence and uniqueness of $\qs{\vartheta,y|x}^*$ is guaranteed under mild conditions (see Online Appendix Lemma~\ref{lem:licq}).}
It can be interpreted as a profiled (quasi)-likelihood where a convex optimization program profiles out the selection mechanism, which is left completely unspecified and may arbitrarily depend on $(\ex,\eu,\vartheta)$. The support of $\ex$ is also unrestricted.
In contrast, related likelihood-based inference methods rely on an infinite-dimensional parameter space to represent the selection mechanism that picks measurable selections from $\G(\cdot|x;\theta)$ (as in \eqref{eq:exampleCT_selection_mech}) and profile it out via non-convex optimization programs with increasing number of (sieve) coefficients \citep[e.g.,][]{Chen:2011aa}; or restrict the class of selection mechanisms by assuming that they do not depend on $\eu$ after conditioning on $\ex$ and that $\ex$ has finite support \citep{Chen_2018}.\footnote{\citet{Chen:2011aa}'s inference method assumes correct model specification. 
\citet[Remark 3]{Chen_2018} suggest that their method may remain valid for some misspecified separable models with discrete covariates.}
Doing so may substantially increase computational burden or narrow the class of models allowed for.
For example, in discrete choice models with unobserved heterogeneity in choice sets \citep{bar:cou:mol:tei21} it would rule out choice set formation based on sequential search or rational inattention (see Online Appendix Example \ref{example:BCMT}).

Putting together \eqref{eq:profiling} and \eqref{eq:q_star} we obtain
\begin{align*}
	\dKL{\ps{0}}{\fq{\vartheta}}=\int_{\cY\times\cX}\ps{0}(y,x)\ln\tfrac{\ps{0,y|x}(y|x)}{\qs{\vartheta,y|x}^*(y|x)}d\zeta(y,x). 
\end{align*}
Hence, the pseudo-true set $\Theta^*(\ps{0})$ in \eqref{eq:def_pseudo_true_set} is equal to the set of maximizers of $\crit(\vartheta)$, with
\begin{align}
\crit(\vartheta)\equiv\E_{\ps{0}}\left[\crit(\vartheta|\ex)\right]~~\text{and}~~\crit(\vartheta|\ex)\equiv \E[\ln \qs{\vartheta,y|x}^*(\ey|\ex)|\ex].	\label{eq:expected_log_likelihood}
\end{align}

\vspace{-.2cm}
\noindent\textbf{Example 1} (Continued).
Given $\theta\in\Theta,x\in\cX$, and $S_{\{(0,0)\}|x;\theta}$, $S_{\{(1,1)\}|x;\theta}$, $S_{\{(1,0)\}|x;\theta}$, $M_{x;\theta}$ as in \eqref{eq:S00}, \eqref{eq:S11}, \eqref{eq:S10}, \eqref{eq:M}, let

\vspace{-1.1cm}
\begin{align}
\eta_1(\theta;x)&\equiv1-\f(S_{\{(0,0)\}|x;\theta})-\f(S_{\{(1,1)\}|x;\theta}),\label{eq:eta1}\\
\eta_2(\theta;x)&\equiv \f(S_{\{(1,0)\}|x;\theta})+\f(M_{x;\theta}),\label{eq:eta2}\\
\eta_3(\theta;x)&\equiv\f(S_{\{(1,0)\}|x;\theta}).\label{eq:eta3}
\end{align}

\vspace{-.4cm}
\noindent
In words, $\eta_1(\theta;x)$ is the probability allocated by the model to either $(1,0)$ or $(0,1)$ occurring as outcome of the game; $\eta_2(\theta;x)$ [$\eta_3(\theta;x)$] is the upper [lower] bound implied by the model on the probability that $(1,0)$ is the outcome of the game.
Define the parameter sets:
\begin{align}
	\Theta_1(x,\ps{0})&\equiv\Big\{\theta\in\Theta:\eta_3(\theta;x)\le\tfrac{\ps{0,y|x}((1,0)|x)}{\ps{0,y|x}((1,0)|x)+\ps{0,y|x}((0,1)|x)}\eta_1(\theta;x)\le \eta_2(\theta;x)\Big\}\label{eq:def_Theta1}\\
	\Theta_2(x,\ps{0})&\equiv\Big\{\theta\in\Theta:	\tfrac{\ps{0,y|x}((1,0)|x)}{\ps{0,y|x}((1,0)|x)+\ps{0,y|x}((0,1)|x)}\eta_1(\theta;x)> \eta_2(\theta;x)\Big\}\label{eq:def_Theta2}\\
	\Theta_3(x,\ps{0})&\equiv\Big\{\theta\in\Theta:	\tfrac{\ps{0,y|x}((1,0)|x)}{\ps{0,y|x}((1,0)|x)+\ps{0,y|x}((0,1)|x)}\eta_1(\theta;x)< \eta_3(\theta;x)\Big\}.\label{eq:def_Theta3}
\end{align}

\vspace{-.2cm}
\noindent
Then the profiled likelihood is given by (see Proposition \ref{prop:games_profiled} in the Online Appendix):
\begin{align}
\qs{\theta,y|x}^*((0,0)|x)&=\f(S_{\{(0,0)\}|x;\theta})\label{eq:eg_prof_likelihood1}\\
\qs{\theta,y|x}^*((1,1)|x)&=\f(S_{\{(1,1)\}|x;\theta})\label{eq:eg_prof_likelihood2}\\
\qs{\theta,y|x}^*((0,1)|x)&=\begin{cases}
\tfrac{\ps{0,y|x}((0,1)|x)}{\ps{0,y|x}((1,0)|x)+\ps{0,y|x}((0,1)|x)}\eta_1(\theta;x)&\theta\in\Theta_1(x,\ps{0})	\label{eq:eg_prof_likelihood3}\\
\eta_1(\theta;x)-\eta_2(\theta;x) & \theta\in\Theta_2(x,\ps{0})\\
\eta_1(\theta;x)-\eta_3(\theta;x) & \theta\in\Theta_3(x,\ps{0})
\end{cases}\\
\qs{\theta,y|x}^*((1,0)|x)&=1-\qs{\theta,y|x}^*((0,0)|x)-\qs{\theta,y|x}^*((1,1)|x)-\qs{\theta,y|x}^*((0,1)|x).\label{eq:eg_prof_likelihood4}
\end{align}
Intuitively, when $\theta\in\Theta_1(x,\ps{0})$,  $(1,0)$ can be assigned a share of $\eta_1(\theta;x)$ equal to the share of $\ps{0,y|x}(\{(0,1),(1,0)\}|x)$ that $(1,0)$ has in the data. 
When $\theta\in\Theta_2(x,\ps{0})$, that allocation yields a probability for $(1,0)$ larger than the model's upper bound $\eta_2(\theta;x)$, and the KL divergence is minimized setting $\qs{\theta,y|x}^*((1,0)|x)=\eta_2(\theta;x)$.
Similarly for $\theta\in\Theta_3(x,\ps{0})$.
In Section~\ref{subsec:geometry_Theta_star} below we discuss the general topological properties of $\Theta^*$.\hfill $\square$

\subsection{The Score Function}
\label{subsec:score_fn}
Here we characterize the \emph{score function} associated with the singleton-valued likelihood function in \eqref{eq:expected_log_likelihood}.
Let $\cA^{(*e)}(\theta;x)\subseteq \cC$ denote a \emph{smallest core determining class} in the sense of \citet[Theorem 1]{luo:pon:wan25}, i.e., a subset of $\cC$ of minimal cardinality such that $\core(\contf(\cdot|x))=\big\{Q\in\cM(\SY,\cX):Q(A|x)\ge\contf(A|x),\forall A\in\cA^{(*e)}(\theta;x)\big\}$.
\begin{assumption}\label{ass:finite_support}
(a) $\cY$ is finite and $\cX$ is compact. 
(b) $\contf(A|x)$ is continuously differentiable with respect to $\theta$ and $\nabla_\theta\contf(A|x)$ is square integrable for all $A\subset \cY$, $\Ps_0-a.s.$
(c) $\Theta^*(\ps{0})\subset \interior\Theta$.
(d) There exists a constant $c>0$ such that for all $\theta\in\Theta$ and $y\in\cY$, $\qs{\theta,y|x}^*(y|x)>c$, $\Ps_0-a.s.$
(e) There is a collection $\cA_\G(x)\subset 2^\cY$, that does not depend on $\theta$, such that $\supp(\G(\cdot|x;\theta))\equiv \{A\subseteq \cY:F_\theta(G(\eu|x;\theta)=A)>0\}=\cA_\G(x)$ for all $\theta\in\Theta$, $\Ps_0-a.s.$
(f) The set of (in)equalities in $\cA^{(*e)}(\theta;x)$ active at $\qs{\theta,y|x}^*$ are irredundant, in the sense that removing any of them strictly enlarges the active face of $\fq{\theta,x}$ at $\qs{\theta,y|x}^*$.
\end{assumption}
Assumption \ref{ass:finite_support}-(a) restricts attention to models with discrete outcomes and compact support for $\ex$.
Part (b) is easily verified when $\f$ is differentiable in $\theta$.
Part (c) implies that first order conditions hold at all $\theta^*\in\Theta^*(\ps{0})$.
Part (d) bounds $\qs{\theta,y|x}^*$ away from zero.
Assumption~\ref{ass:finite_support}-(e) requires the support of $\G(\cdot|x,\theta)$ not to vary with $\theta\in\Theta$, though it can vary with $x\in\cX$.
It holds in several applications, including discrete choice models with unobserved heterogeneity in choice sets \citep[see our Online Appendix~\ref{app:examples}]{bar:cou:mol:tei21} and dynamic monopoly entry models \citep[][as shown in \citeauthor{luo:pon:wan25}, \citeyear{luo:pon:wan25}, p.17]{ber:com22}.
Moreover, \citet[Remark 3.1]{gu:rus:str25} show that with finite $\cY$, $\Theta$ can be finitely partitioned in $x$-dependent partitions, so that within each partition $\cA_\G(x)$ does not depend on $\theta$.
\citet[Corollary 1.1]{luo:pon:wan25} show that under this condition, $\cA^{(*e)}(\theta;x)=\cA^{(*e)}(x)$ for all $\theta$.
Lemma~\ref{lem:licq} shows that when $\qs{\theta,y|x}^*$ belongs to the relative interior of $\fq{\theta,x}$, parts (a)-(e) of Assumption~\ref{ass:finite_support} imply that the Lagrange multiplier vector for the convex program in \eqref{eq:q_star} is unique and hence $\crit(\theta|x)$ is differentiable.
The lemma also shows that Assumption~\ref{ass:finite_support}-(f), which requires that removing a binding inequality strictly enlarges the face of $\fq{\theta,x}$ that $\qs{\theta,y|x}^*$ belongs to (see Definition~\ref{def:irredundant}), ensures the same conclusion when $\qs{\theta,y|x}^*$ belongs to the relative boundary of $\fq{\theta,x}$.
Conditions that imply uniqueness of the Lagrange multipliers can replace parts (e)(f).
These two conditions are restrictive.
For example, as stated they rule out that $\Theta$ includes both values of $\theta$ at which $\G(\cdot|x;\theta)$ collapses to a function and values at which it is a non-singleton correspondence \citep[whether the latter values of $\theta$ are consistent with the DGP may be testable, see][]{che:kai21}.
In Online Appendix~\ref{app:clarke} we show that we can dispense with these conditions and work with the subdifferential of $\crit(\theta|x)$, but doing so makes our procedure less tractable. 
We verify all conditions for Example \ref{example:CT} below and other examples in Online Appendix \ref{app:examples}.
\begin{theorem}\label{thm:score_characterization}
Under Assumption \ref{ass:finite_support}, (i) $\crit(\theta|x)$ 
is differentiable with respect to $\theta$ on $\mathrm{int}(\Theta)$, $\Ps_0-a.s.$
(ii) There exists a function $s:\Theta\times \cY\times\cX\times\Delta\to \R^{d_\theta}$, with $\Delta$ the unit-simplex in $\R^{|\cY|-1}$, such that $\E[\|\score{\theta}{\ey|\ex;\ps{0,y|x}}\|^2]<\infty$, and
\begin{align}
\tfrac{\partial}{\partial\theta}\crit(\theta|x)&=\E[\score{\theta}{\ey|\ex;\ps{0,y|x}}|\ex=x],\label{eq:conditional_foc}	\\
    \E[\score{\theta}{\ey|\ex;\ps{0,y|x}}]&=0,~ \text{ for all }\theta\in\Theta^*(\ps{0}).\label{eq:zero_expected_score}
\end{align}
\end{theorem}
The score function depends on $\ps{0,y|x}$.
When $\theta\mapsto\crit(\theta|x)$ is concave, $\Theta^*(\ps{0})$ equals the set of $\theta\in\Theta$ for which \eqref{eq:zero_expected_score} holds.
When concavity does not hold, this set includes $\Theta^*(\ps{0})$.
As one of our goals is to avoid spuriously tight confidence sets, we view the benefit of an easy-to-implement method to outweigh the cost of a sometimes wider confidence set which asymptotically uniformly covers the set of $\theta$'s satisfying \eqref{eq:zero_expected_score}.
Our Monte Carlo results in Section \ref{sec:monte_carlo} show that, for the examples analyzed there, our procedure performs well relative to existing methods.
The proof of Theorem \ref{thm:score_characterization} is in Appendix \ref{sec:app} and leverages results in \citet{Gauvin:1990uz} to establish differentiability with respect to $\theta$ of
\begin{align}
\crit(\theta|x) = \E[\ln \qs{\theta,y|x}^*(\ey|\ex)|\ex=x]=\sup_{\qs{y|x}\in\fq{\theta,x}}\sum_{y\in\cY} \ps{0,y|x}(y|x)\ln\qs{y|x}(y|x),\label{eq:conditional_L}
\end{align}
where $\fq{\theta,x}$ is defined in \eqref{eq:core_density_x}.
The proof uses results in \citet{luo:pon:wan25} by which under Assumption \ref{ass:finite_support}-(e), the smallest collection of inequalities among the ones in \eqref{eq:core_density_x} that suffice to sharply characterize $\fq{\theta,x}$ does not depend on $\theta$.
As further discussed in Section~\ref{rem:CDC}, Lemma \ref{lem:licq}-(ii) in the Online Appendix shows that $\qs{\theta,y|x}^*$ and the value of $\crit(\theta|x)$, and hence its differentiability and the score function $\score{\theta}{y|x;\ps{0,y|x}}$, are insensitive to inclusion of additional inequalities from \eqref{eq:core_density_x} in the maximization problem in \eqref{eq:conditional_L}.
Hence, so are the pseudo-true set $\Theta^*(\ps{})$ and our inference procedure.
This is in contrast with much related literature, where moment selection is often required for computational tractability or for the inference procedure, but may have substantial implications both on the population region that the researcher targets \citep{ked:li:mou21} and on the properties of the inference procedure \citep[e.g.,][]{and:shi13,bug:can:shi17,kai:mol:sto19}.
\vspace{.2cm}

\noindent\textbf{Example 1} (Continued).
Assumption \ref{ass:finite_support}-(a) holds.
Part (b) holds as long as $\f(S_{\{y\}|x;\theta})$, $y\in\cY$, and $\f(M_{x;\theta})$ are differentiable with respect to $\theta$ (e.g., for $(\eu_1,\eu_2)$ bivariate normal). 
Part (d) holds by compactness of $\cX$ and $\Theta$, as one can find $c>0$ such that $\eta_j(\theta;x)\ge c$ and $\eta_1(\theta;x)-\eta_j(\theta;x)\ge c,j=2,3$, $\f(S_{\{y\}|x;\theta})>c,~y\in\{(0,0),(1,1)\}$, for all $x\in\cX$ and $\theta\in\Theta$.
Part (e) holds if, e.g., $\delta_1,\delta_2<0$, whence $\cA_\G=\{\{(0,0)\},\{(0,1)\},\{(1,0)\},\{(1,1)\},$ $\{(1,0),(0,1)\}\}$ for all $\theta\in\Theta$. 
Part (f) holds when $\delta_1\cdot\delta_2\neq 0$: at most one inequality binds at $\qs{\theta,y|x}^*$ and is irredundant.
If one allows for $\delta_1\cdot\delta_2=0$, at those values the model is complete, the support of $\G(\cdot|x,\theta)$ changes to $\{\{(0,0)\},\{(0,1)\},\{(1,0)\},\{(1,1)\}\}$, and two inequalities with identical gradient bind at $\qs{\theta,y|x}^*$, violating Assumption \ref{ass:finite_support}-(e)(f); in that case, one can extend our approach using subdifferentials, as we show in Online Appendix~\ref{app:clarke}.
In Proposition \ref{prop:games_profiled} in the Online Appendix we show that:
\begin{align}
\score{\theta}{(0,0)|x;\ps{0,y|x}}&=\tfrac{\nabla_\theta\f(S_{\{(0,0)\}|x;\theta})}{\f(S_{\{(0,0)\}|x;\theta})}\label{eq:eg_score1}\\
\score{\theta}{(1,1)|x;\ps{0,y|x}}&=\tfrac{\nabla_\theta\f(S_{\{(1,1)\}|x;\theta})}{\f(S_{\{(1,1)\}|x;\theta})}\label{eq:eg_score2}\\
\score{\theta}{(0,1)|x;\ps{0,y|x}}&=
\begin{cases}
	\tfrac{\nabla_\theta \eta_1(\theta;x)}{ \eta_1(\theta;x)}&\theta\in\Theta_1(x;\ps{0,y|x})\\
	\tfrac{\nabla_\theta [\eta_1(\theta;x)-\eta_2(\theta;x)]}{\eta_1(\theta;x)- \eta_2(\theta;x)}&\theta\in\Theta_2(x;\ps{0,y|x})\\
	\tfrac{\nabla_\theta [\eta_1(\theta;x)-\eta_3(\theta;x)]}{\eta_1(\theta;x)- \eta_3(\theta;x)}&\theta\in\Theta_3(x;\ps{0,y|x})
\end{cases}\label{eq:eg_score3}\\
\score{\theta}{(1,0)|x;\ps{0,y|x}}&=\begin{cases}
	\tfrac{\nabla_\theta \eta_1(\theta;x)}{ \eta_1(\theta;x)}&\theta\in\Theta_1(x;\ps{0,y|x})\\
	\tfrac{\nabla_\theta \eta_2(\theta;x)}{ \eta_2(\theta;x)}&\theta\in\Theta_2(x;\ps{0,y|x})\\
	\tfrac{\nabla_\theta \eta_3(\theta;x)}{ \eta_3(\theta;x)}&\theta\in\Theta_3(x;\ps{0,y|x})
\end{cases}\label{eq:eg_score4}
\end{align}
In Section \ref{sec:computation} below we discuss how to accurately and rapidly compute the score numerically when analytic representations are not available.

\subsection{Geometry of the Pseudo-True Set}\label{subsec:geometry_Theta_star}
We next discuss the topological properties of the pseudo-true set.
By its definition in \eqref{eq:def_pseudo_true_set}, $\Theta^*$ can be viewed as the $\argmin$-set of an optimization problem indexed by $\ps{}$:
\begin{align}
\Theta^*(\ps{})=\argmin_{\theta\in\Theta}\phi(\theta,\ps{}),~~\text{with}~~\phi(\theta,\ps{})\equiv \inf_{\qs{}\in\fq{\vartheta}}\int_{\cY\times\cX}\ps{}(y,x)\ln\tfrac{\ps{y|x}(y|x)}{\qs{y|x}(y|x)}d\zeta(y,x).\label{eq:Theta_star_param_prob}
\end{align}
\begin{theorem}\label{thm:argmin_stability}
    Suppose $\Theta$ is compact, Assumption~\ref{ass:finite_support} holds, and there exists a neighborhood $V$ of $\ps{0}$ and $M>0$ such that for all $p\in V$, $\sup_{\theta\in\Theta}\E[\|\score{\theta}{\ey|\ex;\ps{y|x}}\|^2]\le M$. 
    Then the mapping $\ps{}\mapsto \Theta^*(\ps{})$ is nonempty, compact-valued, and upper hemicontinuous.
\end{theorem}
We prove this theorem by showing that under its assumptions, $\phi(\cdot,\cdot)$ is jointly continuous in $(\theta,\ps{})$ (Lemma~\ref{lemma:continuity}), which along with compactness of $\Theta$ guarantees applicability of Berge's maximum theorem.
The result establishes non-emptyness of $\Theta^*(\ps{})$ and yields a first step towards characterizing how the geometry of the pseudo-true set changes as the extent of model misspecification changes.
Given an observed density $\ps{0}$, let $\{\ps{\gamma}\}\in\Delta$, $\gamma\in\Gamma\subset\R$, be a net with corresponding pseudo-true values $\{\theta_\gamma\}$ such that: (i) $\theta_\gamma\in \Theta^*(\ps{\gamma})$ for all $\gamma\in\Gamma$, (ii) $\ps{\gamma}\to \ps{0}$, and (iii) $\theta_\gamma\to\theta^*$ for some $\theta^*\in\Theta$.
Then Theorem~\ref{thm:argmin_stability} yields that $\theta^*\in \Theta^*(\ps{0})$: the limits of such nets form elements of the pseudo true set at $\ps{0}$.
This property, which holds under weak conditions, rules out that $\Theta^*(\ps{\gamma})$ is persistently larger than $\Theta^*(\ps{0})$, as $\ps{\gamma}\to \ps{0}$. 
However, it does not preclude the possibility that $\Theta^*(\ps{\gamma})$ shrinks from a set at $\ps{0}$ (e.g., an arc in the entry game example) to a smaller set or a singleton at $\ps{\gamma}$, no matter how close $\ps{\gamma}$ is to $\ps{0}$.

This possibility is instead precluded when $\ps{\gamma}\mapsto\Theta^*(\ps{\gamma})$ is lower hemicontinuous, yielding the second step characterizing how the geometry of $\Theta^*(\ps{\gamma})$ depends on the extent of misspecification.
It is well known that lower hemicontinuity at some $\gamma_0\in\Gamma$ is guaranteed under stronger and more challenging to verify conditions than upper hemicontinuity \citep[][p.155]{Rockafellar_Wets2005aBK}.
Equivalent statements include that for all $\theta\in\Theta$, the mapping $\gamma\mapsto\dist(\theta,\Theta^*(\ps{\gamma}))$ is upper semicontinuous at $\gamma_0$; or that for every $\rho>0$ and $\epsilon>0$, there is a neighborhood $V$ of $\gamma_0$ such that $\Theta^*(\ps{\gamma_0})\cap\rho\mathbb{B}\subset\Theta^*(\ps{\gamma})+\epsilon\mathbb{B}$ for all $\gamma\in V\cap\Gamma$, with $\mathbb{B}$ the unit ball in $\R^{d_\theta}$ \citep[e.g.,][Propositions 5.11 and 5.12]{Rockafellar_Wets2005aBK}.
Sufficient conditions for lower hemicontinuity are given, e.g., in \citet[][Chapter 9]{Rockafellar_Wets2005aBK} and \citet[Theorem 1.5.5]{Aubin:1990aa}.

Next, we provide an entry game example where we can transparently show that the correspondence $\gamma\mapsto\Theta^*(\ps{\gamma})$ is both lower and upper hemicontinuous, guaranteeing that the geometric properties under correct specification of the sharp identification region are preserved as the model becomes misspecified. 
In particular, $\Theta^*(\gamma)$ does not abruptly become a singleton as $\ps{\gamma}$ turns incompatible with the model (see Definition~\ref{def:correct_misspec}).
\vspace{.25cm}

\noindent\textbf{Example 1} (Specialized).\label{example:entry_uniform}
Let $(U_1,U_2)\sim \mathrm{Uniform}[0,1]^2$, $\cY=\{(1,1),(0,1),(1,0)\}$. For $\ex^*\in\{0,1\}$ and $\Ps_0(\ex^*=1)=1/2$, set $\pi_j=\ey_j(\delta_{j,0}(\ex^*)\ey_{3-j}+U_j)$, $\delta_{j,0}(\ex^*)=-0.5 + \theta_{j,0} \ex^*$, $\Theta=[-0.45,0]^2$, and $\theta_{j,0}<0$ for $j=1,2$.
Let $y=(0,1)$ be always selected in the region of multiplicity.
Misspecification occurs because $\ex^*$ is replaced by a binary proxy $\ex$ with $\Ps_0(\ex^*=1|\ex=x)=\kappa(x,\gamma)\equiv(1-\gamma)x+\gamma(1-x)$.  
Expressing the observed density of $\ey|\ex$, indexed by $\gamma$ and denoted $\ps{\gamma}(\ey|\ex)$, as a mixture of the density $\ps{0}(\ey|\ex^*=x)$, for $x=0,1$, we obtain $\ps{\gamma}(y|\ex=x)=\ps{0}(y|\ex^*=1)\kappa(x,\gamma)+0.25(1-\kappa(x,\gamma))$, for $y=(1,1),(1,0)$, and $\ps{\gamma}((0,1)|\ex=x)=1- \ps{\gamma}((1,1)|\ex=x)-\ps{\gamma}((1,0)|\ex=x)$.
At $\gamma=0$, $\ps{0}(\ey|\ex=x)=\ps{0}(\ey|\ex^*=x)$ for all $x\in\{0,1\}$ and the model is correctly specified.
Regardless of the value of $\gamma$, denoting $\delta_{j,\theta}(x)=-0.5 + \theta_j x$ for some $\theta_j\in\Theta$, we have
\vspace{-.5cm}
\begin{multline*}
    \fq{\theta,x}=\big\{q\in\Delta:\qs{}((1,1)|x)=(1+\delta_{1,\theta}(x))(1+\delta_{2,\theta}(x)),\\
    -\delta_{2,\theta}(x)(1+\delta_{1,\theta}(x))\le\qs{}((1,0)|x)\le-\delta_{2,\theta}(x)\big\}.
\end{multline*}

\vspace{-.2cm}
\noindent
Ignoring possible misspecification, one would state $\theta$'s sharp identification region as $\Theta_I(\ps{\gamma})=\{\theta\in\Theta:\ps{\gamma}(y|x)\in\fq{\theta,y|x},x=0,1\}=\big\{(\vartheta-[0.5~0.5]^\top)\in\Theta:\vartheta_1\vartheta_2=\ps{\gamma}((1,1)|1)$; $\vartheta_1\le 1-\ps{\gamma}((0,1)|1);~\vartheta_2\le 1-\ps{\gamma}((1,0)|1);~\ps{\gamma}((1,1)|0)=0.25;~\ps{\gamma}((1,0)|0)\in[0.25,0.5]\big\}$.
For $\gamma=0$, $\Theta_I(\ps{0})$ is an arc defined by $\ps{0}(y|1)\in\fq{\theta,y|1}$, with $\ps{0}(y|0)\in\fq{\theta,y|0}$ restricting $\ps{0}(y|0)$ but not $\theta$.
But for any $\gamma>0$, no matter how close to $0$, $\Theta_I(\ps{\gamma})=\emptyset$, as one can verify that $\ps{\gamma}((1,1)|0)<(1+\delta_{1,\theta}(0))(1+\delta_{2,\theta}(0))=0.25$ for all $\theta\in\Theta$.

In contrast, our pseudo-true set $\Theta^*(\ps{\gamma})$ remains robust to misspecification and preserves its arc shape. 
After plugging the uniform distribution for $\f$ into \eqref{eq:eta1}-\eqref{eq:def_Theta3} to define $\eta_j(\theta;x)$ and the sets $\Theta_j(x,\ps{\gamma})$, and into \eqref{eq:eg_prof_likelihood2}-\eqref{eq:eg_prof_likelihood4} to obtain $\qs{\theta,y|x}^*$ and from that $\score{\theta}{y|x;\ps{\gamma,y|x}}$ for $y\in\cY$, one can verify that
\vspace{-.3cm}
\begin{multline}
\hspace{-.25cm}\Theta^*(\ps{\gamma})=\Big\{\theta\in\Theta:~\tfrac{\ps{\gamma}((1,1)|\ex=1)}{\ps{\gamma}((0,1)|\ex=1)+\ps{\gamma}((1,0)|\ex=1)}=\tfrac{(1+\delta_{1,\theta}(1))(1+\delta_{2,\theta}(1))}{1-(1+\delta_{1,\theta}(1))(1+\delta_{2,\theta}(1))},\\
\tfrac{-\delta_{2,\theta}(1)(1+\delta_{1,\theta}(1))}{1-(1+\delta_{1,\theta}(1))(1+\delta_{2,\theta}(1))}\le\tfrac{\ps{\gamma}((1,0)|\ex=1)}{\ps{\gamma}((0,1)|\ex=1)+\ps{\gamma}((1,0)|\ex=1)}\le \tfrac{-\delta_{2,\theta}(1)}{1-(1+\delta_{1,\theta}(1))(1+\delta_{2,\theta}(1))}\Big\}.\label{eq:ex1:Theta_star}
\end{multline}
Even when $\gamma>0$, the equality restriction in \eqref{eq:ex1:Theta_star} is analogous to $\ps{0}((1,1)|1)=(1+\delta_{1,0}(1))(1+\delta_{2,0}(1))$. Along with the other two inequalities, it yields a region with the same shape as $\Theta_I(\ps{\gamma})$ for $\gamma=0$. 
In Proposition~\ref{prop:verify_ex_LHC} in the Online Appendix we derive \eqref{eq:ex1:Theta_star} and prove that the mapping $\gamma\mapsto\Theta^*(\ps{\gamma})$ is both upper and lower hemicontinuous.
\hfill$\square$

\subsection{Asymptotic Distribution of the Average Score and Rao's Test Statistic}
\label{subsec:test_stat}
Any $\theta^*\in\Theta^*(\ps{0})$ satisfies the population first order condition in \eqref{eq:zero_expected_score}.
By the sample analog principle, we propose to estimate $\E[\score{\theta^*}{\ey|\ex;\ps{0,y|x}}]$ through $\bar{s}_{\theta^*,n}(\hat p_{n,y|x})$, with
\begin{align}
    \bar{s}_{\theta,n}(p)\equiv\tfrac{1}{n}\sum_{i=1}^n \score{\theta}{\ey_i|\ex_i;p},\quad\theta\in\Theta,~p\in\Delta\label{eq:average_score}
\end{align}
and $\hat p_{n,y|x}$ a nonparametric estimator of $\ps{0,y|x}$, e.g., a cell mean estimator when $\ex$ has a discrete distribution or a sieve estimator when $\ex$ has a continuous distribution.
\label{ref:core:refult}Our core result consists of showing that $\sqrt{n}\bar{s}_{\theta^*}(\hat p_{n,y|x})$ has an asymptotically normal distribution, which is insensitive to estimation of $\ps{0,y|x}$.
We do so leveraging the literature on semiparametric estimation, in particular \citet[Proposition 2]{Newey:1994aa}, to prove that $\E[\score{\theta^*}{\ey,\ex;\ps{y|x}}]$ has an orthogonality property with respect to $\ps{0,y|x}$.
Here we provide high-level conditions under which our results attain.
In Online Appendix \ref{appendix:entry_game} we verify these conditions for the entry game example, both with discrete and continuous covariates.
To state these conditions, for any $\theta\in\Theta$, let $m_\theta(x;\ps{y|x})\equiv \E[\score{\theta}{\ey|\ex;\ps{y|x}}|\ex=x]$. 
Let $\cH$ be a parameter space to which $\ps{0,y|x}$ belongs,  with $\dim(\cH)=d_Y\times d_X<\infty$ if $\ex$ is finitely supported, and $\cH$ infinite dimensional otherwise.
Let $\|p-p'\|_{\cH}$ be a pseudo-metric on $\cH$ (e.g., the sup-norm $\|p\|_{\cH}=\sup_{x\in\cX}\sup_{y\in\cY}|p(y|x)|$). 
For any $p\in\cH$ and $\theta\in\Theta$, let $\mathbb G_{n,\theta}(p)\equiv\sqrt{n}(\bar{s}_{\theta,n}(p)-\E[\score{\theta}{\ey_i|\ex_i;p}])$.
\begin{assumption}
\label{as:newey_modified1_new}
For each $\theta^*\in\Theta^*(\ps{0})$, the pathwise derivative
\begin{align*}
D(\theta^*,\ps{0,y|x})[\ps{y|x}-\ps{0,y|x}]=\lim_{\tau\to 0}\tfrac{\E[m_{\theta^*}(\ex,\ps{0,y|x}+\tau(\ps{y|x}-\ps{0,y|x}))-m_{\theta^*}(\ex,\ps{0,y|x})]}{\tau}
\end{align*}
exists in all directions $(\ps{y|x}-\ps{0,y|x})\in \cH$.
For any $\delta_n=o(1)$ and all $\|\ps{y|x}-\ps{0,y|x}\|_{\cH}\le\delta_n$,
  \begin{align*} 
  \big\|\E[m_{\theta^*}(\ex;\ps{y|x})]-\E[m_{\theta^*}(\ex,\ps{0,y|x})]-D(\theta^*,\ps{0,y|x})[\ps{y|x}-\ps{0,y|x}]\big\|\le c\|\ps{y|x}-\ps{0,y|x}\|^2_{\mathcal H}.
  \end{align*}
\end{assumption}

\begin{assumption}
\label{as:newey_modified2}
\begin{enumerate}[label=(\roman*)]
\item \label{ass:random_sample} The data is a random sample $(\ey_i,\ex_i)_{i=1}^n$ drawn from $P_0$.
\item \label{ass:consistent_est} $\hat p_{n,y|x}\in\mathcal H$ with probability approaching  1 and $\|\hat p_{n,y|x}-\ps{0,y|x}\|_{\mathcal H}=o_P(n^{-1/4})$.
\item \label{ass:normality} For each $\theta^*\in\Theta^*(\ps{0})$, $\mathbb G_{n,\theta^*}(\ps{0,y|x})\stackrel{d}{\to}N(0,\Sigma_{\theta^*})$, with $\Sigma_{\theta^*}\equiv\E[\score{\theta^*}{\ey_i|\ex_i;\ps{0}}\score{\theta^*}{\ey_i|\ex_i;\ps{0}}^\top]$ the population variance-covariance matrix of the score function.
\item \label{ass:equic} For each $\theta^*\in\Theta^*(\ps{0})$ and all sequences of positive numbers $\{\delta_n\}$ with $\delta_n=o(1)$,
\begin{align*}
		\sup_{\|\ps{y|x}-\ps{0,y|x}\|_{\mathcal H}\le \delta_n}\big\|\mathbb G_{n,\theta^*}(\ps{y|x})-\mathbb G_{n,\theta^*}(\ps{0,y|x})\big\|=o_P(1).
	\end{align*}	
\end{enumerate}
\end{assumption}
Assumptions \ref{as:newey_modified1_new} and \ref{as:newey_modified2}, in their use of $\|\cdot\|_{\mathcal H}$, refer to the same norm.
In Assumption \ref{as:newey_modified1_new}, we follow \citet[Conditions 2.3 and 2.6]{Chen2003} and impose a smoothness condition with respect to $\ps{y|x}$ on $\E[\score{\theta^*}{\ey,\ex;\ps{y|x}}]$. 
Assumption \ref{as:newey_modified2} \ref{ass:random_sample} is a standard random sampling condition (see \citet{eps:kai:seo16} for a discussion of inference under different assumptions).
Assumption \ref{as:newey_modified2} \ref{ass:consistent_est} requires that the estimation error of the nuisance parameter $\ps{0,y|x}$ vanishes fast enough. 
Assumption \ref{as:newey_modified2} \ref{ass:normality} follows from the central limit theorem.
Assumption \ref{as:newey_modified2} \ref{ass:equic} is a stochastic equicontinuity condition with well known primitive conditions \citep{Vaart:1996wk}. 
In Propositions \ref{lemma:pathwise_game}-\ref{lemma:donsker_game} in the Online Appendix, we verify all these assumptions in the two players entry game Example \ref{example:CT}, both with discrete and continuous covariates.
Under these assumptions, we obtain:
\vspace{-.25cm}
\begin{theorem}\label{lem:scores_limit}
Suppose Assumptions \ref{ass:finite_support}, \ref{as:newey_modified1_new}, and \ref{as:newey_modified2} hold. Then, for each $\theta^*\in\Theta^*(\ps{0})$,
\vspace{-.35cm}
\begin{align}
	\tfrac{1}{\sqrt n}\sum_{i=1}^n \score{\theta^*}{\ey_i|\ex_i;\hat p_{n,y|x}}=\tfrac{1}{\sqrt n}\sum_{i=1}^n\score{\theta^*}{\ey_i|\ex_i;\ps{0,y|x}}+o_p(1)\stackrel{d}{\to}N(0,\Sigma_{\theta^*}).\label{eq:scores_limit}
\end{align}
\end{theorem}
Armed with the result in Theorem \ref{lem:scores_limit}, we propose to use a Rao's score statistic  to test at prespecified asymptotic level $\alpha\in(0,1)$ hypotheses of the form
\begin{align}
\mathbb{H}_0: \theta\in\Theta^*(\ps{0})~~\text{against}~~\mathbb{H}_A: \theta\notin\Theta^*(\ps{0}),\label{eq:null_hp}
\end{align}
\vspace{-.35cm}
and to obtain confidence sets by test inversion.
The Rao-type test statistic takes the form
\begin{align}
	T_n(\theta)\equiv\left(\tfrac{1}{\sqrt n}\sum_{i=1}^n \score{\theta}{\ey_i|\ex_i;\hat p_{n,y|x}}\Big)^\top\tilde\Sigma_{n,\theta}^{-1}\Big(\tfrac{1}{\sqrt n}\sum_{i=1}^n \score{\theta}{\ey_i|\ex_i;\hat p_{n,y|x}}\right).\label{eq:score_stat}
\end{align}
Given the score function, the test statistic in \eqref{eq:score_stat} is easy to compute even when the covariates have a continuous distribution.
The weight matrix $\tilde\Sigma_{n,\theta}$ is a consistent estimator of $\Sigma_{\theta}$ when $\Sigma_{\theta}$ is not nearly singular, and assures an asymptotically valid test when $\Sigma_{\theta}$ is nearly singular, as shown below.
Let $\hat\Sigma_{n,\theta}=\tfrac{1}{n}\sum_{i=1}^n (\score{\theta}{\ey_i|\ex_i;\hat p_{n,y|x}}-\bar{s}_{\theta}(\hat p_{n,y|x}))(\score{\theta}{\ey_i|\ex_i;\hat p_{n,y|x}}-\bar{s}_{\theta}(\hat p_{n,y|x}))^\top$ be the sample analog estimator of $\Sigma_\theta$; $\hat\Xi_{n,\theta}=\hat\Psi_{n,\theta}^{-1/2}\hat\Sigma_{n,\theta}\hat\Psi_{n,\theta}^{-1/2}$ the correlation matrix associated with $\hat\Sigma_{n,\theta}$; $\hat\Psi_{n,\theta}=\rm{diag}(\hat\Sigma_{n,\theta})$; and  $\varepsilon>0$ a regularization constant.
We recommend using the estimator proposed in \citet[p. 2808]{and:bar12}, which introduces an adjustment insuring that the weight matrix is always nonsingular and equivariant to scale changes in the score function:
\begin{align}
 \tilde\Sigma_{n,\theta}=\hat\Sigma_{n,\theta}+\max\{\varepsilon-\det(\hat\Xi_{n,\theta}),0\}\hat\Psi_{n,\theta},~~~\theta\in\Theta.\label{eq:tilde_Sigma}
 \end{align}
\begin{corollary}\label{cor:Tn_limit}
Let Assumptions \ref{ass:finite_support}, \ref{as:newey_modified1_new}, and \ref{as:newey_modified2} hold.
Then, under $\mathbb{H}_0$ in \eqref{eq:null_hp}, for any $\theta^*\in\Theta^*(\ps{0})$ such that $\min_{j=1,\dots,d_\theta}\{\rm{diag}(\Sigma_{\theta^*})\}_j>0$ and 
\begin{align}
\hat\Sigma_{n,\theta^*}\stackrel{p}{\to}\Sigma_{\theta^*},\label{ass:consistent_cov}
\end{align}
(a) If $\Sigma_{\theta^*}$ is nonsingular,
$T_n(\theta^*)\stackrel{d}{\to}\chi^2_{d_\theta}$; (b)
Both for singular and nonsingular $\Sigma_{\theta^*}$,
$\lim\sup_{n\to\infty}\Ps(T_n(\theta^*)>c_{d_\theta,\alpha})\le \alpha$,
with $c_{d_\theta,\alpha}$ the $1-\alpha$ quantile of the $\chi^2_{d_\theta}$ distribution.
\end{corollary}
\textbf{Example 1} (Continued).\label{rem:lower_rank}
    The score function for Example 1 on p.~\pageref{example:entry_uniform} implies that $\Theta^*(\ps{\gamma})$ is characterized by a single equality restriction, and hence the rank of $\Sigma_{\theta^*}$ in this case is lower than $d_\theta$ and the critical value in Corollary~\ref{cor:Tn_limit} is asymptotically conservative.\hfill$\square$\smallskip
    
Corollary \ref{cor:Tn_limit} requires, in \eqref{ass:consistent_cov}, that the population covariance matrix can be consistently estimated.
In the semiparametric literature with point identification, this is a standard requirement;\footnote{\label{fn:covariance}In the semiparametric literature, \eqref{ass:consistent_cov} is imposed for $\hat\Sigma_{n,\hat{\theta}_n}$, with $\hat{\theta}_n$ a consistent estimator of a singleton $\theta^*$; in our use of it, $\hat\Sigma_{n,\theta^*},\Sigma_{\theta^*}$ are both evaluated at the same $\theta^*$, but the requirement applies to all $\theta^*\in\Theta^*(\ps{0})$.} in the moment inequalities literature with partial identification, it is also common to assume that the covariance matrix of the moment functions can be consistently estimated for all $\theta$ in the identified set.
The result in Corollary \ref{cor:Tn_limit} is valuable because it implies that no simulations are needed to compute the quantiles of the limiting distribution, and that the critical values used to test the hypothesis in \eqref{eq:null_hp} and to construct the confidence set via test inversion are constant across candidates $\theta\in\Theta$.
This is in contrast with much of the related literature, where the asymptotic distribution of the test statistic is nonpivotal and the critical values need to be recomputed for each $\theta$.\footnote{Nonpivotal asymptotic distributions appear, e.g., in \citet{and:kwo22,and:shi13,kai:mol:sto19}, and \citet{bug:can:shi17}, while \citet{Chen_2018}'s test statistic converges to $\chi^2$ distribution.}
One can construct a confidence region that covers each point in $\Theta^*(\ps{0})$ with asymptotic probability $1-\alpha$ as
\begin{align}
CS_n=\big\{\theta\in\Theta:T_n(\theta)\le c_{d_\theta,\alpha}\big\}.\label{eq:CSn}
\end{align}
In practice, $CS_n$ is computed by specifying a grid of values $\Theta_n$ through which to explore the parameter space, and letting $CS_n=\{\theta\in\Theta_n:T_n(\theta)\le c_{d_\theta,\alpha}\}$. 

We next show that $CS_n$ is an asymptotically \emph{uniformly valid} confidence set.
We posit that $\Ps_0$, the distribution of the observed data, belongs to a class of distributions denoted by $\cP$, where the conditional law $P(\cdot|x)$ for each $P\in\cP$ is absolutely continuous with respect to $\mu$ on $\cY$.
We let $\ps{y|x}$ denote the Radon-Nykodim derivative of $P(\cdot|x)$. 
We write stochastic order relations that hold uniformly over $P \in \cP$ using the notations $o_{\mathcal P}$ and $O_{\mathcal P}$.
\begin{theorem}\label{thm:uniform_coverage}
For constants $c>0$ and all $P\in\cP$, let Assumptions \ref{ass:finite_support}, \ref{as:newey_modified1_new}, and \ref{as:newey_modified2} hold, with the following conditions replacing the corresponding ones in the original assumptions:\footnote{The constants $c$ may differ across appearances but do not depend on $P$; $\mathbb{N}$ denotes the natural numbers; and $\mathbb{B}_c(\theta)$ denotes a ball of radius $c$ centered at $\theta$.}
\begin{enumerate}
\item[1(c)']\label{ass:interior:uniform} $\Theta^*(\ps{})\subset \interior \Theta^{-c}\equiv\{\theta \in \Theta:\mathbb{B}_c(\theta)\subset\Theta\}$. 
\item[1(e)']\label{ass:support:unif} $\cA_\G(x)=\supp(\G(\cdot|x;\theta))\equiv \{A\subseteq \cY:F_\theta(G(\eu|x;\theta)=A)>c\}$ for all $\theta\in\Theta$, $P-a.s.$
\item[2']\label{ass:linearization:uniform} The constant $c$ is the same for all $P\in\cP$. 
\item[3(ii)']\label{ass:conv:phat:uniform} For all $\epsilon>0$ there exists $N\in\mathbb{N}$, with $\epsilon$ and $N$ not dependent on $P\in\cP$, such that $P(\hat p_{n,y|x}\in\mathcal H)\ge1-\epsilon,~\forall n\ge N$, and $\|\hat p_{n,y|x}-\ps{0,y|x}\|_{\mathcal H}=o_\cP(n^{-1/4})$.
\item[3(iv)'] For all sequences of positive numbers $\{\delta_n\}$ with $\delta_n=o(1)$,
\begin{align*}
		\sup_{\theta^*\in\Theta^*(\ps{0})}\sup_{\|\ps{y|x}-\ps{0,y|x}\|_{\mathcal H}\le \delta_n}\Big\|\mathbb G_{n,\theta^*}(\ps{y|x})-\mathbb G_{n,\theta^*}(\ps{0,y|x})\Big\|=o_\cP(1).
	\end{align*}	
\end{enumerate}
Suppose that for all $\Ps\in\cP$ and $\theta^*\in\Theta^*(\ps{0})$, $\min_{j=1,\dots,d_\theta}\{\rm{diag}(\Sigma_{\theta^*})\}_j>0$, and $\Vert\hat\Sigma_{n,\theta^*}-\Sigma_{\theta^*}\Vert=o_\cP(1)$.
Then, for $CS_n$ in \eqref{eq:CSn}, we have 
$$\liminf_{n\to\infty}\inf_{P\in\cP}\inf_{\theta^*\in\Theta^*(\ps{})}P(\theta^*\in CS_n)\ge 1-\alpha.$$
\end{theorem}
Under the assumptions of Theorem \ref{thm:uniform_coverage}, Corollary \ref{cor:Tn_limit} also applies uniformly over $P\in\cP$.
\begin{remark}\label{rem:conf_set}
    Our confidence set can be interpreted as the union over $\theta^*\in\Theta^*$ of ellipsoids centered at $\theta^*$ whose volume is determined by $\tilde{\Sigma}_{n,\theta^*}$ only, and not by the extent of misspecification.
    The lever through which the latter might impact the volume of $CS_n$, is through $\Theta^*$'s volume, which may or may not be impacted by misspecification (see Section~\ref{subsec:geometry_Theta_star}).
    Relative to $\Theta^*$'s volume, $CS_n$'s volume depends only on the sampling variability of the score statistic, and as $\Theta^*$ is always non-empty, $CS_n$ is always non-empty as well.
\end{remark}

\subsection{Role of Sharp Identifying Restrictions}\label{rem:CDC}
Our analysis is able to exploit, through \eqref{eq:core}, all identifying restrictions associated with the economic model.
Doing so is in part motivated by our desire to avoid the risk of potentially discordant conclusions driven by model misspecification and different choices of subsets of moment inequalities on which to base inference \citep[an issue well articulated in][]{ked:li:mou21}.
And in part because it allows us to provide a general proof, under Assumption~\ref{ass:finite_support}-(e)(f), of Lemma~\ref{lem:licq}.
This lemma establishes the existence of a representation of $\crit(\theta|x)$, for all $\theta\in\Theta$, characterized by $\qs{\theta,y|x}^*$ and unique Lagrange multipliers.
Importantly, one does not need to calculate that representation, because $\crit(\theta|x)$ coincides with it pointwise (see Lemma~\ref{lem:licq}-(ii)).
Yet, $\crit(\theta|x)$ in \eqref{eq:expected_log_likelihood} inherits its differentiability properties.\footnote{The presence of redundant (in)equality constraints in the original formulation of $L(\theta|x)$ does not affect differentiability, as there exists an equivalent formulation without such redundancies with unique Lagrange multipliers.}

In practice, when the number of inequalities in \eqref{eq:core} is relatively small, users of our method can derive the score in closed form.
When it is moderate (a few hundred), the numerical score can be computed directly based on \eqref{eq:core} or many redundant inequalities can be quickly eliminated without resorting to specialized methods.\footnote{See the \href{https://github.com/hkaido0718/IncompleteDiscreteChoice}{Python library} created by Hiroaki Kaido to carry out this task for discrete choice models, and the computational simplifications in \citet{bon:kum20} for a certain class of multi-players entry models.}
When the number of inequalities in \eqref{eq:core} is substantially larger, one can use Algorithm 3 in \citet{luo:pon:wan25} to eliminate all redundant ones.
This algorithm remains practically feasible, with reported computational times of seconds for a \textrm{Julia} implementation that removes all redundant inequalities in entry game examples where \eqref{eq:core} includes $10^{14}$ inequalities, and dynamic discrete choice examples with $10^{154}$ inequalities \citep[][Tables 1 and 3]{luo:pon:wan25}.

Yet, when the number of inequalities is prohibitive, one may need to carry out inference based on a subset of inequalities.\footnote{Inference methods that rely on discretization of $\cX$ and aim to use all information in the model may face an additional computational bottleneck, as those methods have a final number of inequalities given at least by the number of nonredundant inequalities in \eqref{eq:core} multiplied by the cardinality of the discretization of $\cX$.}
Then, our procedure remains valid provided the conclusions of Lemma~\ref{lem:licq} continue to hold.
Hence, among the criteria guiding the selection of inequalities, one might include ensuring this is the case.
We note that a rich literature studies how to remove redundant constraints in linear programs.
For a given chosen subset of inequalities in \eqref{eq:core}, which by construction defines a set linear in $q$ that includes $\fq{\theta}$, as long as the identity of the non-redundant inequalities does not change with $\theta$ and those binding at $\qs{\theta,y|x}^*$ are irredundant in the sense of Definition~\ref{def:irredundant}, the conclusions of Lemma~\ref{lem:licq} continue to hold (for the new linear program with enlarged constrained set) and our method remains valid.
In this case, one obtains a different pseudo-true set than the one yielded by sharp restrictions, with weakly lower minimal KL divergence, and our confidence set then covers the elements of this different pseudo-true set with a prespecified asymptotic probability. 
Without the conclusions of Lemma~\ref{lem:licq}, existence of the score is no longer guaranteed, rendering a score-based approach inapplicable. 
A likelihood-based approach may nonetheless be possible, but is beyond the scope of this paper.
\vspace{.2cm}

\noindent\textbf{Example 1} (Non-sharp). Consider the two player entry game on p.~\pageref{example:CT} with $\fq{\theta}$ in \eqref{eq:frak_q_CT} and $\delta_1<0,\delta_2<0$. Instead of using sharp identifying restrictions, replace $\fq{\theta}$ with $\big\{\qs{y|x}\in\Delta:~\qs{y|x}(y|x)\ge\f(S_{\{y\}|x;\theta}),~y\in\cY,x\in\cX\big\}$. Then $\crit(\theta|x)$ in \eqref{eq:expected_log_likelihood} can be represented using a set of constraints of the form $A\qs{y|x}\ge b(\theta)$, with $A$ a $4\times4$ matrix with three rows equal to standard basis vectors and one a vector of $1$. Hence, Lemma~\ref{lem:licq} holds.
\hfill $\square$

\section{Computation of the Score Function}
\label{sec:computation}
Sometimes it is possible to obtain a closed-form expression for $\score{\theta}{y|x;\ps{0,y|x}}$ as gradient of $\ln \qs{\theta,y|x}^*$ with respect to $\theta$, as in Example \ref{example:CT} (p.~\pageref{eq:eg_score1}).
If $\qs{\theta,y|x}^*$ does not have a closed form expression, one needs to compute the score numerically. 
Here we describe how to do so, adapting the method in \citet{for23}.
We omit the dependence of $\score{\theta}{y|x}$ on $\ps{0,y|x}$ or its estimator.
We presume that one can compute $\qs{\theta,y|x}^*$ relatively easily (e.g., using \texttt{cvxpy}).

Consider a smoothed version $f_\varsigma$ of  $f(\theta)\equiv\ln \qs{\theta,y|x}^*$, defined by the convolution:
\begin{align*}
f_\varsigma(\theta)=\int f(\theta+\varsigma z)\phi(z)dz,	
\end{align*}
where $\phi$ is a smooth kernel decaying to 0 in the tails, such as the Gaussian density function. 
The derivative of $f_\varsigma$ exists. 
If it admits integration by parts, one has:
{\small
\begin{algorithm}[t]
\small
        \caption{Construct $CS_n$}\label{alg:procedure_tuning}
        \textbf{Data:} $W^n=(\ey_i,\ex_i),i=1,\dots,n$
        \begin{algorithmic} 
        \REQUIRE $\varepsilon> 0$, $K_n>0$, $\varsigma>0$ {\small[if no closed form for score function, else $\varsigma=0$], $\Theta_n$}
        \STATE $\hat p_{n,y|x}$ $\gets$ \textsc{CCP}($W^n$;$K_n$) \small[$K_n$ the tuning parameter for $\hat p_{n,y|x}$]
        \STATE $\bar{s}_{\theta;\varsigma}(\ey_i|\ex_i;\hat p_{n,y|x})$  $\gets$ \textsc{Score}($W^n$;$\hat p_{n,y|x}$;$\theta$;$\varsigma$)~~{\small[\eqref{eq:approx_average_score} if no closed form, else \eqref{eq:average_score}]}
        \STATE $c_{d_\theta,\alpha}$$\gets$ $\inf\{x:P(\chi^2_{d_\theta}\ge x)\le \alpha\}$
        \STATE $CS_n$ $\gets$ $\emptyset$
        \FOR{$\theta\in \Theta_n$}
        \STATE $\tilde \Sigma_{n,\theta}$ $\gets$ \textsc{RegCovMat}($W^n$,$\theta$;$\varepsilon$)~~\small[\eqref{eq:tilde_Sigma}]
        \STATE $T_n(\theta)$ $\gets$ $\left(\tfrac{1}{\sqrt n}\sum_{i=1}^n \score{\theta}{\ey_i|\ex_i;\hat p_{n,y|x}}\Big)^\top\tilde\Sigma_{n,\theta}^{-1}\Big(\tfrac{1}{\sqrt n}\sum_{i=1}^n \score{\theta}{\ey_i|\ex_i;\hat p_{n,y|x}}\right)$
        \IF{$T_n(\theta)\le c_{d_\theta,\alpha}$} \STATE Add $\theta$ to $CS_n$ \ENDIF
        \ENDFOR
        \RETURN $CS_n$
        \end{algorithmic}
    \end{algorithm}}
\begin{align*}
	\tfrac{\partial}{\partial\theta}f_\varsigma(\theta)&=-\tfrac{1}{\varsigma}\int f(\theta+\varsigma z)\tfrac{\partial}{\partial z} \phi(z)dz=-\tfrac{1}{\varsigma}\E\Big[f(\theta+\varsigma Z)\tfrac{\tfrac{\partial}{\partial z} \phi(Z)}{\phi(Z)}\Big],
\end{align*}
with the last expectation taken with respect to $Z\sim N(0, I_d)$. 
One can then approximate the derivative of $f$ by that of $f_\varsigma$.
Letting $Z_r,r=1,\dots,R$ be i.i.d. draws from $N(0,I_d)$, and noting that $\big(\tfrac{\partial}{\partial z} \phi(Z)\big)/\phi(Z)=\nabla \ln \phi(z)=-z$, an unbiased estimator for $\tfrac{\partial}{\partial\theta}f_\varsigma(\theta)$ is
\begin{align}
	\tfrac{1}{\varsigma R}\sum\nolimits_{r=1}^R [f(\theta+\varsigma Z_r)-f(\theta)]Z_r.\label{eq:estimate_grad_ftau}
\end{align}
Replacing $f$ with $f_\varsigma$ introduces a bias proportional to $\varsigma$; when $f(\theta)$ is evaluated with noise, the variance of $[f(\theta+\varsigma Z_r)-f(\theta)]Z_r/\varsigma$ grows with $1/\varsigma^2$. 
In practice one needs to take a stand on this bias-variance trade off.
We analyze how to do so in research-in-progress (available from the authors upon request), where numerical experiments lead us to recommend setting $\varsigma = c(nR)^{-1/4}$ for $c \in [0.03, 0.12]$.
The Monte Carlo approximation in \eqref{eq:estimate_grad_ftau} inflates the variance by a factor of $(1+\tfrac{\bar{c}}{R})$, for some constant $\bar{c}>0$, as in the method of simulated moments.
This factor can easily be incorporated in the estimator of the asymptotic variance of the score.
Letting $f(\theta;\ey_i,\ex_i)=\ln \qs{\theta,y|x}^*(\ey_i,\ex_i)$, one can obtain the estimator in \eqref{eq:estimate_grad_ftau} for each value of $(\ey_i,\ex_i)$.  
The average score can then be approximated by
\begin{align}
	\bar{s}_{\theta;\varsigma}(\hat p_{n,y|x})=\tfrac{1}{n\varsigma R}\sum\nolimits_{i=1}^n\sum\nolimits_{r=1}^R [f(\theta+\varsigma Z_{i,r};\ey_i,\ex_i)-f(\theta;\ey_i,\ex_i)]Z_{i,r}.\label{eq:approx_average_score}
\end{align}
Algorithm~\ref{alg:procedure_tuning} presents pseudo-code with the steps and tuning parameters required to build the confidence set in \eqref{eq:CSn}. When a closed form expression for $\score{\theta}{\ey_i|\ex_i;\hat p_{n,y|x}}$ is available, one can set $\varsigma=0$ and plug the observed values $\{\ey_i,\ex_i\}_{i=1}^n$ in \eqref{eq:average_score}; when it is not available, one needs to choose $\varsigma>0$ and plug $\{\ey_i,\ex_i\}_{i=1}^n$ in \eqref{eq:approx_average_score}.

\section{Empirical Illustration}
\label{sec:empirical}
We illustrate the usefulness of our method by applying it to answer the question addressed in \citet[Section 8]{kli:tam16}: ``what explains the decision of an airline to provide service between two airports.'' 
\citeauthor{kli:tam16} analyze data for the second quarter of the year 2010, documenting the entry decisions of two types of airline companies: Low Cost Carriers ($LCC$) versus Other Airlines ($OA$).\footnote{We use their data, downloading it from Quantitative Economics' \href{https://www.econometricsociety.org/publications/quantitative-economics/2016/07/01/Bayesian-inference-in-a-class-of-partially-identified-models}{online repository}.}  
They define a market as a trip between two airports, irrespective of intermediate stops.  
They record the entry decision $\ey_{j,m}$ of player $j\in\{LCC,OA\}$ in market $m$ as a $1$ if a firm of type $j$ serves market $m$, and $0$ otherwise.  
They posit that player $j$'s decision to serve a market depends not only on their opponent's entry decision, but also on observable payoff shifters $\ex_{j,m}$ and unobservable payoff shifters $\eu_{j,m}$.
The observable payoff shifters $\ex_{j,m}$ include the constant and two continuously distributed variables, $\ex_{j,m}^{pres}$ and $\ex_m^{size}$. 
The first variable is firm-and-market-specific: it measures the market presence of firms of type $j$ in market $m$ \citep[see][p. 356 for its exact definition]{kli:tam16}. 
Market presence of the LCC airline, $\ex_{LCC,m}^{pres}$ (respectively, $\ex_{OA,m}^{pres}$), is excluded from the payoff of firm $OA$ (respectively, $LCC$).
The second variable, market size, enters the payoff of firms of both types; it measures population size at the two endpoints of the trip and is market-specific. 
The unobservables $\eu_{j,m}$, $j\in\{LCC,OA\}$, are assumed to have a bivariate normal distribution with $\E(\eu_{j,m})=0$, $Var(\eu_{j,m})=1$, $Corr(\eu_{LCC,m},\eu_{OA,m})=r$, and to be i.i.d. across $m$.\footnote{We assume $r\in [-0.9,0.9]$ and estimate it as part of the vector $\theta$. We ensure that the strategic interaction parameters $\delta_{LCC}$ and $\delta_{OA}$ are less than a constant $c<0$ and that $\qs{\theta^*,y|x}>c$ for another constant $c>0$.}\newline
\indent Both \citet{kli:tam16} and we assume that players enter the market if doing so yields non-negative payoffs.
However, we posit different payoff functions.
They posit:
\begin{multline}
    \tilde\pi_{j,m}=\ey_{j,m}(\tilde\beta_j^0+\tilde\beta_j^{pres}\one(\ex_{j,m}^{pres}\ge Med(\ex_j^{pres}))\\
    +\tilde\beta_j^{size}\one(\ex_m^{size}\ge Med(\ex^{size}))+\tilde\delta_j \ey_{-j,m}+\eu_{j,m}).\label{eq:payoff_KT_CCT}
\end{multline}
In words, \citeauthor{kli:tam16} transform each of market size and of the two market presence variables into binary variables, based on whether each of these variables realizes above or below their respective median.
Doing so yields a finite number of unconditional moment inequalities, which they need for their inference procedure, at the cost of foregoing the information provided in the variation in $\ex_{j,m}$ past whether each variable is above or below its median, and of using an arguably more restrictive payoff function.

Leveraging our new method, we are able to avoid discretizing the continuously distributed covariates, thereby exploiting all identifying power in their variation and allowing $\ex_{j,m}$ to impact payoffs proportionally to their value.
We assume that payoffs take the form:
\begin{align}
    \pi_{j,m}=\ey_{j,m}(\beta_j^0+\beta_j^{pres}\ex_{j,m}^{pres}+\beta_j^{size}\ex_m^{size}+\delta_j \ey_{-j,m}+\eu_{j,m}).\label{eq:payoff_KM}
\end{align}

We study how the decision of an $LCC$ airline to enter the market is affected by whether an $OA$ airline is in the market and by the extent of $LCC$ airlines market presence.
To do so, we define the potential entry decision of an $LCC$ player as $\ey_{LCC}(d)=\one(\ex_{LCC,m}^\top\beta_{LCC}+\delta_{LCC} d+\eu_{LCC,m}\ge 0)$, with $\beta_j=(\beta_j^0,\beta_j^{pres},\beta_j^{size})$, $j\in\{LCC,OA\}$.
This is the entry outcome of an $LCC$ airline when we fix the $OA$’s entry to take value $d\in\{0,1\}$.
Based on our model, the entry probability of the $LCC$ airline is
$P(\ey_{LCC}(d)=1|\ex_{LCC,m})=\Phi(\ex_{LCC,m}^\top\beta_{LCC}+\delta_{LCC} d)$.
We obtain a confidence interval for this parameter for each of $d=0,1$ and for specific values of $\ex_{LCC,m}$.
To ease the reporting of results, we set $\ex_m^{size}$ equal to the median of its distribution and compute the $\tau$-quantile of the distribution of $\ex_{LCC,m}^{pres}$ for $\tau\in\mathcal{T}\equiv\{0.125,0.250,0.375,0.5,0.625,0.750,0.875\}$; we evaluate our parameter of interest for $\ex_{LCC,m}^{pres}$ set equal to each of these values, but our confidence set construction uses all information in $\ex_{j,m}$.
Letting $\theta=(\beta_{LCC},\delta_{LCC},\beta_{OA},\delta_{OA},r)$, we report confidence intervals $CI_n(x,d)=\big[\min_{\theta:~T_n(\theta)\le c_{d_\theta,\alpha}} \Phi(x_{LCC,m}^\top\beta_{LCC}+\delta_{LCC} d)$, $\max_{\theta:~T_n(\theta)\le c_{d_\theta,\alpha}} \Phi(x_{LCC,m}^\top\beta_{LCC}+\delta_{LCC} d)\big]$ for the values of $x$ corresponding to $\tau\in\mathcal{T}$ and for $d\in\{0,1\}$.\footnote{As in the Monte Carlo experiments in Section~\ref{sec:monte_carlo}, we estimate $\ps{0,y|x}$ using a series estimator with $J$-th order (tensor product) B-spline basis functions and set $\varepsilon= 0.05$ in \eqref{eq:tilde_Sigma} to compute $\tilde\Sigma_{n,\theta}$.}
Under the conditions of Theorem~\ref{thm:uniform_coverage}, by standard arguments
$$\liminf_{n\to\infty}\inf_{P\in\cP}\inf_{\theta^*\in\Theta^*(\ps{})}P(\Phi(x_{LCC,m}^\top\beta^*_{LCC}+\delta^*_{LCC} d)\in CI_n)\ge 1-\alpha.$$ 
\begin{figure}
    \centering
    \includegraphics[scale=0.5]{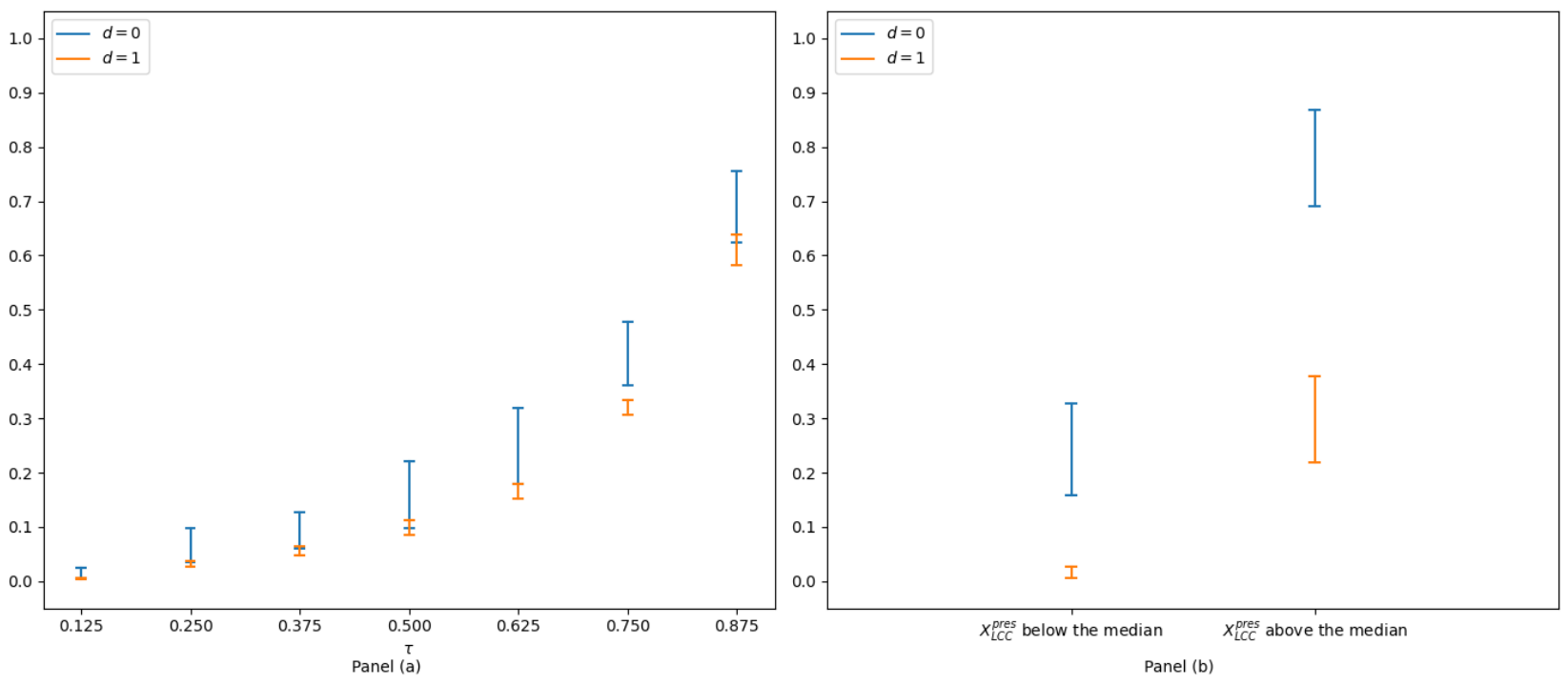}
    \caption{Confidence Intervals for $\Phi(x_{LCC,m}^\top\beta^*_{LCC}+\delta^*_{LCC} d)$ for $d=1$ (orange) and $d=0$ (blue). Panel (a): Rao score test-based inference with $\ex_{LCC,m}^{pres}$ set equal to the $\tau$ quantile of its distribution and $\ex_m^{size}$ set equal to its median. Panel (b): \citet{Chen_2018} projection-based inference with $\one(\ex_{j,m}^{pres}> Med(\ex_j^{pres}))$.} 
    \label{fig:counterfactual}
\end{figure}
Figure \ref{fig:counterfactual}-Panel (a) reports our results, displaying on the horizontal axis the value of $\tau$ and on the vertical axis the candidate value for $\Phi(x_{LCC,m}^\top\beta^*_{LCC}+\delta^*_{LCC} d)$.
The results show potentially sizable heterogeneity in the treatment effects of interest and reject the hypothesis that they are constant across $\tau$. While the confidence intervals for $d=0,1$ are not disjoint for several values of $\tau$, when $OA$ opponents are not in the market (blue segments in Figure \ref{fig:counterfactual} for $d=0$) the entry probability can be much larger than when they are present (orange segments for $d=1$) across all values of $\tau$, with the effect largest for $\tau= 0.750$.

For each fixed value of the entry decision of $OA$, as the market presence of the $LCC$ airlines increases, so does the probability that $LCC$ firms enter a market. 
When OA firms are in the market, the difference is statistically significant for all $\tau$, although the impact of $x_{LCC,m}^{pres}$ on the entry probability is low until market presence reaches its 0.625 quantile, at which point the slope increases rapidly.
This suggests that in order to overcome the presence of $OA$ opponents and enter the market, $LCC$ firms need large market presence.
On the other hand, when $OA$ firms are not in the market, the impact of $x_{LCC,m}^{pres}$ on the entry probability is sizable starting with $\tau=0.375$ (and further increases with $\tau$).

We compare our results to what one would obtain using the likelihood based inference method in \citet{Chen_2018}, which is designed for correctly specified models with discrete covariates.
We note that \citeauthor{Chen_2018} assume the payoff function in \eqref{eq:payoff_KT_CCT}, whereas we use the specification in \eqref{eq:payoff_KM}, and therefore the coefficient estimates on $\ex_{j,m}$ and $\ey_{-j}$ are not directly comparable to each other.
Nonetheless, we believe it to be instructive to compare the counterfactual model-implied entry probabilities across the two approaches.\footnote{We use the replication package provided by \citet{Chen_2018} at Econometrica's \href{https://www.econometricsociety.org/publications/econometrica/browse/2018/11/01/monte-carlo-confidence-sets-identified-sets}{online repository}, where the payoffs are specified as $\tilde\pi_{j,m}=\ey_{j,m}(\tilde\beta_j^0+\tilde\beta_j^{pres}\one(\ex_{j,m}^{pres}> Med(\ex_j^{pres}))+\tilde\beta_j^{size}\one(\ex_m^{size}> Med(\ex^{size}))+\tilde\delta_j \ey_{-j,m}+\eu_{j,m})$ (compare with \ref{eq:payoff_KT_CCT}).
We compute the confidence intervals using the payoffs in their code and obtain a confidence interval on $P(\ey_{LCC}(d)=1|\ex_{LCC,m})$ through the projection method that they propose.} 
Figure \ref{fig:counterfactual}-Panel (b) reports confidence intervals based on \citet{Chen_2018}'s projection method for $d=0,1$ and for $\ex_{LCC,m}^{pres}$ below the median and above the median.
The figure shows that aggregating the value of $\ex_{LCC,m}^{pres}$ at this coarse level hides interesting patterns in the results.
Bundling ``above the median'' and ``below the median'' as single values for the covariates does not allow one to learn the extent of the heterogeneity in the effect of market presence on the probability of entry.
Moreover, using \citet{Chen_2018}'s inference method yields very large treatment effects for the presence of an OA opponent, which instead we do not find.
Of note, our confidence sets are shorter than those based on \citet{Chen_2018} for all values of $\tau$, often substantially, and hence this difference is not driven by a reduction in precision but likely by a different model and our use of all information in the covariates.
While one could use a finer discretization of $(\ex_{LCC,m}^{pres},\ex_{OA,m}^{pres})$ in \eqref{eq:payoff_KT_CCT} combined with \citet{Chen_2018}'s method, doing so would result in a substantially harder computational problem.
Indeed, in \citeauthor{Chen_2018}'s approach the selection probabilities, which are allowed to depend on $\ex$ (but not $\eu$) and have cardinality at least equal to the cardinality of $\cX$, are part of the parameters to be estimated. 
In contrast, the computational complexity of our procedure does not change with the cardinality of $\cX$.

\section{Monte Carlo Experiments}
\label{sec:monte_carlo}
We carry out an empirical Monte Carlo exercise where the data-generating process is calibrated to the data we use for the application in Section~\ref{sec:empirical}; the notation is as in that section and the payoffs are as in \eqref{eq:payoff_KM}.
We normalize each covariate to the unit interval (and continue to denote them $\{\ex_{LCC}^{pres},\ex_{OA}^{pres},\ex^{size}\}$) and assume $(\eu_{LCC,m},\eu_{OA,m})$ is distributed i.i.d.~bivariate standard normal. 
We let $\theta=(\beta_{LCC},\delta_{LCC},\beta_{OA},\delta_{OA})$.

To calibrate a DGP value for the parameter vector $\theta$ from \citet{kli:tam16}'s data, we introduce a parameter $\kappa\in[0,1]$, independent of $(\ex,\eu)$, representing the probability of selecting outcome $Y=(1,0)$ when multiple equilibria are present. Define
\begin{align}
\qs{(\theta,\kappa),y|x}((0,0)|x)&=	[1-\Phi(x_{LCC}^\top\beta_{LCC})][1-\Phi(x_{OA}^\top\beta_{OA})] \label{eq:mc_qtheta1}
\\
\qs{(\theta,\kappa),y|x}((1,1)|x)&=\Phi(x_{LCC}^\top\beta_{LCC}+\delta_{LCC}) \Phi(x_{OA}^\top\beta_{OA}+\delta_{OA}), \\
\qs{(\theta,\kappa),y|x}((0,1)|x)&=\eta_1(\theta;x)-q_{(\theta,\kappa),y|x}((1,0)|x)\\
\qs{(\theta,\kappa),y|x}((1,0)|x)&=\eta_3(\theta;x)+\kappa(\eta_2(\theta;x)-\eta_3(\theta;x))\label{eq:mc_qtheta4}
\end{align}
where the functions $\eta_1(\cdot;x),\eta_2(\cdot;x),\eta_3(\cdot;x)$ are defined in \eqref{eq:eta1}, \eqref{eq:eta2}, \eqref{eq:eta3}.
We report in Table \ref{tab:mle} the value for $(\theta,\kappa)$ estimated by maximizing the likelihood based on \eqref{eq:mc_qtheta1}-\eqref{eq:mc_qtheta4}.
We use this estimate as baseline DGP value in our simulations.\footnote{To ensure that the DGP induces conditional choice probabilities that are bounded away from 0 and 1, we estimate $(\theta,\kappa)$ using observations such that the estimated conditional choice probabilities are in the interval $(\epsilon,1-\epsilon)$ with $\epsilon=1e-3$. This gives us a sample of size 7,017.}
\begin{table}[t]
\centering
\caption{Maximum Likelihood Estimate of $(\theta,\kappa)$}
\begin{tabular}{ccccccccc}
\hline
$\beta^0_{LCC}$ & $\beta^{pres}_{LCC}$ & $\beta^{size}_{LCC}$ & $\delta_{LCC}$ & $\beta^0_{OA}$ & $\beta^{pres}_{OA}$ & $\beta^{size}_{OA}$ & $\delta_{OA}$ & $\kappa$\\
\hline
-0.367 & 2.044 & -0.066 & -0.085 & 0.282 & 1.774 & 0.251 & -0.226 & 0.000 \\
\hline
\end{tabular}
\label{tab:mle}
\end{table}

We consider two designs for our DGPs. 
Design 1 uses only the two player-specific covariates, $\mathbf{X}^{D1}=(X^{pres}_{LCC},X^{pres}_{OA})$, and sets $(\theta_0,\kappa_0) = (\beta^0_{LCC}, \beta^{pres}_{LCC}, \delta_{LCC}, \beta^0_{OA}, \beta^{pres}_{OA}, \delta_{OA}, \kappa)$ to the corresponding maximum likelihood estimates in Table~\ref{tab:mle}. 
Design 2 incorporates the full set of covariates, $\mathbf{X}^{D2}=(X^{pres}_{LCC},X^{pres}_{OA},X^{size})$ and sets $(\theta_0,\kappa_0)$ to the full MLE vector in Table~\ref{tab:mle}.
As we further explain below, we implement Design 1 because of computational difficulties with Design 2 for the comparator inference method of \citet{and:shi13}.

Within each design $k=1,2$, we resample from the original \citet{kli:tam16}'s dataset covariates $\mathbf{X}^{Dk}_m,m=1,\dots,n$.
We evaluate the performance of our inference method \citep[and the][comparator method]{and:shi13} when the model is correctly specified, by generating $(\ey_{LCC,m},\ey_{OA,m})$, $m=1,\dots,n$ through inverse CDF transformation, using as DGP the probability mass functions in \eqref{eq:mc_qtheta1}-\eqref{eq:mc_qtheta4} in which we plug $\mathbf{X}^{Dk}_m$ and, for $(\theta,\kappa)$, the MLE values in Table~\ref{tab:mle} for that design.

To evaluate performance under model misspecification, we simulate data from a two-player entry game where the true payoff for player $j$ in Design $k$ is given by:
\begin{align*}
\pi_j^{Dk}=\ey_j^{Dk}(\ex_{j}^{Dk\top}\beta_j^{Dk}+(\delta_j+\gamma \ex^*)\ey_{-j}^{Dk}+\eu_j),~\ey_j^{Dk}\in\{0,1\},~j\in\{LCC,OA\},
\end{align*}
with $\ex^*$ a binary variable omitted from the model, $\ex_{j}^{D1}=[1~~X^{pres}_{j}]^\top$, and $\ex_j^{D2}=[1~~X^{pres}_{j}~~X^{size}]^\top$. Given $(\ex^{pres}_{LCC},\ex^{pres}_{OA})=(\tilde x_{LCC},\tilde x_{OA})$, $\ex^*=1$ with probability 
$\Phi(\tfrac{\tilde x_{LCC}-\mu_{LCC}}{\sigma_{LCC}}+\tfrac{\tilde x_{OA}-\mu_{OA}}{\sigma_{OA}})$, with $\mu_j$ and $\sigma^2_j$ the mean and variance of $X^{pres}_{j,m}$. 
The value of $\gamma$ determines the extent of misspecification. 
We report results for $\gamma\in\{-.1,-.2,-.3,-.4\}$. 

\begin{figure}
	\begin{center}
		\includegraphics[width=.475\textwidth]{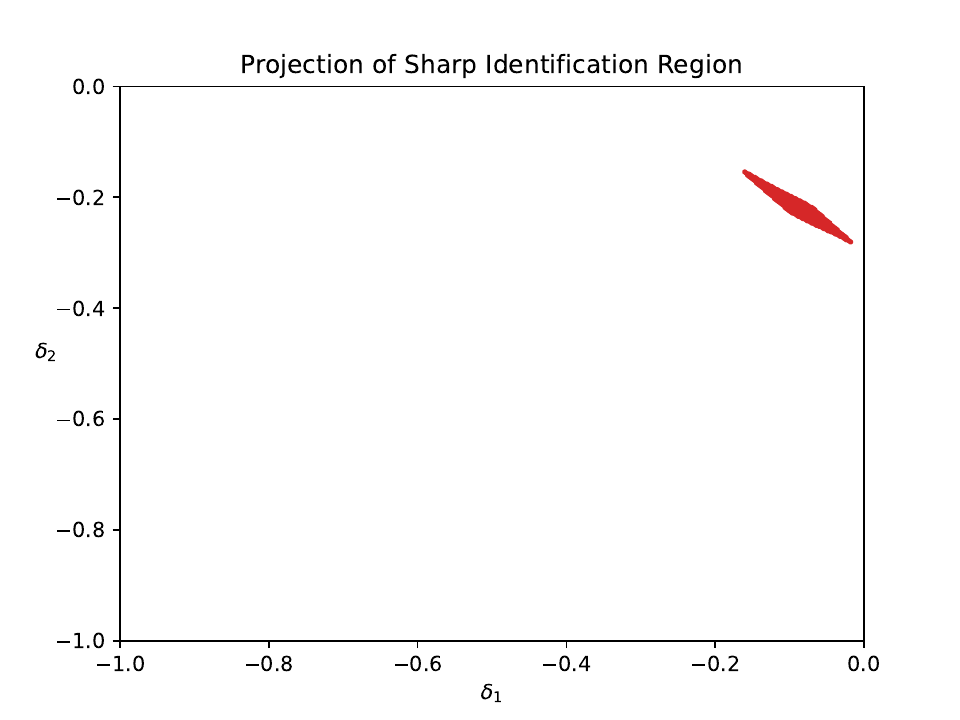}
		\includegraphics[width=.475\textwidth]{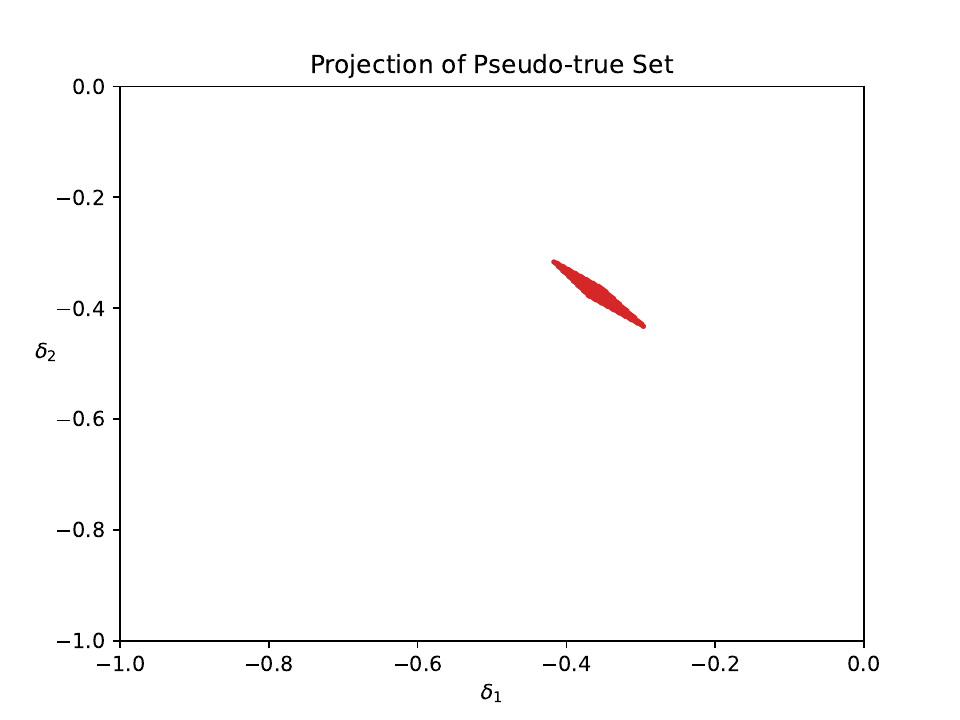}
\caption{Projections of $\Theta^*(\ps{0})$ (left, correctly specified; right, misspecified with $\gamma=-0.4$)}\label{fig:Theta_star_MC}
	\end{center}
\end{figure}

Figure \ref{fig:Theta_star_MC} shows the projections of $\Theta^*(\ps{0})$ onto the space of $(\delta_{LCC},\delta_{OA})$ under Design 2.\footnote{A similar figure for Design 1 is omitted to conserve space but available from the authors upon request.}
When the model is correctly specified (left panel), $\Theta^*(\ps{0})$ coincides with the sharp identification region of $\theta$.
With misspecification ($\gamma=-0.4$), the optimal value of the KL divergence measure in \eqref{eq:def_pseudo_true_set} is strictly positive (see Table \ref{tab:power} for the value of $\dKL{\ps{0}}{\qs{\theta}^*}$), indicating that the sharp identification region is empty.
In contrast, the pseudo-true set $\Theta^*(\ps{0})$ remains nonempty, as shown in the right panel of Figure \ref{fig:Theta_star_MC}, and the shape of $\Theta^*(\ps{0})$ remains similar both in the correctly specified and in the misspecified case.
Nonetheless, under misspecification it shifts to the (lower) left.
This is because when $X^*=1$,  the true DGP allocates a large mass to either $(1,0)$ or $(0,1)$ (whereas with $X^*=0$ that mass is allocated to $(1,1)$).
This is not captured by the model in \eqref{eq:payoff_KM}.
A reduction in the values of $(\delta_1,\delta_2)$ (an increase in absolute value) 
enlarges the region of multiplicity, thereby allowing for a larger mass to be allocated to $(1,0)$ and $(0,1)$ than the model in \eqref{eq:payoff_KM} allows for. 

To implement our test, we estimate $\ps{0,y|x}$ by a series estimator with $J$-th order (tensor-product) B-spline basis functions.\footnote{\label{fn:varepsilon}We compute $\tilde\Sigma$ in \eqref{eq:tilde_Sigma} setting $\varepsilon= 0.05$ as done in \cite{and:shi13}; additional simulations (available from the authors) with $\varepsilon= 0.025$ and $0.1$ indicate that increasing $\varepsilon$ makes our procedure more conservative.} For Design 1, we compare the performance of our procedure to that of \citet{and:shi13}.\footnote{We implement the method in \citet{and:shi13} using their $S_3$ test statistic, with moment functions
		\begin{align*}
			m_{\le}(Z_i,\theta)&=\begin{bmatrix}
				1\{\ey_i=(1,0)\}-\eta_2(x;\theta)\\
				\eta_3(x;\theta)-1\{\ey_i=(1,0)\}
			\end{bmatrix};~
			m_{=}(Z_i,\theta)=
			\begin{bmatrix}
				1\{\ey_i=(0,0)\}-[1-\Phi(x^\top_{LCC}\beta_1)][1-\Phi(x^\top_{OA}\beta_2)]\\
				1\{\ey_i=(1,1)\}-\Phi(x^\top_{LCC}\beta_1+\delta_1) \Phi(x^\top_{OA}\beta_2+\delta_2)\}
			\end{bmatrix},
		\end{align*}
where $m(z,\theta)=(m_{\le}(z,\theta)^\top,m_{=}(z,\theta)^\top)^\top$, and with $\bar m_n(\theta)=\tfrac{1}{n}\sum_{i=1}^n m(Z_i,\theta)$. 
The number of inequalities is two times the number of hyper-cubes used, and similarly for the equalities.} 
We report a power comparison where the parameter value $\theta^*$ lies on the boundary of $\Theta^*(p_0)$. Specifically, we examine the rejection probabilities for local alternatives of the form $\theta_{0,h} = (\beta^*_{LCC},
\delta^*_{LCC}+\tfrac{h}{\sqrt n},\beta^*_{OA},
\delta^*_{OA}+\tfrac{h}{\sqrt n})$, 
with $h > 0$. This corresponds to drifting the strategic interaction effects toward $(0,0)$ while holding the other components fixed.
\citeauthor{and:shi13}'s test transforms the conditional moment inequalities into unconditional ones using as instruments indicator functions of whether each covariate belongs to specified hypercubes. The side length of each hypercube is $1/(2r)$, for $r = 1, \dots, r_{1n}$, where a larger value of $r_{1n}$ corresponds to finer conditioning information. Following \citet[Section 10.4]{and:shi13}, we set $r_{1n} = 3$, and also report results for $r_{1n} = 2$ and $4$.
For Design 2, we report results only for the score-based test, as the moment-based test was computationally infeasible in this setting.

Table \ref{tab:power} reports the results of this exercise for $500$ Monte Carlo repetitions.
Panel (A) documents the size and power of our test as well as the moment inequality-based tests for the case that the model is correctly specified. 
The test of \citet{and:shi13} over-rejects slightly in this correctly specified DGP, while our Rao's score-based test has valid size but under-rejects.
Nonetheless, the power curve of our test quickly dominates that of the moment inequality based test.

\begin{table}[htbp]
\centering
\caption{Rejection Probabilities of Score and Moment Inequality Tests}\label{tab:power}
\resizebox{\textwidth}{!}{	
\begin{threeparttable}
\begin{tabular}{lllllllllllll}
\hline
Design     & Specification info.   & Tests                  & Size   &      \multicolumn{9}{c}{Power (values of $h/\sqrt n$ below)}      \\
\cline{5-13}
           &                 &       &        & 0.013 & 0.025 & 0.038 & 0.050 & 0.063 & 0.076 & 0.088 & 0.101 & 0.113 \\
\hline
\multicolumn{3}{l}{Panel A: Correctly specified $(\gamma=0)$}   & &       &       &       &       &       &       &       &       &       \\
 Design 1: &Two Covariates&          &        &       &       &       &       &       &       &       &       &       \\
           && Score Test             & 0.028 & 0.058 & 0.154 & 0.432 & 0.732 & 0.916 & 0.996 & 1.000 & 1.000 & 1.000 \\
           & & Moment Ineq. Test ($r_{1n}=4$) & 0.070 & 0.106 & 0.222 & 0.430 & 0.722 & 0.896 & 0.976 & 0.998 & 1.000 & 1.000 \\
           & & Moment Ineq. Test ($r_{1n}=3$) & 0.066 & 0.102 & 0.214 & 0.410 & 0.714 & 0.892 & 0.974 & 0.998 & 1.000 & 1.000 \\
		   & & Moment Ineq. Test ($r_{1n}=2$) & 0.058 & 0.088 & 0.202 & 0.394 & 0.690 & 0.886 & 0.974 & 0.998 & 1.000 & 1.000 \\
          &        &       &       &       &       &       &       &       &       &       \\

 Design 2: & Three Covariates  &    &        &       &       &       &       &       &       &       &       &       \\
     & & Score Test      &   0.018 & 0.050 & 0.150 & 0.390 & 0.708 & 0.904 & 0.984 & 1.000 & 1.000 & 1.000\\
           &&                        &        &       &       &       &       &       &       &       &       &       \\
           \hline
\multicolumn{3}{l}{Panel B: Misspecified $(\gamma=-0.1)$} & &       &       &       &       &       &       &       &       &       \\
 Design 1: & Two Covariates &       &        &       &       &       &       &       &       &       &       &       \\
& $\dKL{\ps{0}}{\qs{\theta}^*}=$ 1.99e-04       & Score Test              & 0.046 & 0.086 & 0.192 & 0.476 & 0.746 & 0.932 & 0.992 & 1.000 & 1.000 & 1.000 \\
&IM rej. = 0.002   & Moment Ineq. Test ($r_{1n}=4$) & 0.160 & 0.196 & 0.344 & 0.600 & 0.818 & 0.940 & 0.988 & 0.998 & 1.000 & 1.000 \\
			   &   & Moment Ineq. Test ($r_{1n}=3$) & 0.158 & 0.192 & 0.342 & 0.590 & 0.816 & 0.936 & 0.988 & 0.998 & 1.000 & 1.000 \\
			   &   & Moment Ineq. Test ($r_{1n}=2$) & 0.148 & 0.178 & 0.310 & 0.550 & 0.808 & 0.936 & 0.984 & 0.996 & 1.000 & 1.000 \\
           &                        &        &       &       &       &       &       &       &       &       &       \\
 Design 2: & Three Covariates  &    &        &       &       &       &       &       &       &       &       &       \\
& $\dKL{\ps{0}}{\qs{\theta}^*}=$ 2.00e-04       & Score Test             & 0.050 & 0.092 & 0.242 & 0.494 & 0.762 & 0.938 & 0.986 & 1.000 & 1.000 & 1.000 \\
& IM rej.= 0.002   & &  &  &  &  &  &  &  &  &  &  \\
           &&                        &        &       &       &       &       &       &       &       &       &       \\
           \hline
\multicolumn{3}{l}{Panel C: Misspecified $(\gamma=-0.2)$} & &       &       &       &       &       &       &       &       &       \\
 Design 1: & Two Covariates   &      &        &       &       &       &       &       &       &       &       &       \\
& $\dKL{\ps{0}}{\qs{\theta}^*}= $8.15e-04      & Score Test            & 0.042 & 0.070 & 0.152 & 0.414 & 0.710 & 0.910 & 0.986 & 1.000 & 1.000 & 1.000\\
&IM rej. = 0.076   & Moment Ineq. Test ($r_{1n}=4$) & 0.386 & 0.462 & 0.644 & 0.824 & 0.936 & 0.986 & 0.996 & 1.000 & 1.000 & 1.000 \\
&     			   & Moment Ineq. Test ($r_{1n}=3$) & 0.374 & 0.452 & 0.636 & 0.816 & 0.932 & 0.982 & 0.998 & 1.000 & 1.000 & 1.000 \\
&     			   & Moment Ineq. Test ($r_{1n}=2$) & 0.358 & 0.440 & 0.624 & 0.792 & 0.922 & 0.980 & 0.996 & 1.000 & 1.000 & 1.000 \\
&           &                        &        &       &       &       &       &       &       &       &       &       \\
 Design 2: & Three Covariates  &     &        &       &       &       &       &       &       &       &       &       \\
& $\dKL{\ps{0}}{\qs{\theta}^*}=$ 8.15e-04     & Score Test             & 0.036 & 0.076 & 0.194 & 0.412 & 0.716 & 0.910 & 0.982 & 1.000 & 1.000 & 1.000 \\
& IM rej.= 0.038     &  &  &  &  &  &  &  &  &  &  &  \\
           &&                        &        &       &       &       &       &       &       &       &       &       \\
           \hline
\multicolumn{3}{l}{Panel D: Misspecified $(\gamma=-0.3)$} & &       &       &       &       &       &       &       &       &       \\
 Design 1: &Two Covariates&          &        &       &       &       &       &       &       &       &       &       \\
&  $\dKL{\ps{0}}{\qs{\theta}^*}=$ 1.81e-03         & Score Test             & 0.028 & 0.054 & 0.136 & 0.358 & 0.660 & 0.864 & 0.970 & 0.996 & 1.000 & 1.000\\
&IM rej. =   0.498   & Moment Ineq. Test ($r_{1n}=4$) & 0.762 & 0.830 & 0.896 & 0.968 & 0.986 & 1.000 & 1.000 & 1.000 & 1.000 & 1.000 \\
&  				     & Moment Ineq. Test ($r_{1n}=3$) & 0.746 & 0.812 & 0.890 & 0.966 & 0.984 & 1.000 & 1.000 & 1.000 & 1.000 & 1.000 \\
&  			   	     & Moment Ineq. Test ($r_{1n}=2$) & 0.724 & 0.794 & 0.878 & 0.958 & 0.982 & 0.998 & 1.000 & 1.000 & 1.000 & 1.000 \\
           &&                        &        &       &       &       &       &       &       &       &       &       \\
 Design 2: &Three Covariates&        &        &       &       &       &       &       &       &       &       &       \\
&$\dKL{\ps{0}}{\qs{\theta}^*}=$ 1.80e-03           & Score Test             & 0.020 & 0.046 & 0.150 & 0.346 & 0.644 & 0.876 & 0.962 & 0.998 & 1.000 & 1.000 \\
&IM rej. =  0.274         &  &  &  &  &  &  &  &  &  &  &  \\
           &&                        &        &       &       &       &       &       &       &       &       &       \\
           \hline
\multicolumn{3}{l}{Panel E: Misspecified $(\gamma=-0.4)$} & &       &       &       &       &       &       &       &       &       \\
 Design 1: &Two Covariates&          &        &       &       &       &       &       &       &       &       &       \\
&$\dKL{\ps{0}}{\qs{\theta}^*}=$ 3.14e-03         & Score Test            & 0.024 & 0.046 & 0.130 & 0.312 & 0.616 & 0.852 & 0.968 & 0.992 & 1.000 & 1.000\\
&IM rej. =   0.916    & Moment Ineq. Test ($r_{1n}=4$) & 0.964 & 0.980 & 0.992 & 0.998 & 1.000 & 1.000 & 1.000 & 1.000 & 1.000 & 1.000 \\
&                     & Moment Ineq. Test ($r_{1n}=3$) & 0.964 & 0.980 & 0.992 & 0.996 & 1.000 & 1.000 & 1.000 & 1.000 & 1.000 & 1.000 \\
&                     & Moment Ineq. Test ($r_{1n}=2$) & 0.954 & 0.974 & 0.988 & 0.996 & 1.000 & 1.000 & 1.000 & 1.000 & 1.000 & 1.000 \\
           &&                        &        &       &       &       &       &       &       &       &       &       \\
 Design 2: &Three Covariates&        &        &       &       &       &       &       &       &       &       &       \\
&$\dKL{\ps{0}}{\qs{\theta}^*}=$ 3.14e-03           & Score Test            & 0.030 & 0.054 & 0.148 & 0.372 & 0.658 & 0.862 & 0.972 & 0.996 & 1.000 & 1.000 \\
&IM rej. = 0.760                    &  &  &  &  &  &  &  &  &  &  &  \\
           &&                        &        &       &       &       &       &       &       &       &       &       \\
\hline
\end{tabular}
\begin{tablenotes}
      \small
      \item Note: The simulation results are based on samples of size $n=7,017$ and 500 Monte Carlo repetitions. Panel A reports results for correctly specified models. Panels B-F report results for misspecified models with different values of $\gamma$. The second column reports the number of continuous covariates and the degree of misspecification measured by the KL divergence, and the rejection probability of an infeasible Information Matrix test.
    \end{tablenotes}
	 \end{threeparttable}
}
\end{table}

\begin{table}[!ht]
	\small
    \centering
    \caption{Computational Time Comparisons for Design 1 (in seconds)}\label{tab:computation_time}
\begin{tabular}{lcccc}
\hline
 & &\multicolumn{3}{c}{Moment Inequality Test}\\
 \cline{3-5}
 & Score Test & $r_{1n}=2$ & $r_{1n}=3$ & $r_{1n}=4$\\
\hline
Calculating the Statistic       & 0.25 & 8.75   & 30.04 & 73.62\\
Calculating the Critical Value  & 2.04e-4 & 48.4 & 247.31 & 818.23\\
\hline 
\end{tabular}
\end{table}

In the misspecified case (Panels (B)-(F)), as expected the moment inequality-based test is oversized.
The extent of the size distortion grows with the extent to which the model is misspecified.
To quantify the latter across DGPs, we compute the rejection probability of an infeasible Information Matrix test \citep{White1982} that uses knowledge of the fact that the selection mechanism $\er$ in \eqref{eq:exampleCT_selection_mech} is distributed $\mathrm{Bernoulli}(\kappa_0)$.
Enriched with this information, the model yields a unique prediction $\qs{(\theta,\kappa_0),y|x}$ and a well defined likelihood function, and hence we can obtain the (point identified) maximum likelihood estimator $\hat{\theta}^{MLE}$ that a researcher would obtain if they knew the selection mechanism.
We compute the Hessian and the outer product forms for the covariance matrix, and evaluate them at $\hat{\theta}^{MLE}$ to carry out the Information Matrix test.
We report the rejection probability of this infeasible test in Table \ref{tab:power}, labeling it ``IM rej.''
As can be seen from the table, for levels of $\gamma\in\{-.1,-.2\}$ the rejection probability is low, reaching at most $7.6\%$ under Design 1; yet, \citet{and:shi13}'s test already shows non-trivial size distortions ($16\%$ and $37\%$, respectively). For $\gamma=-0.3,-.4$, the power is higher, reaching $91.6\%$ under Design 1.
For such settings, the size of \citet{and:shi13}'s test is substantially distorted ($75\%$ and $96\%$, respectively, against a $5\%$ nominal level).
In contrast, our test has correct size throughout all simulations, and maintains a power curve that is very similar to the one it displays in the case of correct model specification.

Table \ref{tab:computation_time} reports average computational time in seconds to calculate test statistics and critical values in Design 1 for 500 Monte Carlo replications on Boston University's computing cluster (with Intel Xeon Gold 6132 Processors and 192GB RAM).
Our test statistic is 35-294 times faster to compute than \citeposs{and:shi13}, but the most substantial gain comes from calculation of the critical value: 0.0002 seconds for us, against 48-818 for \citet{and:shi13}, for a total computation time reduction of 228-3,564 times.

\section{Conclusions}
\label{sec:conclusions}
This paper is concerned with statistical inference in incomplete models with set valued predictions.
Such models are typically partially identified, and can be misspecified.
Misspecification can make the identification region of the model's parameters spuriously tight or even empty, raising a challenge for interpreting identification results, and can cause existing testing procedures to severely overreject.
We propose to resolve these problems through an information-based method.
Our method delivers a non-empty pseudo true set which can be interpreted as the set of minimizers of the researcher's ignorance about the true structure, as in \citet{White1982}.
For any given parameter value, our inference method solves a convex program to find the density function that is closest to the data generating process with respect to the Kullback-Leibler information criterion.
It then obtains the score of the likelihood function associated with this density and a Rao score test statistic.
We show that the test statistic has an asymptotically pivotal distribution, is easy to compute, and does not require moment selection.
The associated test has uniformly valid asymptotic size, is applicable to both correctly specified and misspecified models, and allows for discrete and continuous covariates.
Monte Carlo simulations confirm the good computational and statistical properties of our proposed inference method.

\appendix

\section{Proofs of Main Theorems}
\label{sec:app} 

\begin{proof}[\rm \textbf{Proof of Theorem \ref{thm:score_characterization}.}]
\noindent\textbf{Part (i).}
As shown in \eqref{eq:conditional_L}, $\crit(\theta|x)$ is the optimal value function of a convex program.  Below, we fix $x$ and drop conditioning from $\crit$, $\ps{0}$, $\qs{}$, and $\contf$ to ease notation. 
Lemma \ref{lem:licq} in the Online Appendix delivers two key results.
First, there is a collection of events $\cA^{(*e)}\subseteq 2^\cY$ that does not depend on $\theta\in\Theta$, such that
\begin{align}
\core(\contf(\cdot))=\left\{Q\in\cM(\SY,\cX):Q(A)\ge\contf(A),~\forall A\in\cA^{(*e)}\right\},\label{eq:core_det:eq}
\end{align}
where $\core(\contf(\cdot))$ is defined in \eqref{eq:core}, and the cardinality of the collection $\cA^{(*e)}$ is the smallest among any collection of test sets guaranteeing \eqref{eq:core_det:eq}.
Hence, it suffices to verify the dominance condition in \eqref{eq:core} for all $A\in\cA^{(*e)}$ rather than for all $A\in\cC$.\footnote{\cite{gal:hen11} call collections of sets with this property \emph{core determining}. Applying results in \citet{luo:pon:wan25}, Lemma \ref{lem:licq} in the Online Appendix shows that $\cA^{(*e)}$ is an \emph{exact core determining class}, i.e., it has  smallest cardinality among core determining classes.}

Second, the problem
\begin{align*}
\crit(\theta) =\max_{\qs{}\in \Delta}&\sum_{y\in\cY}\ps{0}(y)\ln \qs{}(y)~~~\text{s.t.}~~~ \contf(A)\le \sum_{y\in A}\qs{}(y),~A\in \cA^{(*e)},
\end{align*}
has a unique solution $\qs{}^*$ with unique Lagrange multiplier vector $\lambda^*$ associated with the constraints.
We let $J=|\mathcal A^{(*e)}|$ and we denote sets in $\mathcal A^{(*e)}$ by $A_j,j=1,\dots,J$.

Consider $V(t)=\crit(\theta+th)$ for $h\in\R^{d_\theta}$ and $t\in(-\epsilon,\epsilon)$ for some $\epsilon>0$.
Note that $\qs{}$ may be viewed as a vector because $\cY$ is finite. 
Below, we view $V(t)$ as the optimal value function of the convex program with objective function $f(q,t)=\sum_{y\in\cY}\ps{0}(y)\ln \qs{}(y)$ and convex (affine) constraints $g_j(q,t)= \nu_{\theta+th}(A_j)-\sum_{y\in A_j}\qs{}(y)$, $j=1,\dots,J$.
Note that $\Delta$ is compact and convex. 
Both $f$ and $g_j$'s are continuous and concave in $q$. 
Therefore, for any sequence $\{t_n\}$ with $t_n\downarrow 0$, the maximizer of $L(\theta+t_nh)$ exists. Furthermore, since the domain of the control variable and parameter $\Delta\times (-\epsilon,\epsilon)$ is bounded, the inf-boundedness assumption of \cite{Rockafellar1984} holds. This ensures that the parametric optimization problem indexed by $t$ is directionally stable in the sense of \cite{Gauvin:1990uz}.
Furthermore, their derivatives with respect to $t$ are $f_t(q,t)=0$ and $g_{j,t}(q,t)=\nabla_\theta\nu_{\theta+th}(A_j)^\top h$, and they are continuous in $(q,t)$ by assumption. 
Let $\mathcal L(q,\lambda,t)=f(q,t)+\sum_{j}^J\lambda_jg_j(q,t)$ be the Lagrangian.
By \citet[Corollary 4.2]{Gauvin:1990uz} and $(\qs{}^*,\lambda^*)$ being unique, $V$ is differentiable at $t=0$ and its derivative is given by
\begin{align}
    V'(0)=\tfrac{d}{dt}\mathcal L(\qs{}^*,\lambda^*,t)|_{t=0}=\sum_{j}^J\lambda^*_j\nabla_\theta\nu_\theta(A_j)^\top h.
\end{align}
Since this holds for any $h\in\mathbb R^{d_\theta}$, $\crit(\theta)$ is differentiable with
\begin{align}
    \nabla_\theta \crit(\theta)=\sum_j\lambda^*_j\nabla_\theta\contf(A_j).\label{eq:nablaL}
\end{align}

\noindent\textbf{Part (ii)-Eq.\eqref{eq:conditional_foc}.}
In what follows we construct a score function. 
Let $M=|\cY|$, and order the elements of $\cY$ as $y_1,y_2,\dots, y_M$. 
Let $\cJ=\{j\in\{1,\dots,J\}:\sum_{\tilde y\in A_j}\qs{}^*(\tilde y)=\nu_\theta(A_j)\}$ be the set of active constraints, and let $\cJ^c=\{1,\dots,J\}\setminus\cJ$ collect slack constraints. 
For each $y\in\cY$, let $\cJ(y)=\{j\in\cJ:y\in A_j\}$ collect the indices associated with the active constraints such that $y$ belongs to $A_j$. 
By the Karush-Kuhn-Tucker conditions, differentiating $\mathcal L$ with respect to $\qs{}$ and evaluating it at $\qs{}^*$ yields
\begin{align}
    \tfrac{\ps{0}(y_m)}{\qs{}^*(y_m)}-\sum_{j\in\cJ(y_m)}\lambda^*_j=0,~~~m=1,\dots,M.\label{eq:kkt1}
\end{align}
For each $y\in \cY$, let $e_{\cJ(y)}\in\{0,1\}^J$ be a vector whose $j$-th component is 1 if $j\in \cJ(y)$ and 0 otherwise.
Then, the system of equations \eqref{eq:kkt1} 
can be written as
\begin{align}
    B\lambda^* =r,\label{eq:LMeq}
\end{align}
where $B$ is an $M$-by-$J$ matrix and $r$ is an $M$-by-1 vector with $m$-th row defined as follows
\begin{align}
    [B]_{[m,\cdot]}=e'_{\cJ(y_m)},
    \hspace{1cm}~[r]_m=\tfrac{\ps{0}(y_m)}{\qs{}^*(y_m)}.\label{eq:defBr}
\end{align}
By the complementary slackness conditions, $\lambda_j^*=0$ for any $j\in \cJ^c.$
Hence, \eqref{eq:LMeq} can be reduced to a system of $M$ equations with $S=|\cJ|$ unknowns. 
Eliminate the columns of $B$ corresponding to $j\in\cJ^c$ and let $\tilde B$ denote the resulting submatrix of $B$. 
Similarly, eliminate the components of $\lambda$ corresponding to $j\in\cJ^c$ and let $\tilde\lambda^*$ denote the resulting subvector.
\eqref{eq:LMeq} can be rewritten as $\tilde B\tilde \lambda^* =r$,
where $\tilde B$ is a $M\times S$ matrix whose columns are the representers $\{b_{A_j},j\in \cJ\}$ of the active constraints. 
By Lemma \ref{lem:licq} (iv), the vectors $\{b_{A_j},j\in \cJ\}$ are linearly independent. 
Hence, $\tilde\lambda^*$ solves equation $\tilde B\tilde \lambda^* =r$ 
uniquely, and there exists an $S$-by-$M$ matrix $C$ such that
$\tilde \lambda^*=Cr$.
Let $E_\theta$ be a $d\times S$ matrix that stacks the column vectors $\{\nabla_\theta\contf(A_j),j\in\cJ\}.$ 
Then,
\begin{align*}
    \nabla_\theta \crit(\theta)=\sum_j^{J}\lambda^*_j\nabla_\theta\contf(A_j)=\sum_{j\in \cJ}\lambda^*_j\nabla_\theta\contf(A_j)=E_\theta \tilde\lambda^*=E_\theta Cr
\end{align*}
by \eqref{eq:nablaL} and $\tilde \lambda^*=Cr$. Recall that $r$ is defined in \eqref{eq:defBr}. Hence, $\nabla_\theta \crit(\theta)$ can be written as
\begin{align}
    \nabla_\theta \crit(\theta)=\sum_{m=1}^M\ps{0}(y_m)\tfrac{[E_\theta C]_{[\cdot, m]}}{\qs{}^*(y_m)},\label{eq:score_rep}
\end{align}
where $[\cdot]_{[\cdot, m]}$ selects the $m$-th column of its argument. 
Now let $\score{\theta}{y_m}=\tfrac{[E_\theta C]_m}{\qs{}^*(y_m)}$ and recall that so far we dropped conditioning on $x$ and dependence on $\ps{0,y|x}$.  
\eqref{eq:score_rep} therefore shows that $\nabla_\theta \crit(\theta|x)=\E[\score{\theta}{\ey|\ex;\ps{0,y|x}})|\ex=x]$, and $s_\theta$'s square integrability follows from $\qs{}^*(y)>c$,  $\nabla_\theta\contf(A_j|\ex)$ being square integrable for all $j$, and $\cY$ being a finite set.\smallskip

\noindent\textbf{Part (ii)-\eqref{eq:zero_expected_score}.}
By law of iterated expectations and dominated convergence theorem,
\begin{multline*}
   \E[\score{\theta}{\ey|\ex;\ps{0,y|x}}]=\E[\E[\score{\theta}{\ey|\ex;\ps{0,y|x}}|X]]\\
   = \E\left[\tfrac{\partial}{\partial\theta}\crit(\theta|X)\right]=\tfrac{\partial}{\partial\theta} \E[\crit(\theta|X)]=\tfrac{\partial}{\partial\theta} \crit(\theta)=0,
\end{multline*}
for any $\theta\in \Theta^*(\ps{0})$, where the last equality follows from the first-order condition for maximizing $\theta\mapsto \crit(\theta)$ and the fact that $\Theta^*(\ps{0})\subset\interior\Theta$.
\end{proof}

\begin{proof}[\rm \textbf{Proof of Theorem \ref{thm:argmin_stability}.}]
    The result follows from Lemma~\ref{lemma:continuity}, which establishes joint continuity of $\phi(\cdot,\cdot)$ in $(\theta,\ps{})$, and the use of Berge's maximum theorem, observing that $\fq{\theta}$ does not depend on $\ps{}$.
\end{proof}

We next turn to the proof of Theorem \ref{lem:scores_limit}.
To establish the result, we first show that the expected score satisfies an asymptotic orthogonality condition with respect to the nuisance parameter. 
This result ensures that the score statistic's limiting distribution is insensitive to the nonparametric estimation of the conditional choice probability.
Below, let $\qs{\theta,h_x,y|x}\in \fq{\theta,x}$ be indexed by the structural parameter $\theta$ and equilibrium selection $h_x=\{h_{x,u},\, u\in\cU\}$, with $h_{x,u}\in\cS(x,u;\theta)$ and $\cS(x,u;\theta)$ the set of all conditional densities of $\ey|\ex,\eu$ such that, for any $(x,u)$, its conditional support is $\G(u|x;\theta)$.
Let $\cS(\theta)=\{\cS(x,u;\theta),x\in\cX,u\in\cU\}$.
\begin{lemma}\label{lem:orthogonality}
Suppose Assumptions \ref{ass:finite_support} and \ref{as:newey_modified1_new} hold. Then,  
\begin{align}
D(\theta^*,\ps{0,y|x})[\ps{y|x}-\ps{0,y|x}]=0.\label{eq:orthogonality} 
\end{align}
\end{lemma}
\begin{proof}[\rm Proof of Lemma \ref{lem:orthogonality}]
We rely on an application of \citet[Proposition 2]{Newey:1994aa}.
Let
\begin{align}
	\nq(x,\theta,h_x)&=\E_{\Ps_0}[\ln \qs{\theta,h_x,y|x}(\ey|\ex)|\ex=x].\label{eq:Newey1} 
\end{align}	
Let $h_x(\ps{})=[h_{x,u}(\ps{}),u\in\cU]$, $h_{x,u}(\ps{})\in\cS(x,u;\theta)$, be the selection such that $ \qs{\theta,h_x(\ps{}),y|x}=\qs{\theta,y|x}^*$ when $\ps{}$ replaces $\ps{0}$ in \eqref{eq:q_star}; this selection exists by \citet{art83}. 
Solving the KL projection problem in \eqref{eq:profiling} (with $\ps{}$ replacing $\ps{0}$) is equivalent to maximizing out the equilibrium selection. 
Therefore, the function valued parameter $h(\ps{})\in\cS(\theta)$ solves $h(\ps{0})=\argmax_{\tilde h\in\cS(\theta)}\E_{\Ps_0}[\nq(x,\theta,\tilde h)]$, where the expectation is taken with respect to the true DGP distribution $\Ps_0$ and the dependence of $h(\cdot)$ on $\ps{}$ results from the KL projection step. 
Arguing as in \citet{Newey:1994aa}, for a path $\Ps_\tau$, denoting $h(\tau)=h(\ps{\tau})$, we have that $\E_{\Ps_0}[\nq(x,\theta,h(\tau))]\le\max_{\tilde h\in\cS(\theta)}\E_{\Ps_0}[\nq(x,\theta,\tilde h)]$ and hence $\E_{\Ps_0}[\nq(x,\theta,h(\tau))]$ is maximized at $\tau=0$.
The first order conditions for this maximum are $\partial\E[\nq(x,\theta,h(\tau))]/\partial \tau =0$ for all $\theta$. 
Differentiating one more time with respect to $\theta$ and using the law of iterated expectations,
\begin{align*}
	0&=\tfrac{\partial^2}{\partial \tau\partial\theta}\E_{\Ps_0}[\nq(x,\theta^*,h_x(\tau))]\big\vert_{\tau=0}=\tfrac{\partial}{\partial\tau}\E_{\Ps_0}\left[\tfrac{\partial}{\partial\theta}\E_{\Ps_0}[\ln \qs{\theta,y|x}^*(\ey|\ex)|\ex=x]\right]\Bigg\vert_{\tau=0}\\
	&=\tfrac{\partial}{\partial\tau}\E_{\Ps_0}\left[\E_{\Ps_0}[\score{\theta}{\ey|\ex;p_{\tau,y|x}}|\ex=x]\right]\big\vert_{\tau=0}=\tfrac{\partial}{\partial\tau}\E_{\Ps_0}[\score{\theta^*}{\ey|\ex;\ps{\tau,y|x}}]\big\vert_{\tau=0},
\end{align*} 
where \eqref{eq:conditional_foc} yields the third equality.
Hence, the pathwise derivative is zero.
\end{proof}

\begin{proof}[\rm \textbf{Proof of Theorem \ref{lem:scores_limit}.}]
Let us write the left-hand side of \eqref{eq:scores_limit} as
\begin{align}
&\tfrac{1}{\sqrt n}\sum_{i=1}^n \score{\theta^*}{\ey_i|\ex_i;\hat p_{n,y|x}}=\mathbb G_{n,\theta^*}(\hat p_{n,y|x})+\sqrt n	\E[\score{\theta^*}{\ey_i|\ex_i;\hat p_{n,y|x}}]\notag\\
&=\mathbb G_{n,\theta^*}(\ps{0,y|x})+(\mathbb G_{n,\theta^*}(\hat p_{n,y|x})-\mathbb G_{n,\theta^*}(\ps{0,y|x}))+\sqrt n	\E[\score{\theta^*}{\ey_i|\ex_i;\hat p_{n,y|x}}].\label{eq:scores_limit1}
\end{align}
By Assumption \ref{as:newey_modified2}-\ref{ass:normality}, $\mathbb G_{n,\theta^*}(\ps{0,y|x})=\tfrac{1}{\sqrt n}\sum_{i=1}^n\score{\theta^*}{\ey_i|\ex_i;\ps{0,y|x}})\stackrel{d}{\to}N(0,\Sigma_{\theta^*})$. 
Furthermore, by Assumptions \ref{as:newey_modified2} \ref{ass:consistent_est} and \ref{as:newey_modified2} \ref{ass:equic},  $\mathbb G_{n,\theta^*}(\hat p_{n,y|x})-\mathbb G_{n,\theta^*}(\ps{0,y|x})=o_p(1)$.

Let $r_n\equiv \E[m_{\theta^*}(\ex_i;\hat p_{n,y|x})]-\E[m_{\theta^*}(\ex_i;\ps{0,y|x})]-D(\theta^*,\ps{0,y|x})[\hat p_{n,y|x}-\ps{0,y|x}]$. Then,
\begin{align*}
		\E[\score{\theta^*}{\ey_i|\ex_i;\hat p_{n,y|x}}]&=\E[\E[\score{\theta^*}{\ey_i|\ex_i;\hat p_{n,y|x}}|X=x]]\\
		&=\E[m_{\theta^*}(\ex_i;\ps{0,y|x})]+\E[m_{\theta^*}(\ex_i;\hat p_{n,y|x})-m_{\theta^*}(\ex_i;\ps{0,y|x})]\\
		&=\E[\score{\theta^*}{\ey_i|\ex_i;\ps{0,y|x}}]+D(\theta^*,\ps{0,y|x})[\hat p_{n,y|x}-\ps{0,y|x}]+r_n,
\end{align*}
where the first equality follows from the law of iterated expectations, the second equality follows from the definition of $m_\theta$, and the third equality follows from the law of iterated expectations and the definition of $r_n$. 
As $\theta^*\in\Theta^*(p_0)$, $\E[\score{\theta^*}{\ey_i|\ex_i;\ps{0,y|x}}]=0$. 
By Lemma \ref{lem:orthogonality}, $\sqrt n D(\theta^*,\ps{0,y|x})[\hat p_{n,y|x}-\ps{0,y|x}]=0$. 
Finally, using Assumptions \ref{as:newey_modified1_new} and \ref{as:newey_modified2} \ref{ass:consistent_est},
\begin{align*}
    |\sqrt n r_n|\le 
    \sqrt n\|r_n\|_{L^2_P}\le \sqrt n c\|\hat p_{n,y|x}-\ps{0,y|x}\|_{\cH}^2=o_p(1).
\end{align*}
Hence, by the triangle inequality,
\begin{multline}
	|\sqrt n\E[\score{\theta^*}{\ey_i|\ex_i;\hat p_{n,y|x}}]|\le |\sqrt n\E[\score{\theta^*}{\ey_i|\ex_i;\ps{0,y|x}}]|+|\sqrt n r_n| =o_p(1).\label{eq:scores_limit2}
\end{multline}
Combining \eqref{eq:scores_limit1}-\eqref{eq:scores_limit2} yields the result in \eqref{eq:scores_limit}.
\end{proof}

\begin{proof}[\rm \textbf{Proof of Corollary \ref{cor:Tn_limit}.}]
For the case that $\Sigma_{\theta^*}$ is nonsingular, standard arguments, Assumption \ref{as:newey_modified2}, and \eqref{ass:consistent_cov} yield $T_n(\theta^*)\stackrel{d}{\to}\mathsf{J}$, with $\mathsf{J}\sim\chi^2_{d_\theta}$.
For the case that $\Sigma_{\theta^*}$ is singular, let $\zeta\in\R^{d_\theta}$ be a random vector such that $\zeta = \eta+\nu$, with $\eta\independent\nu$, $\eta\sim N(0,\Sigma_{\theta^*})$, and $\nu\sim N(0,\varepsilon\Psi_{\theta^*})$,
where $\Psi_{\theta^*}$ is the population analog of $\hat\Psi_{n,\theta^*}$ (see \eqref{eq:tilde_Sigma} and subsequent explanation of notation).
Let $\tilde\Sigma_{\theta^*}=\Sigma_{\theta^*}+\varepsilon\Psi_{\theta^*}$.
It follows from standard arguments that $T_n\stackrel{d}{\to}\eta^\top\tilde\Sigma_{\theta^*}^{-1}\eta$, and that $\zeta^\top\tilde\Sigma_{\theta^*}^{-1} \zeta\sim \mathsf{J}$.
Next, let $K=\{x\in\R^{d_\theta}:x^\top\tilde\Sigma_{\theta^*}^{-1}x\le c_{d_\theta,\alpha}\}$, and note that this set is convex and symmetric.
By Anderson's Lemma \citep[][Lemma 3.11.4]{Vaart:1996wk},
\begin{align*}
1-\alpha=\Ps(\zeta^\top\tilde\Sigma_{\theta^*}^{-1} \zeta\le c_{d_\theta,\alpha})=\Ps(\eta+\nu\in K)\le \Ps(\eta\in K)=\Ps(\eta^\top\tilde\Sigma_{\theta^*}^{-1} \eta\le c_{d_\theta,\alpha}).
\end{align*}
Then 
$
\lim\sup_{n\to\infty}\Ps(T_n(\theta^*)>c_{d_\theta,\alpha})=\Ps(\eta^\top\tilde\Sigma_{\theta^*}^{-1} \eta > c_{d_\theta,\alpha})
\le \Ps(\zeta^\top\tilde\Sigma_{\theta^*}^{-1} \zeta > c_{d_\theta,\alpha})=\alpha
$.
\end{proof}

\begin{proof}[\rm \textbf{Proof of Theorem \ref{thm:uniform_coverage}.}]
Let $\{\ps{0n},\theta^*_n\}\in\{(p,\vartheta^*):p~\text{is the Radon-Nykodim derivative of }P\in\cP,\vartheta^*\in\Theta^*(p)\}$ be a sequence such that:
\begin{equation*}
\liminf_{n\to\infty}\inf_{P\in\cP}	\inf_{\vartheta^* \in \Theta^*(p)}P(\vartheta^* \in CS_n)=\liminf_{n\to\infty}P_n(\theta_n^* \in CS_n),
\end{equation*}
with $CS_n$ defined in \eqref{eq:CSn}. 
Let $\{l_n\}$ be a subsequence of $\{n\}$ such that 
\begin{equation*}
\liminf_{n\to\infty}P_n(\theta_n^* \in CS_n)=\lim_{n\to\infty}P_{l_n}(\theta_{l_n}^* \in CS_{l_n})).
\end{equation*}
Then there is a further subsequence $\{a_n\}$ of $\{l_n\}$ such that
\begin{align*}
\lim_{a_n \to \infty} \Sigma_{\theta^*_{a_n}}=\Sigma^*\in\mathbb{S},
\end{align*}
where $\mathbb{S}$ is the collection of positive semi-definite $d_\theta\times d_\theta$ matrices.
To avoid multiple subscripts, with some abuse of notation we write $(P_n,\theta^*_n)$ to refer to $(P_{a_n},\theta^*_{a_n})$.
We establish the claim by showing that along the subsequence $(P_n,\theta^*_n)$, the results in Theorem \ref{thm:score_characterization}, Lemma \ref{lem:orthogonality}, and Theorem \ref{lem:scores_limit} continue to hold.

For Theorem \ref{thm:score_characterization}, note first that the collection of events $\cA^{(*e)}$ does not depend on $(P_n,\theta^*_n)$, as can be seen in the proof of Lemma \ref{lem:licq}-(i).
Second, parts (ii), (iii) and (iv) of Lemma \ref{lem:licq} continue to hold along the subsequence $(P_n,\theta^*_n)$ under the uniform version of Assumption \ref{ass:finite_support} stated in Theorem \ref{thm:uniform_coverage}.

For Lemma \ref{lem:orthogonality}, we again note that the result holds uniformly over $\cP$, under the uniform version of Assumptions \ref{ass:finite_support}, \ref{as:newey_modified1_new}, and \ref{as:newey_modified2} \ref{ass:consistent_est} stated in Theorem \ref{thm:uniform_coverage}.

For Theorem \ref{lem:scores_limit}, under the uniform version of Assumptions \ref{ass:finite_support}, \ref{as:newey_modified1_new}, and \ref{as:newey_modified2} stated in Theorem \ref{thm:uniform_coverage}, we have that by Assumption \ref{as:newey_modified2}, $\mathbb G_{n,\theta^*_n}(\ps{0n,y|x})=\tfrac{1}{\sqrt n}\sum_{i=1}^n\score{\theta^*_n}{\ey_i|\ex_i;\ps{0n,y|x}})\stackrel{d}{\to}N(0,\Sigma^*)$, and $\mathbb G_{n,\theta^*_n}(\hat p_{n,y|x})-\mathbb G_{n,\theta^*_n}(\ps{0n,y|x})=o_{P_n}(1)$.  
Arguing as in the proof of Theorem \ref{lem:scores_limit}, \eqref{eq:scores_limit1}-\eqref{eq:scores_limit2} continue to hold along the sequence $(P_n,\theta^*_n)$, and therefore
\begin{align*}
	\tfrac{1}{\sqrt n}\sum_{i=1}^n \score{\theta^*_n}{\ey_i|\ex_i;\hat p_{n,y|x}}=\tfrac{1}{\sqrt n}\sum_{i=1}^n\score{\theta^*_n}{\ey_i|\ex_i;\ps{0n,y|x}}+o_\cP(1)\stackrel{d}{\to}N(0,\Sigma^*).
\end{align*}
The final result follows arguing as in the proof of Corollary \ref{cor:Tn_limit}.
\end{proof}

\newpage
\thispagestyle{empty}

\begin{center}
    \uppercase{Online Appendix for: ``Information Based Inference\\
in Models with Set-Valued Predictions and Misspecification"}
\end{center}

\section{Lemmas Used in Proofs of Main Theorems}
\subsection{Pseudo-True Sets and Score Function}
\label{proof:score:theorem}
We give two lemmas that we use to show the differentiability of $\crit(\theta|x)$.
We let $\cC$ denote the collection of closed subsets of $\cY$.
Let $\cA\subseteq 2^{\cY}$ be a collection of events, including $\{y\}$ for each $y\in\cY$ and the entire $\cY$, such that the constraint set in \eqref{eq:const_ineq}-\eqref{eq:const_eq} equals $\fq{\theta,x}$. 
To ease notation, we drop conditioning on $\ex=x$.  
Let $\cA_=$ collect all events in $\cA$ corresponding to equality constraints in \eqref{eq:core} along with all inequalities such that $Q(A)=\sum_{y\in A}q(y)=\contf(A)$ for all $Q\in \core(\contf(\cdot))$ \citep[labeled implicit-equalities in][Definition 2.2]{luo:pon:wan25}. Let $\cA_{\ge}$ collect the remaining events. 
The inclusion of $\{y\}$ for each $y\in\cY$ and $\cY$ in $\cA$ guarantees that  the constraint set in \eqref{eq:const_ineq}-\eqref{eq:const_eq} is a subset of the $d_Y-1$ dimensional unit-simplex.
Consider the following problem:
\begin{align}
\mathbf P(\cA):\qquad\upsilon(\theta;\cA)\equiv \max_{\qs{}\in\R^{d_Y}}&\sum_{y\in\mathcal Y}\ps{0}(y)\ln \qs{}(y)\label{eq:define_valfun}\\
	s.t.~~~& \sum_{y\in A}\qs{}(y)\ge \contf(A),~A\in \mathcal A_{\ge}\label{eq:const_ineq}\\
	&\sum_{y\in A}\qs{}(y)= \contf(A),~A\in \mathcal A_{=},\label{eq:const_eq}
\end{align}
Let $A\subset \cY$. 
As $\cY$ is finite, one may represent the probability that any distribution $P$ with probability mass function $p\in\Delta$ assigns to a set $A$ through a \emph{representer} $a$ of $A$, by writing  
\begin{align*}
	\Ps(A)=p^\top a,
\end{align*}
with $a\in \{0,1\}^{d_Y}$.  
For example, take $\cY=\{(0,0),(0,1),(1,0),(1,1)\}$, $A=\{(1,0),(1,1)\}$, and $a=(0,0,1,1)^\top$. 
Then, $P(A)=p^\top a$. 
Similarly, for some $b_A\in\{0,1\}^{d_Y}$, the constraints in \eqref{eq:const_ineq}-\eqref{eq:const_eq} can be written as
\begin{align*}
\qs{}^\top b_A\ge \contf(A),~A\in \cA_{\ge}\quad\quad\text{and}\quad\quad
	\qs{}^\top b_A= \contf(A),~A\in \cA_=.
\end{align*}
The next lemma establishes a basic fact about polyhedra, adapted to our problem.
\begin{lemma}
\label{lem:ri-poly}
Let Assumption~\ref{ass:finite_support} hold and $b_A\neq 0$ for all $A\in\cA_{\ge}$. Then the relative interior of 
$\fq{\theta}$ is $\ri(\fq{\theta})=\{q\in\fq{\theta}:q^\top b_A>\contf(A),~\forall A\in \cA_{\ge}\}$.
\end{lemma}

\begin{proof}
Express $\fq{\theta}$ as $\fq{\theta}=\fq{\theta}^=\cap\Big\{\bigcap_{A\in\cA_{\ge}}\fq{\theta,A}^\ge\Big\}$, where
\begin{align*}
    \fq{\theta}^=&\equiv\{q\in\R^{d_Y}:q^\top\one =1,~q^\top b_A=\contf(A),~\forall A\in \cA_=\},\\
    \fq{\theta,A}^\ge&\equiv\{q\in\R^{d_Y}:q^\top b_A\ge\contf(A)\},\quad A\in\cA_{\ge}.
\end{align*}
Since $\fq{\theta}^=$ is affine, $\ri(\fq{\theta}^=)=\fq{\theta}^=$. For $A\in \cA_{\ge}$, since $b_A\neq 0$, we have
$\ri(\fq{\theta,A}^\ge)=\{q\in\R^{d_Y}:q^\top b_A>\contf(A)\}$.
Moreover, for each $A\in \cA_{\ge}$ there exists $q^A\in \fq{\theta}$ with $(q^A)^\top b_A>\contf(A)$ (because $A\notin \cA_=$).
By averaging finitely many such points, we obtain $\bar q\in \fq{\theta}$ satisfying $\bar{q}^\top b_A>\contf(A)$ for all $A\in \cA_\ge$, hence $\bar{q}\in\ri\left(\fq{\theta}^=\right)\cap\ri\left(\left\{\cap_{A\in\cA_{\ge}}\fq{\theta,A}^\ge\right\}\right)\neq\emptyset$.
The conclusion follows by \citet[Proposition~2.42]{Rockafellar_Wets2005aBK}.
\end{proof}
For given $\cA$, let $\cA_{\mathrm{bind}}(\qs{}^*)=\{A\in\cA_\ge:(\qs{}^*)^\top b_A=\contf(A)\}$ collect the events associated with non-identity inequality constraints that bind at $\qs{}^*$. 
Define the \emph{active face} of $\fq{\theta}$ at $\qs{}^*$ (i.e., the face cut out by the binding non-identity inequalities) as
\begin{align*}
    \fq{\mathrm{act},\theta}(\qs{}^*)\equiv\left\{\qs{}\in \fq{\theta}:\ \qs{}^\top b_A=\contf(A),~\forall A\in\cA_{\mathrm{bind}}(\qs{}^*)\right\}.
\end{align*}
For each event $\tilde{A}$ in the \emph{active constraint list} $\{A\in\cA_=\cup\cA_{\mathrm{bind}}(\qs{}^*):\qs{}^\top b_A=\contf(A)\}$, define the active face at $\qs{}^*$ obtained by \emph{removing $\tilde{A}$} as follows. If $\tilde{A}\in\cA_=$, then
\begin{align}
    \fq{\theta}^{-\tilde{A}}&\equiv\left\{q\in\R^{d_Y}: q^\top b_A=\contf(A),~ \forall A\in\cA_{=}\setminus\{\tilde{A}\},q^\top b_A\ge \contf(A),~ \forall A\in\cA_\ge\right\},\label{eq:remove_eq}\\
    \hspace{-.5cm}\fq{\mathrm{act},\theta}^{-\tilde{A}}(\qs{}^*)&\equiv
       \left\{q\in\fq{\theta}^{-\tilde{A}}: q^\top b_A=\contf(A),~ \forall A\in\cA_{\mathrm{bind}}(\qs{}^*)\right\}.\label{eq:remove1}
\end{align}
If $\tilde{A}\in\cA_{\mathrm{bind}}(\qs{}^*)$, then all inequalities in $\cA^{(*e)}$ remain enforced as weak inequalities and
\begin{align}
    \fq{\mathrm{act},\theta}^{-\tilde{A}}(\qs{}^*)&\equiv
       \left\{q\in\fq{\theta}: q^\top b_A=\contf(A),~ \forall A\in\cA_{\mathrm{bind}}(\qs{}^*)\setminus \tilde{A}\right\}.\label{eq:remove2}
\end{align}
\begin{definition}\label{def:irredundant}
    We say that the set of (in)equalities in $\cA$ active at $\qs{\theta,y|x}^*$ are irredundant if when any event $\tilde{A}\in\{A\in\cA_=\cup\cA_{\mathrm{bind}}(\qs{}^*):\qs{}^\top b_A=\contf(A)\}$ is removed, $\fq{\mathrm{act},\theta}^{-\tilde{A}}(\qs{}^*)\supsetneq \fq{\mathrm{act},\theta}(\qs{}^*)$ with $\fq{\mathrm{act},\theta}^{-\tilde{A}}(\qs{}^*)$ defined in \eqref{eq:remove1}-\eqref{eq:remove2}.
\end{definition}
\vspace{-.2cm}
\begin{lemma}\label{lem:licq}
	Let Assumption \ref{ass:finite_support} hold.  
	Then (i) there exists a collection $\cA^{(*e)}\subseteq 2^\cY$ not dependent on $\theta\in\Theta$ such that $\core(\contf(\cdot))=\{Q\in\cM(\SY):Q(A)\ge\contf(A),A\in\cA^{(*e)}\}$ for $\core(\contf(\cdot))$ in \eqref{eq:core}, and no collection $\cA^*\subseteq 2^\cY$ of cardinality smaller than $\cA^{(*e)}$ suffices to characterize $\core(\contf(\cdot))$; (ii) the optimal value of $\mathbf P(\mathcal A^{(*e)})$ is $L(\theta)$; (iii) the solution $\qs{}^*\in\Delta$ to problem $\mathbf P(\mathcal A^{(*e)})$ is unique and also solves problem $\mathbf P(\mathcal A^{all})$ with $\mathcal A^{all}=\{A:A\subseteq\cY,~A~\text{closed}\}$; (iv) its associated Lagrange multiplier vector $\lambda^*$ is unique.
\end{lemma}
\begin{proof}
\textbf{Part (i).}
A family of closed sets $\cA^*$ is a \emph{core determining class} \citep{gal:hen11} if any probability measure $Q$ defined on $\cY$ satisfying the inequalities $Q(A)\ge\contf(A)~~\text{for all}~A\in\cA^*$
satisfies the inequalities $Q(A)\ge\contf(A)$ for all closed sets $A\subseteq \cY$.
A family of closed sets $\cA$ is a \emph{smallest} core determining class \citep*{luo:pon:wan25} if it has the smallest cardinality among all core determining classes.
Part (i) follows directly by Corollary 1.1 in \citet*{luo:pon:wan25}, due to Assumption~\ref{ass:finite_support}-e, as their result shows that if the support of the random set $\G(\eu|x;\theta)$, conditional on $\ex = x$ does not depend on $\theta$, neither does the smallest core determining class.

\textbf{Part (ii).}
By the definition of a smallest core determining class, the collection of inequalities in $\cA^{(*e)}$ yields the same constraint set as in \eqref{eq:conditional_L}.
The solution to $\mathbf P(\mathcal A^{(*e)})$ exists by the continuity of the objective function and the compactness of the probability simplex.

\textbf{Part (iii).}
As $\qs{} \mapsto \E\ln \qs{}$ is strictly concave and the domain of $\qs{}$ is convex, uniqueness of $\qs{}^*$ follows. 
As the collection of inequalities in $\cA^{(*e)}$ yields the same constraint set as in \eqref{eq:conditional_L}, $\qs{}^*$ solves also the original problem in \eqref{eq:conditional_L}.

\textbf{Part (iv).}
We show that the Linear Independence Constraint Qualification (LICQ) holds and hence the Karush-Kuhn-Tucker conditions hold at the feasible point $\qs{}^*$ with a unique Lagrange multiplier $\lambda^*$. 
Consider first the case $\qs{}^*\in\ri(\fq{\theta})$.
Let $\cA^{(*e)}_=$ collect all equalities and implicit equalities in $\cA^{(*e)}$ and $\cA^{(*e)}_\ge=\cA^{(*e)}\setminus\cA^{(*e)}_=$. 
By Assumption~\ref{ass:finite_support}-(d), the non-negativity constraints do not bind at $\qs{}^*$. By Lemma~\ref{lem:ri-poly}, for any $A\in\cA^{(*e)}_{\ge}$, $(\qs{}^*)^\top b_A>\contf(A)$. Hence, $\qs{}^*$ is defined by restrictions associated with events in $\cA^{(*e)}_{=}$ such that $q^\top b_A=\contf(A)$ for all $q\in\fq{\theta}$.
Since the constraints are affine, their gradients are constant and equal to $b_A$.
Assume by contradiction that $\{b_A:A\in\cA^{(*e)}_=\}$ is linearly dependent. Then there exists $\tilde{A}\in\cA^{(*e)}_=$ such that $b_{\tilde{A}}\in \operatorname{span}\{b_A:A\in\cA^{(*e)}_=\setminus\{\tilde{A}\}\}$, and hence there exist scalars $\{\beta_A\}_{A\in\cA^{(*e)}_=\setminus\{\tilde{A}\}}$ such that $b_{\tilde{A}}=\sum_{A\in\cA^{(*e)}_=\setminus\{\tilde{A}\}}\beta_A b_A$.
Let $\fq{\theta}^{-\tilde{A}}$ denote the feasible set in \eqref{eq:remove_eq} associated with $\cA^{(*e)}$ obtained by removing the equality constraint for $\tilde{A}$ from $\cA^{(*e)}_=$.
We show $\fq{\theta}^{-\tilde{A}}=\fq{\theta}$, which contradicts $\cA^{(*e)}$ being a smallest core determining class.
The inclusion $\fq{\theta}^{-\tilde{A}}\supseteq\fq{\theta}$ is immediate. To show $\fq{\theta}^{-\tilde{A}}\subseteq\fq{\theta}$, take any $q\in\fq{\theta}^{-\tilde{A}}$.
Left multiplying $b_{\tilde{A}}$ by $q^\top$ gives $q^\top b_{\tilde{A}}=\sum_{A\in\cA^{(*e)}_=\setminus\{\tilde{A}\}}\beta_A(q^\top b_A)=\sum_{A\in\cA^{(*e)}_=\setminus\{\tilde{A}\}}\beta_A\contf(A)\equiv\bar\beta$, so $q^\top b_{\tilde{A}}$ is constant on $\fq{\theta}^{-\tilde{A}}$ and equal to $\bar\beta$.
Since $\qs{}^*\in\fq{\theta}\subseteq\fq{\theta}^{-\tilde{A}}$, it also satisfies the removed constraint: $(\qs{}^*)^\top b_{\tilde{A}}=\contf(\tilde{A})$. By the preceding calculation applied to $\qs{}^*$, $(\qs{}^*)^\top b_{\tilde{A}}=\bar\beta$. Hence $\bar\beta=\contf(\tilde{A})$, and therefore every $\qs{}\in\fq{\theta}^{-\tilde{A}}$ satisfies $\qs{}^\top b_{\tilde{A}}=\bar\beta=\contf(\tilde{A})$. Hence, $\fq{\theta}^{-\tilde{A}}= \fq{\theta}$, contradicting that $\cA^{(*e)}$ is a smallest core determining class.
Therefore, $\{b_A:A\in\cA^{(*e)}_=\}$ must be linearly independent and LICQ holds at $\qs{}^*$.

Finally, consider the case that $\qs{}^*\in\operatorname{rbd}(\fq{\theta})\equiv\fq{\theta}\setminus\ri(\fq{\theta})$. 
By Assumption~\ref{ass:finite_support}-(f), the set of equalities and inequalities active at $\qs{}^*$ is irreducible per Definition~\ref{def:irredundant}. Hence, the same argument as in the previous case applies by replacing $\cA^{(*e)}_=$ with $\cA^{(*e)}_=\cup\cA_{\mathrm{bind}}^{(*e)}(\qs{}^*)$.
\end{proof}

\begin{lemma}\label{lemma:continuity}
    Under the assumptions of Theorem~\ref{thm:argmin_stability}, and recalling that for some $\underline{c}>0$, $\inf_{x\in\cX}\min_{y\in\cY}\ps{0,y|x}(y|x)\ge\underline{c}$, $\phi(\cdot,\cdot)$ as defined in \eqref{eq:Theta_star_param_prob} is jointly continuous in $(\theta,\ps{})$.
\end{lemma}
\begin{proof}
Passing the infimum in \eqref{eq:Theta_star_param_prob} inside the integral and recalling $\fq{\theta,x}$ from \eqref{eq:core_density}, we can write $\phi(\theta,\ps{})=\int_\cX\ps{x}(x)v(\theta,\ps{y|x},x)d\xi(x)$, where $v(\theta,\ps{y|x};x)=\inf_{\qs{y|x}\in\fq{\theta,x}}I(\ps{y|x}||q_{y|x})$ and, denoting $H(\ps{y|x})\equiv-\int_{\cY}\ps{y|x}(y|x)\ln \ps{y|x}(y|x)d\mu(y)$,
\begin{align}
I(\ps{y|x}||q_{y|x})=-H(\ps{y|x})-\int_\cY\ps{y|x}(y|x)\ln\qs{y|x}(y|x)d\mu(y).	\label{eq:entropy1cont}
\end{align}
Let $(\theta_n,p_n)$ be a sequence such that $(\theta_n,p_n)\to (\theta,p)$, where $p_{n,y|x}(y|x)-p_{y|x}(y|x)\to 0$ uniformly in $y$ and $x,$ and $\| p_{n,x}-p_{x}\|_{L^1_\xi}\to 0$. Then,
\vspace{-.5cm}
\begin{multline*}
\phi(\theta_n,p_n)-\phi(\theta,p)=\int_\cX p_{n,x}(x)v(\theta_n,\ps{n,y|x},x)d\xi(x)-\int_\cX\ps{x}(x)v(\theta,\ps{y|x},x)d\xi(x)\\
= \underbrace{\int_\cX p_{n,x}(x)[v(\theta_n,\ps{n,y|x},x)-v(\theta,\ps{y|x},x)]d\xi(x)}_{(i)}+\underbrace{\int_\cX [p_{n,x}(x)-\ps{x}(x)]v(\theta,\ps{y|x},x)d\xi(x)}_{(ii)}.
\end{multline*}
Using Assumption~\ref{ass:finite_support}-(d), $|v(\theta,\ps{y|x},x)|\le \sum_{y\in\cY}\ps{y|x}(y|x)\ln \tfrac{1}{c}\le d_Y\ln \tfrac{1}{c}=:K$, and $(ii)$ can be bounded by  $K\| p_{n,x}-p_{x}\|_{L^1_\xi}\to 0$. Term $(i)$ can be further decomposed as follows:
\small{\begin{align*}
\hspace{-.1cm}\underbrace{\int_\cX p_{n,x}(x)[v(\theta_n,\ps{n,y|x},x)-v(\theta,\ps{n,y|x},x)]d\xi(x)}_{(i,a)}
+\underbrace{\int_\cX p_{n,x}(x)[v(\theta,\ps{n,y|x},x)-v(\theta,\ps{y|x},x)]d\xi(x)}_{(i,b)}.
\end{align*}}
Recall that $v(\theta_n,\ps{n,y|x},x)=-H(\ps{n,y|x})-\mathbb E_{\ps{n,y|x}}[\ln \qs{\theta_n,y|x}^*(Y|x)|x].$
By the mean-value theorem and Theorem~\ref{thm:score_characterization}, there exists $\tilde \theta_n$ between $\theta_n$ and $\theta$ such that
\begin{align*}
(i,a)\le \int_{\cX} p_{n,x}(x)\big|\mathbb E_{\ps{n,y|x}}[s_{\tilde\theta_n,y|x}(Y|x;\ps{n,y|x})|x]\big|\|\theta_n-\theta\|d\xi(x)
\le M^{1/2}\|\theta_n-\theta\|,
\end{align*}
where the last inequality follows because $\sup_{\theta\in\Theta}\E[\|\score{\theta}{\ey|\ex;\ps{y|x}}\|^2]\le M$ and 
\begin{multline*}
\|E_{\ps{n,y|x}}[s_{\tilde\theta_n,y|x}(Y|X;\ps{n,y|x})|X]\|_{L^1_{\ps{n,y|x}}}\le 	\|E_{\ps{n,y|x}}[s_{\tilde\theta_n,y|x}(Y|X;\ps{n,y|x})|X]\|_{L^2_{\ps{n,y|x}}}
\le M^{1/2}.
\end{multline*}
Hence, $(i,a)$ tends to 0 as $\theta_n\to \theta$.
To show that $(i,b)$ vanishes, let $\|\ps{y|x}-\ps{y|x}'\|_1=\sum_{y\in\cY}|\ps{y|x}(y|x)-\ps{y|x}'(y|x)|$. 
Given the lower bound on $\ps{0,y|x}(y|x)$, for $n$ large enough, $\inf_{x\in\cX}\min_{y\in\cY}\ps{n,y|x}(y|x)\ge\underline{c}/2$. By~\eqref{eq:entropy1cont} and $\ps{0,y|x}(y|x)\ln q(y|x)\le 0$ for all $y\in\cY$, $\qs{n,\theta,y|x}^*(y|x)\ge\beta>0$ for a constant $\beta$ that depends on $(\underline{c},c,d_Y)$, and $\qs{n,\theta,y|x}^*=\arg\inf_{q\in\fq{\theta,x}}I(\ps{n,y|x}||q)$.
Take $c^*=\min\{\beta,c\}$ and let $\fq{\theta,x,c^*}=\{q\in\fq{\theta,x}:q(y|x)\ge c^*,\forall y\in\cY\}$.
For any $\theta$ and $\ps{y|x},\ps{y|x}'$ such that $\|\ps{y|x}-\ps{y|x}'\|_1\le 1/2$,
\begin{align*}
\Big|\inf_{\qs{y|x}\in\fq{\theta,x,c^*}}I(\ps{y|x}||q_{y|x})-\inf_{\qs{y|x}\in\fq{\theta,x,c^*}}I(\ps{y|x}'||q_{y|x})\Big|
&\le \sup_{\qs{y|x}\in\fq{\theta,x,c^*}} |I(\ps{y|x}||q_{y|x})-I(\ps{y|x}'||q_{y|x})|\\
&\le \omega(\|\ps{y|x}-\ps{y|x}'\|_1),
\end{align*}
where $\omega(v)=v \ln \tfrac{d_Y}{v}+v \ln \tfrac{1}{c^*}$ by Lemma 2.7 in \cite{csiszar2011information} and \eqref{eq:entropy1cont}, and the convergence is uniform in $x$ as we assumed $p_{n,y|x}(y|x)-p_{y|x}(y|x)\to 0$ uniformly in $y$ and $x$. 
\end{proof}

\subsection{Derivation of Results and Verification of Conditions for the Entry Game Example \ref{example:CT}}
\label{appendix:entry_game}
\begin{proposition}\label{prop:games_profiled}
Under the assumptions laid out in Example \ref{example:CT}, (i) the profiled likelihood $\qs{\theta,y|x}^*$ for $y\in\{(0,0),(0,1),(1,0),(1,1)\}$ is given in equations \eqref{eq:eg_prof_likelihood1}-\eqref{eq:eg_prof_likelihood4}.
(ii) The score function is given in equations \eqref{eq:eg_score1}-\eqref{eq:eg_score4}.
\end{proposition}
\begin{proof}
In this example $\qs{y|x}((0,1)|x) =\eta_1(\theta;x)-\qs{y|x}((1,0)|x)$.
Let $z\equiv\qs{y|x}((1,0)|x)$ 
and $c(\theta)=\ps{0,y|x}((0,0)|x)\ln\f(S_{\{(0,0)\}|x;\theta})+\ps{0,y|x}((1,1)|x)\ln \f(S_{\{(1,1)\}|x;\theta})$.
Using that $\sum_{\tilde{y}} \qs{y|x}(\tilde{y}|x)=1$, 
rewrite the optimization problem as
\vspace{-.25cm}
\begin{align*}
V(\theta)=	\sup_{z}&~c(\theta)+ \ps{0,y|x}((1,0)|x)\ln z+\ps{0,y|x}((0,1)|x)\ln (\eta_1(\theta;x)-z)\\
	s.t.& ~z-\eta_3(\theta;x)\ge 0 \quad\text{ and }\quad\eta_2(\theta;x)-z\ge 0,
\end{align*}
Define the Lagrangian of this problem by
\vspace{-.25cm}
\begin{multline*}
\mathcal L(z,\lambda,\theta)=c(\theta)+	\ps{0,y|x}((1,0)|x)\ln z+\ps{0,y|x}((0,1)|x)\ln (\eta_1(\theta;x)-z)\\
+\lambda_1(z-\eta_3(\theta;x))+\lambda_2(\eta_2(\theta;x)-z).
\end{multline*}

\noindent\emph{Part (i).}
Since $c(\theta)$ does not affect the solution, we drop it in what follows.
The Karush-Kuhn-Tucker (KKT) conditions of this problem are $\ps{0,y|x}(1,0|x)\tfrac{1}{z}-\ps{0,y|x}(0,1|x)\tfrac{1}{\eta_1(\theta;x)-z}+\lambda_1-\lambda_2=0$, $\lambda_1(z-\eta_3(\theta;x))=0$, $\lambda_2(\eta_2(\theta;x)-z)=0$, and $\lambda_1,\lambda_2\ge 0$.
Then, \eqref{eq:eg_prof_likelihood1}-\eqref{eq:eg_prof_likelihood4} are obtained by solving the KKT conditions for three cases: (1) $\lambda_1=\lambda_2=0$; (2) $\lambda_2>0$; (3) $\lambda_1>0$. \citep[See][pp.6-7, for a full derivation.]{KM24}

\noindent\emph{Part (ii).} 
To establish this result, we let $\theta(t)=\theta+th, h\in \mathbb R^d$ and define
$
\tilde V(t)=V(\theta(t)),~\tilde{\mathcal L}(z,\lambda,t)=\mathcal L(z,\lambda,\theta(t))
$,
so that the derivative of the Lagrangian becomes
\vspace{-.2cm}
\begin{multline*}
\tilde {\mathcal L}_t(z,\lambda,\theta(t))=\ps{0,y|x}((0,0)|x)\tfrac{\nabla_\theta \f(S_{\{(0,0)\}|x;\theta})^\top h}{\f(S_{\{(0,0)\}|x;\theta})}+\ps{0,y|x}((1,1)|x)\tfrac{\nabla_\theta\f(S_{\{(1,1)\}|x;\theta})^\top h}{ \f(S_{\{(1,1)\}|x;\theta})}\\
+\ps{0,y|x}((0,1)|x)\tfrac{\nabla_\theta \eta_1(\theta;x)^\top h}{ \eta_1(\theta;x)-z}-\lambda_1 \nabla_\theta\eta_3(\theta;x)^\top h+\lambda_2\nabla_\theta\eta_2(\theta;x)^\top h.
\end{multline*}
Hence, for any sequence $\{t_n\}$ with $t_n\downarrow 0$, the maximizer of $\mathcal L(z,\lambda,\theta+t_nh)$ exists. 
The domain of the control variable and parameter $[0,1]\times (-\epsilon,\epsilon)$ is bounded, hence the inf-boundedness assumption of \cite{Rockafellar1984} holds, ensuring that the parametric optimization problem indexed by $t$ is directionally stable in the sense of \cite{Gauvin:1990uz}.
Using that $(\qs{}^*,\lambda^*)$ are unique as shown above, we apply \citet[Corollary 4.2]{Gauvin:1990uz} to obtain full differentiability of $V$.
The derivative of $\tilde V(t)$ can be obtained analytically solving three cases: (1) $\lambda_1=\lambda_2=0$; (2) $\lambda_1=0,\lambda_2>0$; (3) $\lambda_1>0,\lambda_2=0$. \citep[See][p.8, for a full derivation.]{KM24}
\end{proof}
\vspace{-.5cm}

\begin{proposition}
\label{prop:verify_ex_LHC}
The mapping $\gamma\mapsto\Theta^*(\gamma)$ in Example~\ref{example:CT} on p.~\pageref{example:entry_uniform} is continuous.
\end{proposition}
\begin{proof}
Denoting $\vartheta_j=0.5+\theta_j$, we have $\Theta^*(\ps{\gamma})=\Xi^*(\ps{\gamma})-[0.5~0.5]^\top$ for $\Xi^*(\ps{\gamma})\equiv\{\vartheta\in[0.05,0.5]^2:\vartheta_1\vartheta_2=\ps{\gamma}((1,1)|1)$; $\tfrac{\ps{\gamma}((1,1)|1)}{1-\ps{\gamma}((1,0)|1)}\le\vartheta_1\le 1-\ps{\gamma}((0,1)|1)\}$. 
We then have that for any $\theta\in\Theta$ and $\vartheta=[0.5~0.5]^\top+\theta$,
\begin{align}
    \dist(\theta,\Theta^*(\ps{\gamma}))=\dist(\vartheta,\Xi^*(\ps{\gamma}))=\min_{\tilde\vartheta_1\in B(\ps{\gamma})}\left\|\left[\tilde\vartheta_1~~\tfrac{\ps{\gamma}((1,1)|1)}{\tilde\vartheta_1}\right]^\top-\vartheta\right\|,\label{eq:distRW}
\end{align}
with $B(\ps{\gamma})=\Big[\max\Big\{0.05,2\ps{\gamma}((1,1)|1),\tfrac{\ps{\gamma}((1,1)|1)}{1-\ps{\gamma}((1,0)|1)}\Big\},\min\Big\{0.5,\tfrac{\ps{\gamma}((1,1)|1)}{0.05},1-\ps{\gamma}((0,1)|1)\Big\}\Big]$.
Recall that $\ps{\gamma}(y|1)$ is an affine function of $\gamma$ for all $y\in\cY$. Hence, the objective function in the minimization problem in \eqref{eq:distRW} is jointly continuous in $\vartheta_1$ and $\ps{\gamma}$, and $B(\ps{\gamma})$ converges in Hausdorff distance to $B(\ps{0})$ when $\ps{\gamma}\to\ps{\gamma_0}$ (where for non-empty intervals $A=[a_1,a_2],B=[b_1,b_2]\subset\R$, $\dist_H(A,B)=\max\{|a_1-b_1|,|a_2-b_2|\}$), and hence it is a continuous correspondence on $\Gamma$.
It follows that all assumptions of Berge's maximum theorem are satisfied, and $\gamma\mapsto\dist(\theta,\Theta^*(\ps{\gamma}))$ is continuous on $\Gamma$. Hence, by \citet[Proposition 5.11]{Rockafellar_Wets2005aBK}, $\gamma\mapsto\Theta^*(\ps{\gamma})$ is both upper and lower hemicontinuous on $\Gamma$, and hence is continuous at $\gamma=0$.
\end{proof}

\subsubsection{Verification of Assumptions}\label{app:verify:ass:entry}
Assumptions \ref{as:newey_modified1_new}, \ref{as:newey_modified2} \ref{ass:equic}, and the requirement in \eqref{ass:consistent_cov}, are high-level conditions to be verified in each application of our method. 
Below we do so for the entry game example, under some regularity conditions. 
We first provide some notation that will be useful throughout.
Let $\|\ps{y|x}\|_{\mathcal H}=\sup_{y\in \mathcal Y}\sup_{x\in\cX}|\ps{y|x}(y|x)|$.
For $j=1,2,$ we let:
	\begin{align}
\hspace{-.69cm}Z_j(x;\ps{y|x})\equiv \ps{y|x}((1,0)|x)\eta_1(\theta;x)
-(\ps{y|x}((1,0)|x)+\ps{y|x}((0,1)|x))\eta_{j+1}(\theta;x),\label{eq:defZ}
	\end{align}
with the functions $\eta_1(\theta;x),\eta_2(\theta;x),\eta_3(\theta;x)$ defined in \eqref{eq:eta1}-\eqref{eq:eta3}. As $\eta_2(\theta;x)\ge\eta_3(\theta;x)$, $Z_1(x;\ps{y|x})\le Z_2(x;\ps{y|x})$.
We use the functions $Z_1,Z_2$ to define indicators $\mathbb I_\ell$ that re-express the sets $\Theta_\ell,\ell=1,2,3,$ in \eqref{eq:def_Theta1}-\eqref{eq:def_Theta3}:
\begingroup
\allowdisplaybreaks\begin{align}
    \mathbb I_1(x;\ps{y|x})
	&=1\big\{Z_1(x;\ps{y|x})\le 0\}1\{Z_2(x;\ps{y|x})\ge 0\}\label{eq:defI1}\\
    \mathbb I_2(x;\ps{y|x})
	&=1\{Z_1(x;\ps{y|x})>0\}\label{eq:defI2}\\
    \mathbb I_3(x;\ps{y|x})
	&=1\{Z_2(x;\ps{y|x})<0\}.\label{eq:defI3}
\end{align}
\endgroup
One may rewrite the score functions in \eqref{eq:eg_score1}-\eqref{eq:eg_score4} as
\begin{align*}
\score{\theta}{y|x;\ps{y|x}}&=\tfrac{\nabla_\theta\f(S_{\{y\}|x;\theta})}{\f(S_{\{y\}|x;\theta})},\quad y\in\{(0,0),(1,1)\}\\ 
\score{\theta}{(0,1)|x;\ps{y|x}}&=
\tfrac{\nabla_\theta \eta_1(\theta;x)}{ \eta_1(\theta;x)}\mathbb I_1(x;\ps{y|x})+
  \tfrac{\nabla_\theta [\eta_1(\theta;x)-\eta_2(\theta;x)]}{\eta_1(\theta;x)- \eta_2(\theta;x)}\mathbb I_2(x;\ps{y|x})+
  \tfrac{\nabla_\theta [\eta_1(\theta;x)-\eta_3(\theta;x)]}{\eta_1(\theta;x)- \eta_3(\theta;x)}\mathbb I_3(x;\ps{y|x}),\\ 
\score{\theta}{(1,0)|x;\ps{y|x}}&=
  \tfrac{\nabla_\theta \eta_1(\theta;x)}{ \eta_1(\theta;x)}\mathbb I_1(x;\ps{y|x})+
  \tfrac{\nabla_\theta \eta_2(\theta;x)}{ \eta_2(\theta;x)}\mathbb I_2(x;\ps{y|x})+
  \tfrac{\nabla_\theta \eta_3(\theta;x)}{ \eta_3(\theta;x)}\mathbb I_3(x;\ps{y|x}).
\end{align*}
For any vector $a=(a_1,\dots, a_{d_\ex})$, define the differential operator by $D^{|a|}=\tfrac{\partial^{|a|}}{\partial x_1^{a_1}\cdots\partial x_d^{a_{d_\ex}}}$,
where $|a|=\sum_i^{d_\ex} a_i$. Then, for a function $h:\cX\to\mathbb R$, let
\begin{align*}
\|h\|_{\infty,\alpha}=\max_{|a|\le [\alpha]}\sup_x|D^a h(x)|+\max_{|a|=[\alpha]}\sup_{x\ne x'}
\tfrac{|D^a h(x)-D^a h(x')|}{\|x-x'\|^{\alpha-[\alpha]}}.	
\end{align*}
Let $\mathcal C^\alpha_M(\cX)$ be the set of continuous functions $h:\cX\to\mathbb R$ with $\|h\|_{\infty,\alpha}\le M$.
We next provide regularity conditions under which we verify Assumptions \ref{as:newey_modified1_new} and \ref{as:newey_modified2} \ref{ass:equic}.
\vspace{-.5cm}

\begin{assumptionA}\label{ass:verify_entry_game}
For the entry game model in Example \ref{example:CT},	
\begin{enumerate}[label=(\roman*)]
\item \label{ass:entry_bounded_derivative} There exists $C>0$ s.t. $\|\nabla_\theta\eta_j(\theta;x)\|\le C,~j=1,\dots,3,$ for all $x\in\cX$.
\item \label{ass:p_yx_bounded_below} There exists $c>0$ s.t. $\mathcal H=\{\ps{y|x}:\cX\to [0,1]^{\cY}: \ps{y|x}(y|x)\ge c,\forall (y,x)\in\cY\times\cX\}$.
\item \label{ass:entry_smooth_CDF} If $X$ has a component with continuous distribution, the probability density function (p.d.f.) of $Z_j|\ex_d$, for $j=1,2$,  is uniformly bounded on the support of $Z_j|\ex_d$, where $X_d$ denotes the subvector of $X$ containing discrete covariates with finite support. 
If there are no discrete covariates, the restriction is on the unconditional p.d.f. of $Z_j$.
\end{enumerate}
\end{assumptionA}
\begin{assumptionA}\label{as:eg_support}
\begin{enumerate}[label=(\alph*)]
\item \label{as:eg_bounds} $\E\left[\left\|\tfrac{\nabla_\theta\f(S_{\{y\}|x;\theta})}{\f(S_{\{y\}|x;\theta})}\right\|^2\right]\le C$ for $y=(0,0),(1,1)$ for some $0<C<\infty$.
\item\label{ass:support}One of the following conditions hold:
\begin{enumerate}[label=(\roman*)]
	\item\label{ass:support:discrete} $X$ is a vector of discrete random variables and $\cX\subset\mathbb R^{d_\ex}$ is a finite set. 
	\item\label{ass:support:continuous} $X$ is a vector of continuous random variables and $\cX\subset\mathbb R^{d_\ex}$ is a bounded, convex set with nonempty interior. 
	For some $c>0$,  $M>0$, and $\alpha>d_\ex$, 
	\begin{align*}
	\hspace{-1.35cm}\mathcal H=\{\ps{y|x}:\cX\to [0,1]^{|\cY|}:\ps{y|x}(y|\cdot)\in \mathcal C^\alpha_M(\cX), y\in\cY,~
	\ps{y|x}(y|x)\ge c>0,
 \forall (y,x)\in\cY\times\cX\}
	\end{align*}	
	\item\label{ass:support:mixed} $\ex=(\ex_c^\top,\ex_d^\top)^\top$ consists of subvectors $\ex_c$ and $\ex_d$, where $\ex_c$ is continuously distributed and $\ex_d$ is discretely distributed.
	 $\cX=\cX_c\times\cX_d\subset\mathbb R^{d_\ex}$, where $\cX_c\subset \mathbb R^{d_{\ex_c}}$ is a bounded convex set with nonempty interior, and $\mathcal \ex_d\subset\mathbb R^{d_{\ex_d}}$ is a finite set. 
	For some $c>0$, $M>0$, $\alpha>d_{\ex_c}$,  Lipschitz functions $\phi_k,k=1,\dots,|\cY|$,  and  some functions $\ell_c$ and $\ell_d$,
	\begin{multline*}
		\hspace{-1.35cm}\mathcal H=\{\ps{y|x}:\cX\to [0,1]^{|\cY|}:\ps{y|x}(y_k|x)=\phi_k\big(\ell_c(y_k|x_c),\ell_d(y_k|x_d)\big),~\ell_c(y_k|\cdot)\in \mathcal C^\alpha_M(\cX_c), \\
		\hspace{-1.05cm} -M\le \ell_d(y_k|x_d)\le M, \forall x_d\in\cX_d, ~ k=1,\dots,d_Y,
		~\ps{y|x}(y|x)\ge c>0,\forall (y,x)\in\cY\times\cX\}.
		\end{multline*}
\end{enumerate}
\end{enumerate}
\end{assumptionA}

\begin{remark}\rm
Assumption \ref{ass:verify_entry_game}\ref{ass:entry_bounded_derivative} is satisfied, for example, when the vector $\eu$ has a multivariate Normal distribution, provided the correlation among any two of its entries is bounded away from one (in absolute value).
In Assumption \ref{as:eg_support}\ref{ass:support}\ref{ass:support:mixed}, we assume that $\ps{y|x}$ combines a function of continuous covariates $\ex_c$ with a function of discrete covariates $\ex_d$ using a Lipschitz function, which covers many transformations of interest \citep[see, e.g.,][p.~192]{Vaart:1996wk}. More general transformations can be allowed for, as far as one may ensure that the metric entropy of $\mathcal H$ can be controlled properly.	
\end{remark}
\begin{table}[tbp]
  \caption{Values of $\Delta(x;\ps{y|x},\ps{0,y|x})$ when $\mathbb I(x;\ps{y|x})\neq\mathbb I(x;\ps{0,y|x})$}
  \label{tab:modulus}
  \begin{center}
  \scalebox{0.7}{
  \begin{tabular}{ccl}
  \hline
  $\mathbb I(x;\ps{y|x})$ & $\mathbb I(x;\ps{0,y|x})$ & $\Delta(x;\ps{y|x},\ps{0,y|x})$ \\
  \hline
  (1,0,0) & (0,1,0) &
    \multirow{2}{*}{$
      \Big\|
      \ps{0,y|x}((1,0)|x)\Big(
        \tfrac{\nabla_\theta\eta_2(\theta;x)}{\eta_2(\theta;x)}
        -\tfrac{\nabla_\theta\eta_1(\theta;x)}{\eta_1(\theta;x)}
      \Big)
      +\ps{0,y|x}((0,1)|x)\Big(
        \tfrac{\nabla_\theta \eta_1(\theta;x)-\nabla_\theta\eta_2(\theta;x)}
              {\eta_1(\theta;x)- \eta_2(\theta;x)}
        -\tfrac{\nabla_\theta\eta_1(\theta;x)}{\eta_1(\theta;x)}
      \Big)
      \Big\|
    $} \\[1ex]
  (0,1,0) & (1,0,0) & \\[2ex]
  (1,0,0) & (0,0,1) &
    \multirow{2}{*}{$
      \Big\|
      \ps{0,y|x}((1,0)|x)\Big(
        \tfrac{\nabla_\theta\eta_3(\theta;x)}{\eta_3(\theta;x)}
        -\tfrac{\nabla_\theta\eta_1(\theta;x)}{\eta_1(\theta;x)}
      \Big)
      +\ps{0,y|x}((0,1)|x)\Big(
        \tfrac{\nabla_\theta \eta_1(\theta;x)-\nabla_\theta\eta_3(\theta;x)}
              {\eta_1(\theta;x)- \eta_3(\theta;x)}
        -\tfrac{\nabla_\theta\eta_1(\theta;x)}{\eta_1(\theta;x)}
      \Big)
      \Big\|
    $} \\[1ex]
  (0,0,1) & (1,0,0) & \\[2ex]
  (0,1,0) & (0,0,1) &
    \multirow{2}{*}{$
      \Big\|
      \ps{0,y|x}((1,0)|x)\Big(
        \tfrac{\nabla_\theta\eta_2(\theta;x)}{\eta_2(\theta;x)}
        -\tfrac{\nabla_\theta\eta_3(\theta;x)}{\eta_3(\theta;x)}
      \Big)
      +\ps{0,y|x}((0,1)|x)\Big(
        \tfrac{\nabla_\theta \eta_1(\theta;x)-\nabla_\theta\eta_2(\theta;x)}
              {\eta_1(\theta;x)- \eta_2(\theta;x)}
        -\tfrac{\nabla_\theta \eta_1(\theta;x)-\nabla_\theta\eta_3(\theta;x)}
              {\eta_1(\theta;x)- \eta_3(\theta;x)}
      \Big)
      \Big\|
    $} \\[1ex]
  (0,0,1) & (0,1,0) & \\[2ex]
  \hline
  \end{tabular}}
  \end{center}
\end{table}
\vspace{-.5cm}
\begin{proposition}
\label{lemma:pathwise_game}
Suppose Assumptions \ref{ass:finite_support} and \ref{ass:verify_entry_game} hold for the entry game model in Example \ref{example:CT}.
Then Assumption \ref{as:newey_modified1_new} also holds.
\end{proposition}
\vspace{-.5cm}
\begin{proof}
Recall $m_\theta(x;\ps{y|x})\equiv \E[\score{\theta}{\ey|\ex;\ps{y|x}}|\ex=x]=\sum_{y\in \cY}\ps{0,y|x}(y|x)\score{\theta}{y,x;\ps{y|x}}$. 
For $\ps{y|x},\ps{0,y|x}\in \mathcal H$, we aim to bound $\E[\|m_\theta(X;\ps{y|x})-m_\theta(X;\ps{0,y|x})\|]$. 
The score depends on $\ps{y|x}$ only through $\mathbb I(x;\ps{y|x})=(\mathbb I_1(x;\ps{y|x}),\mathbb I_2(x;\ps{y|x})$, $\mathbb I_3(x;\ps{y|x}))$. 
Hence, $\Delta(x;\ps{y|x},\ps{0,y|x})\equiv \|m_\theta(x;\ps{y|x})-m_\theta(x;\ps{0,y|x})\|\neq 0$ only if
$\mathbb I(x;\ps{y|x})\ne \mathbb I(x;\ps{0,y|x})$. The values of $\Delta(x;\ps{y|x},\ps{0,y|x})$ are given in Table \ref{tab:modulus}.
We consider two subcases (i) $X$ is discrete and $\cX$ is finite, and (ii) $X$ contains a continuously distributed variable (the case where all components of $X$ are continuously distributed is treated as a special case of the latter).

\noindent
\textbf{(i) Discrete $X$:}
Let $\cX$ be a finite set, $\cX_0=\{x\in\cX: Z_1(x;\ps{0,y|x})\ne 0, Z_2(x;\ps{0,y|x})\ne 0\}$, and $c\equiv \min_{x\in\cX_0}\min_{j=1,2}|Z_j(x;\ps{0,y|x})|.$ 
By Lemma \ref{lem:eg_separated}, $\Delta(x;\ps{y|x},\ps{0,y|x})=0$ for all $x\in\cX_0$ and $\ps{y|x}$ such that $\|\ps{y|x}-\ps{0,y|x}\|_{\mathcal H}\le\delta$ for all $\delta\le c/4$, hence they do not contribute to the $L^1$-norm of $\Delta(\cdot;\ps{y|x},\ps{0,y|x})$.
Take $x\in\cX^c_0$. 
Suppose, e.g., that $0=Z_1(x;\ps{0,y|x})<Z_2(x;\ps{0,y|x})$. 
To have $\Delta(x;\ps{y|x},\ps{0,y|x})\neq 0$, let $\ps{y|x}$ be such that $\|\ps{y|x}-\ps{0,y|x}\|_{\mathcal H}\le \delta$ and $Z_1(x;\ps{y|x})>0$. 
This yields  $\mathbb I(x,\ps{y|x})=(0,1,0)$ and $\mathbb I(x,\ps{0,y|x})=(1,0,0)$. 
Per Table \ref{tab:modulus},
\vspace{-.5cm}
\begin{multline}
\Delta(x;\ps{y|x},\ps{0,y|x})=\big\|\ps{0,y|x}((1,0)|x)\tfrac{\nabla_\theta\eta_2(\theta;x)}{\eta_2(\theta;x)}+\ps{0,y|x}((0,1)|x)\tfrac{\nabla_\theta \eta_1(\theta;x)-\nabla_\theta\eta_2(\theta;x)}{\eta_1(\theta;x)- \eta_2(\theta;x)}\\-[\ps{0,y|x}((1,0)|x)+\ps{0,y|x}((0,1)|x)]\tfrac{\nabla_\theta\eta_1(\theta;x)}{\eta_1(\theta;x)}\big\|.\label{eq:difference1}
\end{multline}
Note that $Z_1(x;\ps{0,y|x})=0$ is equivalent to 
\vspace{-.5cm}
\begin{align}
\ps{0,y|x}((1,0)|x)=[\ps{0,y|x}((1,0)|x)+\ps{0,y|x}((0,1)|x)]\tfrac{\eta_2(\theta;x)}{\eta_1(\theta;x)}.\label{eq:Z1eq1}
\end{align}
Since $\tfrac{\ps{0,y|x}((0,1)|x)}{\ps{0,y|x}((1,0)|x)+\ps{0,y|x}((0,1)|x)}=1-\tfrac{\ps{0,y|x}((1,0)|x)}{\ps{0,y|x}((1,0)|x)+\ps{0,y|x}((0,1)|x)}$, 
we obtain
\vspace{-.35cm}
\begin{align}
\ps{0,y|x}((0,1)|x)=[\ps{0,y|x}((1,0)|x)+\ps{0,y|x}((0,1)|x)]\tfrac{\eta_1(\theta;x)-\eta_2(\theta;x)}{\eta_1(\theta;x)}.	\label{eq:Z1eq2}
\end{align}
Substituting \eqref{eq:Z1eq1}-\eqref{eq:Z1eq2} into \eqref{eq:difference1} yields
$\Delta(x;\ps{y|x},\ps{0,y|x})=
\big\|[\ps{0,y|x}((1,0)|x)+\ps{0,y|x}((0,1)|x)]\big(\tfrac{\nabla_\theta\eta_2(\theta;x)}{\eta_1(\theta;x)}+\tfrac{\nabla_\theta \eta_1(\theta;x)-\nabla_\theta\eta_2(\theta;x)}{\eta_1(\theta;x)}-\tfrac{\nabla_\theta\eta_1(\theta;x)}{\eta_1(\theta;x)}\big)\big\|=0$.
A similar argument can be applied to $x\in \cX_0^c$ such that $Z_1(x;\ps{0,y|x})<Z_2(x;\ps{0,y|x})=0$. Finally, consider $x\in \cX_0^c$ such that  $Z_1(x;\ps{0,y|x})=Z_2(x;\ps{0,y|x})=0$. 
This occurs only if $\eta_2(\theta;x)=\eta_3(\theta;x)$. 
Hence, $Z_1(x;\ps{y|x})=Z_2(x;\ps{y|x})$ for any $\ps{y|x}$. 
It then suffices to consider only one of $Z_j$'s. 
For example let $\ps{y|x}$ be such that $\|\ps{y|x}-\ps{0,y|x}\|_{\mathcal H}\le \delta$ and $Z_1(x;\ps{y|x})>0$. 
Then, the same  analysis as above leads to $\Delta(x;\ps{y|x},\ps{0,y|x})=0$.
Therefore, for all $\ps{y|x}$ such that $\|\ps{y|x}-\ps{0,y|x}\|_{\mathcal H}\le \delta$ for a sufficiently small $\delta$, the pathwise derivative is 0, as
\vspace{-.35cm}
\begin{multline*}
\big\Vert \E\big[m_\theta(X;p'_{y|x})- m_\theta(X;\ps{0,y|x})\big]\big\Vert
\le \E\big[\|m_\theta(X;p'_{y|x})- m_\theta(X;\ps{0,y|x})\|\big]\\
=\sum_{x\in\cX_0}\ps{0,x}(x)\Delta(x;\ps{y|x},\ps{0,y|x})+\sum_{x\in\cX^c_0}\ps{0,x}(x)\Delta(x;\ps{y|x},\ps{0,y|x})=0.
\end{multline*}

\noindent
\textbf{(ii) $X$ contains a continuously distributed variable:}

Let $\ex_d$ be a subvector of $X$ containing discrete covariates. 
Recall that $\Delta(x,\ps{y|x},\ps{0,y|x})\ne 0$ when $\mathbb I(x;\ps{y|x})\ne \mathbb I(x;\ps{0,y|x})$. 
This occurs when $\sgn(Z_j(x;\ps{y|x}))\ne\sgn(Z_j(x;\ps{0,y|x}))$ for some $j$. 
By \eqref{eq:defZ} and $\sup_{x\in\cX}|\eta_j(\theta;x)|\le 1$, $\|\ps{y|x}-\ps{0,y|x}\|_{\mathcal H}\le \delta$ implies 
\begin{align*}
\sup\nolimits_{x\in\cX}|Z_j(x;\ps{y|x})-Z_j(x;\ps{0,y|x})|\le 3\delta,~j=1,2
\end{align*}
Therefore, if $\sgn(Z_j(x;\ps{y|x}))\ne\sgn(Z_j(x;\ps{0,y|x}))$ and $\|\ps{y|x}-\ps{0,y|x}\|_{\mathcal H}\le \delta$, one must have $|Z_j(x;\ps{0,y|x})|\le 3\delta.$ 
Hence, the pathwise derivative is again zero because
\begin{align*}
&\Big\Vert E\Big[m_\theta(X;\ps{y|x})- m_\theta(X;\ps{0,y|x})\Big]\Big\Vert
\le \E\Big[\|m_\theta(X;\ps{y|x})- m_\theta(X;\ps{0,y|x})\|\Big]\\
&\le\E\Big[\Delta(X;\ps{y|x},\ps{0,y|x})(1\{-3\delta\le Z_1(X;\ps{0,y|x})\le 3\delta\}+1\{-3\delta\le Z_2(X;\ps{0,y|x})\le 3\delta\})\Big]\\
&\stackrel{(i)}{\le} K\delta \E\Big[\int 1\{-3\delta\le z_1\le 3\delta\}f_{Z_1|\ex_d}(z_1)dz_1+\int 1\{-3\delta\le z_2\le 3\delta\}f_{Z_2|\ex_d}(z_2)dz_2\Big]\stackrel{(ii)}{\le} c\delta^2,
\end{align*}
where inequality (i) follows by Lemma \ref{lem:boundDelta} and the law of iterated expectations, and inequality (ii) follows by Assumption \ref{ass:verify_entry_game} \ref{ass:entry_smooth_CDF} for $0<c<\infty$ some constant.
\end{proof}

\begin{lemma}\label{lem:eg_separated}
 Let $\cX$ be a finite set, and let $\cX_0=\{x\in\cX: Z_1(x;\ps{0,y|x})\ne 0, Z_2(x;\ps{0,y|x})\ne 0\}$. Let $c\equiv \min_{x\in\cX_0}\min_{j=1,2}|Z_j(x;\ps{0,y|x})|.$ Then, $\Delta(x;\ps{y|x},\ps{0,y|x})=0$ for any $x\in\cX_0$ and $\ps{y|x}$ such that $\|\ps{y|x}-\ps{0,y|x}\|_{\cH}\le c/4$.
\end{lemma}
\begin{proof}
Take $x\in\cX_0.$
Suppose $Z_1(x;\ps{0,y|x})\ge c>0$ so that $\mathbb I(x;\ps{0,y|x})=(0,1,0)$. Let  $\ps{y|x}$ satisfy $\|\ps{y|x}-\ps{0,y|x}\|_{\cH}\le c/4$. Then, by \eqref{eq:defZ} and $|\eta_j(x;\theta)|\le 1$, and the triangle inequality,
$Z_1(x;\ps{y|x})\ge Z_1(x;\ps{0,y|x})-\tfrac{3}{4}c\ge \tfrac{1}{4}c>0$,
implying $\mathbb I(x;\ps{y|x})=(0,1,0)$. From Table \ref{tab:modulus}, $\Delta(x;\ps{y|x},\ps{0,y|x})=0$. Other cases can be analyzed similarly.
\end{proof}

\begin{lemma}\label{lem:boundDelta}
Suppose Assumptions \ref{ass:finite_support} and \ref{ass:verify_entry_game} hold for the entry game model in Example \ref{example:CT}.
For $\delta>0$,  let $\ps{y|x}$ be such that $\|\ps{y|x}-\ps{0,y|x}\|_{\mathcal H}\le \delta$. Then, there exists $0<K<\infty$ such that for all $x\in\cX$,
$\Delta(x;\ps{y|x},\ps{0,y|x})\le K\delta$.
\end{lemma}
\begin{proof}
From Table \ref{tab:modulus},  $\Delta(x;\ps{y|x},\ps{0,y|x})=0$ when $\mathbb I(x;\ps{y|x})=\mathbb I(x;\ps{0,y|x})$. Therefore, we focus on cases with $\mathbb I(x;\ps{y|x})\ne \mathbb I(x;\ps{0,y|x})$ below. Consider the case where $\mathbb I(x;\ps{y|x})=(0,1,0)$ and $\mathbb I(x;\ps{0,y|x})=(1,0,0)$. By \eqref{eq:defI1}-\eqref{eq:defI3}, this occurs when
\begin{align}
	Z_1(x;\ps{0,y|x})\le 0,~Z_1(x;\ps{y|x})>0.\label{eq:boundDelta1}
\end{align}
Furthermore, by \eqref{eq:defZ} and $\sup_{x\in\cX}|\eta_j(\theta;x)|\le 1$, $\|\ps{y|x}-\ps{0,y|x}\|_{\mathcal H}\le \delta$ implies 
\begin{align}
\sup_{x\in\cX}|Z_1(x;\ps{y|x})-Z_1(x;\ps{0,y|x})|\le 3\delta.\label{eq:boundDelta2}
\end{align}
Combining \eqref{eq:boundDelta1}-\eqref{eq:boundDelta2} yields
$-3\delta\le Z_1(x;\ps{0,y|x})\le 0$.
By \eqref{eq:defZ} and $\eta_j(\theta;x)\ge c$,
\begin{multline}
(\ps{0,y|x}((1,0)|x)+\ps{0,y|x}((0,1)|x))\tfrac{\eta_{2}(\theta;x)}{\eta_1(\theta;x)}-\tfrac{3}{c}\delta	
\\\le \ps{0,y|x}((1,0)|x)\le (\ps{0,y|x}((1,0)|x)+\ps{0,y|x}((0,1)|x))\tfrac{\eta_{2}(\theta;x)}{\eta_1(\theta;x)}\label{eq:boundDelta3}
\end{multline}
Using Assumption \ref{ass:verify_entry_game} \ref{ass:p_yx_bounded_below}, this may also be written as 
\begin{align*}
\tfrac{\eta_{2}(\theta;x)}{\eta_1(\theta;x)}-\tfrac{3}{c(\ps{0,y|x}((1,0)|x)+\ps{0,y|x}((0,1)|x))}\delta\le\tfrac{\ps{0,y|x}((1,0)|x)}{(\ps{0,y|x}((1,0)|x)+\ps{0,y|x}((0,1)|x))}\le \tfrac{\eta_{2}(\theta;x)}{\eta_1(\theta;x)}.
\end{align*}
Since $\tfrac{\ps{0,y|x}((0,1)|x)}{(\ps{0,y|x}((1,0)|x)+\ps{0,y|x}((0,1)|x))}=1-\tfrac{\ps{0,y|x}((1,0)|x)}{(\ps{0,y|x}((1,0)|x)+\ps{0,y|x}((0,1)|x))}$, we obtain
\begin{align*}
1-\tfrac{\eta_{2}(\theta;x)}{\eta_1(\theta;x)}	\le\tfrac{\ps{0,y|x}((0,1)|x)}{(\ps{0,y|x}((1,0)|x)+\ps{0,y|x}((0,1)|x))}
\le 1-\tfrac{\eta_{2}(\theta;x)}{\eta_1(\theta;x)}	+\tfrac{3}{c(\ps{0,y|x}((1,0)|x)+\ps{0,y|x}((0,1)|x))}\delta.
\end{align*}
This may in turn be written as 
\begin{multline}
(\ps{0,y|x}((1,0)|x)+\ps{0,y|x}((0,1)|x))\tfrac{\eta_1(\theta;x)-\eta_{2}(\theta;x)}{\eta_1(\theta;x)}	\\
\le\ps{0,y|x}((0,1)|x)
\le(\ps{0,y|x}((1,0)|x)+\ps{0,y|x}((0,1)|x))\tfrac{\eta_1(\theta;x)-\eta_{2}(\theta;x)}{\eta_1(\theta;x)}+\tfrac{3}{c}\delta.	\label{eq:boundDelta6}
\end{multline} 
By \eqref{eq:boundDelta3} and \eqref{eq:boundDelta6}, let us write
\begin{align}
	\ps{0,y|x}((1,0)|x)&=(\ps{0,y|x}((1,0)|x)+\ps{0,y|x}((0,1)|x))\tfrac{\eta_{2}(\theta;x)}{\eta_1(\theta;x)}+r_{(1,0)}(x)\label{eq:boundDelta7}\\
	\ps{0,y|x}((0,1)|x)&=(\ps{0,y|x}((1,0)|x)+\ps{0,y|x}((0,1)|x))\tfrac{\eta_1(\theta;x)-\eta_{2}(\theta;x)}{\eta_1(\theta;x)}+r_{(0,1)}(x),\label{eq:boundDelta8}
\end{align}
where $r_{(1,0)}(x)\in [-3\delta/c,0]$ and $r_{(0,1)}(x)\in [0,3\delta/c]$ for all $x\in\cX$.
From Table \ref{tab:modulus}, the value of $\Delta(x;\ps{y|x},\ps{0,y|x})$ when $\mathbb I(x;\ps{y|x})=(0,1,0)$ and $\mathbb I(x;\ps{0,y|x})=(1,0,0)$ is
\begin{multline}
\Delta(x;\ps{y|x},\ps{0,y|x})=\Big\|\ps{0,y|x}((1,0)|x)\tfrac{\nabla_\theta\eta_2(\theta;x)}{\eta_2(\theta;x)}+\ps{0,y|x}((0,1)|x)\tfrac{\nabla_\theta \eta_1(\theta;x)-\nabla_\theta\eta_2(\theta;x)}{\eta_1(\theta;x)- \eta_2(\theta;x)}\\-[\ps{0,y|x}((1,0)|x)+\ps{0,y|x}((0,1)|x)]\tfrac{\nabla_\theta\eta_1(\theta;x)}{\eta_1(\theta;x)}\Big\|.\label{eq:boundDelta9}
\end{multline}
By \eqref{eq:boundDelta7}-\eqref{eq:boundDelta8}, the terms inside the norm in \eqref{eq:boundDelta9} can therefore be written as
\begin{multline*}
[\ps{0,y|x}((1,0)|x)+\ps{0,y|x}((0,1)|x)]\Bigg(\tfrac{\nabla_\theta\eta_2(\theta;x)}{\eta_1(\theta;x)}+\tfrac{\nabla_\theta \eta_1(\theta;x)-\nabla_\theta\eta_2(\theta;x)}{\eta_1(\theta;x)}-\tfrac{\nabla_\theta\eta_1(\theta;x)}{\eta_1(\theta;x)}\Bigg)\\
+\tfrac{\nabla_\theta\eta_2(\theta;x)}{\eta_2(\theta;x)}r_{(1,0)}(x)+\tfrac{\nabla_\theta \eta_1(\theta;x)-\nabla_\theta\eta_2(\theta;x)}{\eta_1(\theta;x)-\eta_2(\theta;x)}r_{(0,1)}(x).
\end{multline*}
By the triangle inequality, $\eta_j(\theta;x)\ge c$, and Assumption \ref{ass:verify_entry_game} \ref{ass:entry_bounded_derivative}, we obtain
\begin{align*}
\Delta(x;\ps{y|x},\ps{0,y|x})&\le\|\nabla_\theta\eta_2(\theta;x)\|\Big|\tfrac{r_{(1,0)}(x)}{\eta_2(\theta;x)}\Big|+( \|\nabla_\theta \eta_1(\theta;x)\|+\|\nabla_\theta\eta_2(\theta;x)\|)\Big|\tfrac{r_{(0,1)}(x)}{\eta_1(\theta;x)-\eta_2(\theta;x)}\Big|\\
&\le\tfrac{3C}{c^2}\delta+\tfrac{6C}{c^2}\delta,
\end{align*}
which establishes the claim of the lemma for $\mathbb I(x;\ps{y|x})=(0,1,0)$ and $\mathbb I(x;\ps{0,y|x})=(1,0,0)$. 
The other cases can be analyzed similarly.
\end{proof}

We next establish stochastic equicontinuity for the empirical process
\begin{align*}
\mathbb G_n(p)=\tfrac{1}{\sqrt n}\sum_{i=1}^n \left(\score{\theta}{\ey_i|\ex_i;p}-\E[\score{\theta}{\ey_i|\ex_i;p}]\right),~~p\in\mathcal H
\end{align*}
In the definition of $\mathbb G_n$, the structural parameter $\theta$ is fixed. Hence, the function class $\mathcal F=\{f(y,x):f(y,x)=\score{\theta}{y|x;p},p\in\mathcal H\}$ is defined by mixing and matching $p$ with fixed functions such as $\eta_j(x,\theta)$. For example, we may write $\score{\theta}{(1,0)|x;p}$ as
\begin{align*}
    \score{\theta}{(1,0)|x;\ps{0,y|x}}=\sum_{j=1}^3\tfrac{\nabla_\theta\eta_j(\theta;x)}{\eta_j(\theta;x)}\mathbb I_j(x;\ps{y|x}),
\end{align*}
where $\mathbb I_j(x;\ps{y|x}),j=1,2,3$, are defined in \eqref{eq:defI1}-\eqref{eq:defI3}.
\begin{proposition}
\label{lemma:donsker_game}
Suppose Assumptions \ref{ass:finite_support}, \ref{ass:verify_entry_game}, and \ref{as:eg_support} hold for the entry game model in Example \ref{example:CT}.
Then Assumption \ref{as:newey_modified2} \ref{ass:equic} also holds.
\end{proposition}
\begin{proof}
Let 
{\small
\begin{align*}
\mathcal F_{y} &=\Big\{f:f(w;p)=\tfrac{\nabla_\theta\f(S_{\{y\}|x;\theta})}{\f(S_{\{y\}|x;\theta})}	\Big\},\quad y\in\{(0,0),(1,1)\}\\
	\mathcal F_{(0,1)}&=\Big\{f:f(w;p)=\tfrac{e_l'\nabla_\theta\eta_1(\theta;x)}{\eta_1(\theta;x)}\mathbb I_1(x;p)+\sum_{j=2}^3\tfrac{e_l'(\nabla_\theta\eta_1(\theta;x)-\nabla_\theta\eta_j(\theta;x))}{\eta_1(\theta;x)-\eta_j(\theta;x)}\mathbb I_j(x;p),l=1,\dots,d_\theta,~ p\in\mathcal H\Big\},\\
	\mathcal F_{(1,0)}&=\Big\{f:f(w;p)=\sum_{j=1}^3\tfrac{e_l'\nabla_\theta\eta_j(\theta;x)}{\eta_j(\theta;x)}\mathbb I_j(x;p),~l=1,\dots,d_\theta,~ p\in\mathcal H\Big\},
\end{align*}	}
where for each $l$, $e_l$ denotes the $l$-th standard basis vector in $\mathbb R^{d_\theta}$.

The score satisfies $\score{\theta}{\cdot;\ps{y|x}}\in\mathcal F\equiv \sum_{\bar y\in\mathcal Y}\mathcal F_{\bar y}\cdot 1\{y=\bar y\}$. 
In view of Theorem 2.10.6 (and Examples 2.10.7 and 2.10.10) in \citet{Vaart:1996wk}, to verify Assumption \ref{as:newey_modified2} \ref{ass:equic} it suffices to show that $\mathcal F_{\bar y}$ is $P$-Donsker for each $\bar y$. 
The $P$-Donskerness of $\mathcal F_{\bar y}$ for $\bar y=(0,0), (1,1)$ follows from each set being a singleton and Assumption \ref{as:eg_support}  \ref{as:eg_bounds}. Below, we show   $\mathcal F_{(1,0)}$ is $P$-Donsker. The analysis for $\mathcal F_{(0,1)}$ is similar and is therefore omitted.
	
\noindent \textbf{Discrete $X$:}	
First, suppose that $X$ is a vector of discrete random variables and Assumption \ref{as:eg_support}\ref{ass:support}\ref{ass:support:discrete} holds. We show that  $\mathcal F_{(1,0)}$  is a Vapnik-Chervonenkis (VC) class, which satisfies Pollard's uniform entropy condition. 
As $\cY\times\cX$ is finite, $\mathcal H$ is finite-dimensional.
By \citet[Lemma 2.6.15]{Vaart:1996wk}, this class has VC-index $V(\mathcal H)\le d_Y\times d_\ex+2$, and hence $\mathcal H$ is a VC-class.  The finite set of functions $\mathcal E=\{\eta_j(\theta,\cdot),\tfrac{\partial}{\partial\theta_k}\eta_j(\theta,\cdot),j=1,\dots,3, k=1,\dots,d_\theta\}$ is also a VC-class. As $\mathcal F_{(1,0)}$ collects functions that can be expressed as combinations of functions from $\cH$ and $\mathcal E$ by multiplication, addition, division, and composition with an indicator function $1\{\cdot>0\}$, $\mathcal F_{(1,0)}$ is a VC-class \citep[][Lemma 2.6.18]{Vaart:1996wk}.
Assumptions \ref{ass:finite_support} and \ref{ass:verify_entry_game}\ref{ass:entry_bounded_derivative} ensure that there is an envelope (constant) function $F=3C/c$  such that $|f|\le F$ for all $f\in \mathcal F_{(1,0)}$. By Theorem 2.5.2 in \cite{Vaart:1996wk},  $\mathcal F_{(1,0)}$ is a $P$-Donsker class.

\noindent \textbf{Continuous $X$:}
Next, suppose that  $X$ is a vector of continuous random variables and Assumption \ref{as:eg_support}\ref{ass:support}\ref{ass:support:continuous} holds.
We show  $\mathcal F_{(1,0)}$ is $P$-Donsker by verifying the conditions in \citet[Theorem 3]{Chen2003}.
We first show the $L^2$-H\"older continuity of $f\in \mathcal F_{(1,0)}$ in $p$. 
In what follows, let $U_{l,j}=e_l'\nabla_\theta\eta_j(\theta;X)/\eta_j(\theta;X)$, $l=1,\dots,d$, $j=1,\dots,3$ and $u_{l,j}=e_l'\nabla_\theta\eta_j(\theta;x)/\eta_j(\theta;x)$. By the triangle inequality, for some constant $C>0$,
\begin{align}
\sup_{\|p-p'\|_{\mathcal H}\le \delta}|f(w;p')-f(w;p)|^2&\le C\sup_{\|p-p'\|_{\mathcal H}\le \delta}\sum_{j=1}^3 u_{l,j}^2|\mathbb I_j(x;p')-\mathbb I_j(x;p)|,\label{eq:entropy1}
\end{align}
Below, we focus on $u_{l,3}^2|\mathbb I_3(x;p')-\mathbb I_3(x;p)|$, one of the terms in the sum on the right hand side of \eqref{eq:entropy1}. The two other terms can be analyzed similarly. 
For $\delta$ sufficiently small,
\begin{multline*}
\E\Big[\sup_{\|p-p'\|_{\mathcal H}\le \delta}U_{l,3}^2|\mathbb I_3(X;p')-\mathbb I_3(X;p)|\Big]\\
=\E\Big[ \sup_{\|p-p'\|_{\mathcal H}\le \delta} U_{l,3}^2|1\{Z_2(X;p'_{y|x})<0\}-1\{Z_2(X;p_{y|x})<0\}|\Big].
\end{multline*}
By \eqref{eq:defZ} and $\sup_{x\in\cX}|\eta_j(\theta;x)|\le 1$, whenever $\|p'_{y|x}-\ps{y|x}\|_{\mathcal H}\le \delta$, we have
\begin{align}
\sup_{x\in\cX}|Z_j(x;p'_{y|x})-Z_j(x;\ps{y|x})|\le 3\delta,~j=1,2.\label{eq:bound:Z_j}
\end{align}
We next use the argument in \citet[p.~1600]{Chen2003}.
Combining one side of \eqref{eq:bound:Z_j}, with the addition of a non-negative constant, we have $Z_2(x;p_{y|x})-3\delta \le Z_2(x;p'_{y|x}) \le Z_2(x;p'_{y|x})+3\delta$, and hence
\begin{align}
1\{Z_2(x;p_{y|x})-3\delta< 0 \}\ge 1\{Z_2(x;p'_{y|x})< 0\}\ge 1\{ Z_2(x;p'_{y|x})+3\delta< 0\}.\label{eq:ind1}
\end{align}
Similarly,  $Z_2(x;p'_{y|x})-3\delta \le Z_2(x;p_{y|x}) \le Z_2(x;p_{y|x})+3\delta$ implies
\begin{align}
1\{Z_2(x;p'_{y|x})-3\delta< 0 \}\ge 1\{Z_2(x;p_{y|x})< 0\}\ge 1\{ Z_2(x;p_{y|x})+3\delta< 0\}.\label{eq:ind2}
\end{align}
Combining \eqref{eq:ind1}-\eqref{eq:ind2}, for any $p',p$  with $\|p'-p\|_{\mathcal H}\le \delta$,
\begin{align*}
|1\{Z_2(x;p'_{y|x})< 0\}-1\{Z_2(x;p_{y|x})< 0\}|&\le 	1\{Z_2(x;p_{y|x})-3\delta< 0 \}-1\{ Z_2(x;p_{y|x})+3\delta< 0\}\\
&\le 1\{-3\delta< Z_2(x;p_{y|x})<3\delta\}
\end{align*}
where without loss of generality we assumed that $1\{Z_2(x;p_{y|x})-3\delta< 0 \}-1\{ Z_2(x;p_{y|x})+3\delta< 0\}>1\{Z_2(x;p'_{y|x})-3\delta< 0 \}-1\{ Z_2(x;p'_{y|x})+3\delta< 0\}$.
By the argument above, the law of iterated expectations, and Assumptions \ref{ass:verify_entry_game}\ref{ass:entry_bounded_derivative}, \ref{as:eg_support}\ref{ass:support}\ref{ass:support:continuous}, and \ref{as:eg_support}\ref{as:eg_bounds}, 
\begin{multline}
\E\Big[ \sup_{\|p-p'\|_{\mathcal H} \le \delta}  U_{l,3}^2|1\{Z_2(X;p'_{y|x})<0\}-1\{Z_2(X;p_{y|x})<0\}|\Big]\\
\le \E\Big[ U_{l,3}^2 1\{-3\delta< Z_2(X;p_{y|x})<3\delta\}\Big]\\
\le\tfrac{C^2}{c^2}\int 1\{-3\delta< z_2<3\delta\}f_{Z_2}(z_2)d z_2\le K\delta,\label{eq:ind3}
\end{multline}
for some constant $K>0$, where the last inequality follows from Assumption \ref{ass:verify_entry_game} \ref{ass:entry_smooth_CDF}. 
Applying a similar argument to the other two terms in \eqref{eq:entropy1}, one can obtain
\begin{align}
\E\Big[\sup_{\|p-p'\|_{\mathcal H}\le \delta}|f(W;p')-f(W;p)|^2\Big]^{1/2}\le K'\delta^{1/2},	\label{eq:ind4}
\end{align}
for some $K'>0$. Hence $f$ is $L^2$-H\"{o}lder continuous in $p$ with H\"{o}lder exponent $1/2$.

Recall that $\cX$ is a bounded convex subset of $\mathbb R^{d_\ex}$ with nonempty interior. By  Theorem 2.7.1 in \cite{Vaart:1996wk}, $\ln N(\epsilon^{2},\mathcal C^\alpha_M(\cX),\|\cdot\|_\infty)\le K\left(\tfrac{1}{\epsilon}\right)^{2d_\ex/\alpha}$ for some $K>0$. 
Note that $\mathcal H\subset(\mathcal C^\alpha_M(\cX))^{\cY}$ and $|\cY|=4$. For each $y\in \cY$, let $\{p_1(y|\cdot),\dots,p_k(y|\cdot)\}$ be an $\epsilon^2$-cover for $\mathcal C^\alpha_M(\cX)$ with respect to the sup norm. Then, $\{(p_{j_1}((0,0)|\cdot),p_{j_2}((0,1)|\cdot),p_{j_3}((1,0)|\cdot),$ $p_{j_4}((1,1)|\cdot),j_l\in\{1,\dots,k\},l=1,\dots,4\}$ forms an $\epsilon^2$-cover for  $(\mathcal C^\alpha_M(\cX))^{\cY}$ with respect to the maximum of the sup norms. Hence,
\begin{align}
	 N(\epsilon^{2},\mathcal C^\alpha_M(\cX)^{\cY},\|\cdot\|_\infty)\le e^{4K\left(\tfrac{1}{\epsilon}\right)^{2d_\ex/\alpha}},\label{eq:epsilon_cover}
\end{align} 
which in turn implies $\ln N(\epsilon^{2},\mathcal H,\|\cdot\|_\infty)\le 4K\left(\tfrac{1}{\epsilon}\right)^{2d_\ex/\alpha}$.
Since $\alpha>d_\ex$, we have 
\begin{align}
\int_0^\infty \sqrt{\ln N(\epsilon^{2},\mathcal H,\|\cdot\|_{\mathcal H})}d\epsilon<\infty.	\label{eq:ind6}
\end{align}
We can now apply Theorem 3 in \cite{Chen2003}, which ensures that $\mathcal F_{(1,0)}$ is $P$-Donsker.

\noindent \textbf{Mixed $X$:}
Finally, suppose that $X$ contains both continuous and discrete variables and Assumption \ref{as:eg_support}\ref{ass:support}\ref{ass:support:mixed} holds. Again, we use Theorem 3 in \cite{Chen2003}.
We can  argue as in the previous case, but \eqref{eq:ind3} is modified as follows:
\begin{multline}
\E\Big[ \sup_{\|p-p'\|_{\mathcal H} \le \delta}  U_{l,3}^2|1\{Z_2(X;p'_{y|x})<0\}-1\{Z_2(X;p_{y|x})<0\}|\Big]\\
\le \E\Big[ U_{l,3}^2 1\{-3\delta< Z_2(X;p_{y|x})<3\delta\}\Big]\\
\le\tfrac{C^2}{c^2}\E\left[\int 1\{-3\delta< z_2<3\delta\}f_{X_2|\ex_d}(z_2)d z_2\right]\le K\delta,\label{eq:ind7}
\end{multline}
for some constant $K>0$, where the last inequality follows from Assumption \ref{ass:verify_entry_game}\ref{ass:entry_smooth_CDF}. 
Therefore, \eqref{eq:ind4} holds. 

Next we show \eqref{eq:ind6}. 
Recall that $N(\epsilon^{2},\mathcal C^\alpha_M(\cX),\|\cdot\|_\infty)\le e^{K\left(\tfrac{1}{\epsilon}\right)^{2d_\ex/\alpha}}$ for some $K>0$. Furthermore, $x_d\mapsto \ell_d(y_k|x_d)$ belongs to a finite-dmensional space $[-M,M]^{\cX_d}$ with covering number $N(\epsilon^2,[-M,M]^{\cX_d},\|\cdot\|_\infty)\le \big(\tfrac{\sqrt{2M}}{\epsilon}\big)^{2\dim(\cX_d)}$. For each $l$, let $p_{c,1}(y_l|\cdot)$, $\cdots$, $p_{c,N_1}(y_l|\cdot)$ be an $\epsilon^2$-cover of $\mathcal C^\alpha_M(\cX_c)$. Similarly, let $p_{d,1}(y_l|\cdot),\dots,p_{d,N_2}(y_l|\cdot)$ be an $\epsilon^2$-cover of $[-M,M]^{\cX_d}$. Then, for any $\ps{y|x}\in\mathcal H$ and $l\in\{1,\dots,4\}$, there exist $k_1\in\{1,\dots,N_1\}$, $k_2\in\{1,\dots,N_2\}$, and $(\ell_c(y_k|\cdot),\ell_d(y_k|\cdot))\in \mathcal C^\alpha_M(\cX_c)\times [-M,M]^{\cX_d}$ such that
\begin{multline*}
	\sup_{x=(x_c',x_d')'\in \cX_c\times\cX_d}\big|\ps{y_l|x}(y|x)-\phi_k(p_{c,k_1}(y_l|x_c),p_{d,k_1}(y_l|x_d))\big|\\
	=\sup_{x=(x_c',x_d')'\in \cX_c\times\cX_d}\big|\phi_k(p_{c}(y_l|x_c),p_{d}(y_l|x_d))-\phi_k(p_{c,k_1}(y_l|x_c),p_{d,k_2}(y_l|x_d))\big|\\
	\le C\max\{\|p_{c}(y_l|\cdot)-p_{c,k_1}(y_l|\cdot)\|_\infty,\|p_{d}(y_l|\cdot)-p_{d,k_2}(y_l|\cdot)\|_\infty\}\le C\epsilon^2,
\end{multline*}
for some $0<C<\infty$ due to the Lipschitz continuity of $\phi_k$.
Therefore $\{(p_{c,k_1}(y_l|\cdot),p_{d,k_2}(y_l|\cdot))_{l=1}^4, $ $k_1\in\{1,\dots,N_1\}, k_2\in\{1,\dots,N_2\},l\in\{1,\dots,4\}\}$ is an $C\epsilon^2$-cover of $\mathcal H.$ Hence,
\begin{align*}
	N(\epsilon^{2},\mathcal H,\|\cdot\|_\infty)\le \Big(N(\epsilon^{2}/C,\mathcal C^\alpha_M(\cX_c),\|\cdot\|_\infty)\times N(\epsilon^2/C,[-M,M]^{\cX_d},\|\cdot\|_\infty)\Big)^4,
\end{align*}
which in turn implies, for some $K'>0$ for all $\epsilon$ small enough,
\begin{align*}
	\ln N(\epsilon^{2},\mathcal H,\|\cdot\|_\infty)\le 4K\left(\tfrac{\sqrt C}{\epsilon}\right)^{2d_{\ex_c}/\alpha} +8\dim(\cX_d) \ln  \big(\tfrac{\sqrt{2M}}{\epsilon}\big)\le K'\left(\tfrac{\sqrt C}{\epsilon}\right)^{2d_{\ex_c}/\alpha}.
\end{align*}
Again, by $\alpha>d_{X_c}$, we obtain \eqref{eq:ind6}. This completes the proof of the proposition.
\end{proof}

We conclude this section by arguing that provided $\eX$ has at least one component with continuous distribution, under Assumptions \ref{as:newey_modified2} \ref{ass:consistent_est}, \ref{ass:verify_entry_game} \ref{ass:entry_smooth_CDF}, and \ref{as:eg_support} \ref{as:eg_bounds}, the consistency of the covariance matrix estimator $\hat\Sigma_{n,\theta^*}$ required in \eqref{ass:consistent_cov} holds.
This follows from \eqref{eq:ind4}, arguing as in \citet[Theorem 3.4]{pow:sto:sto89}, leveraging Assumption \ref{as:newey_modified2} \ref{ass:consistent_est} and that for $\bar y=(0,0), (1,1)$ the score does not depend on $\ps{n,y|x}$ together with Assumption \ref{as:eg_support}  \ref{as:eg_bounds}.

\subsection{Uniform Convergence for Series Estimators of $\ps{0}$}\label{app:verify:series}
In this section we provide sufficient conditions under which it is possible to verify Assumption~3(ii)' in Theorem~\ref{thm:uniform_coverage} through an application of results in \citet{ChenChristensen2015}.
A formal verification of Assumption~3(ii)' is available from the authors upon request.
For simplicity, we focus on the setting where all components of $X$ are continuous.
We let $b^K(x)=(b_{K1}(x),\dots,b_{KK}(x))'\in\mathbb R^K$ be a collection of $K$ basis functions, and let $B=(b^K(X_1),\dots,b^K(X_n))'\in \mathbb R^{n\times K}$. We let $\mathbf 1_y=(1\{Y_1=y\},\dots,1\{Y_n=y\})'$.
We approximate each $\ps{0}(y|\cdot)$ by the series estimator:
\begin{align*}
    \hat p_n(y|x) = b^K(x)'(B'B)^-B'\mathbf 1_y.
\end{align*}
We let $\text{BSpl}(K,[0,1]^{d_X},\gamma)$ denote a B-spline sieve of degree $\gamma$ and dimension $K$ on the domain $[0, 1]^{d_X}$, and $\text{Wav}(K,[0,1]^{d_X},\gamma)$ denote a Wavelet sieve basis of regularity $\gamma$ and dimension $K$ on the domain $[0, 1]^{d_X}$.
The construction of these sieve spaces is discussed in \citet[Section 6]{ChenChristensen2015}.
We maintain the following restrictions:
\begin{assumptionA}\label{ass:series:est}
    \textbf{(1)} $\cX=[0,1]^{d_X}$ and $\ps{0,\ex}$ is uniformly bounded away from zero on $\cX$. 
    \textbf{(2)} Assumption~\ref{as:eg_support}-\ref{ass:support:continuous} holds and $\|\ps{y|x}\|_{\mathcal H}=\sup_{y\in \mathcal Y}\sup_{x\in\cX}|\ps{y|x}(y|x)|$.
    \textbf{(3)} The sieve space is either $\text{BSpl}(K,[0,1]^{d_X},\gamma)$ or $\text{Wav}(K ,[0, 1]^{d_X},\gamma)$ with $\gamma>\max\{\alpha,1\}$.
    \textbf{(4)} $\lambda_{K,n}\equiv\sup_{P\in \mathcal P}[\lambda_{min}(\E_P[b^K(X_i)b^K(X_i)'])]^{-1/2}\lesssim 1$.
\end{assumptionA}
Assumption~\ref{ass:series:est}-(4) essentially assumes a uniform lower bound on the minimum eigenvalue of $\E_P[b^K(X_i)b^K(X_i)']$. Under Assumption~\ref{ass:series:est}, one can leverage Theorems 2.1 and 3.4 and Lemma 2.3 in \citet{ChenChristensen2015} to show that if $K \asymp (n/\ln n)^{d_X/(2\alpha+d_X)}$, then $\|\hat p_n-\ps{0,y|x}\|_\cH=O_{\mathcal P}\big((n/\ln n)^{-\tfrac{\alpha}{2\alpha+d_X}}\big)$.
In particular, one can attain the desired rate $o_{\mathcal P}(n^{-1/4})$ if $\alpha>d_X/2$.

\section{Additional Examples}\label{app:examples}
\begin{exampleA}[Discrete choice with unobserved heterogeneity in choice sets]
\label{example:BCMT}
Consider a discrete choice model, with a finite universe of alternatives $\mathcal{J} = \{1,\dots,J\}$.
Let each alternative be characterized by a vector of covariates $\ex_j$, which might vary across decision makers, and let $\ex=[\ex_j,j\in\mathcal{J}]$.
Let $\eu\in\R^{d_\eu}$ denote a vector of decision maker and/or alternative specific unobservable attributes.
As in the model proposed by \citet{bar:cou:mol:tei21}, the decision maker draws a \emph{choice set} $C\subseteq\mathcal{J}$ according to an unknown distribution, with $\PP(|C|\ge\kappa)=1$ for some known $\kappa\ge 2$, and chooses the alternative $\ey\in C$ that maximizes utility $\pi(\ex_{j},\eu;\theta)$,
$\ey\in\argmax_{j\in C} \pi(\ex_{j},\eu;\theta)$, with $\pi(\ex_{j},\eu;\theta)$ defined for all $j\in\mathcal{J}$.
The researcher observes $(\ey,\ex)$, but not $C$, and wishes to learn features of $\theta$ and the distribution of $\eu$. For given $\theta\in\Theta$ and $x\in\cX$,  \citet[Lemma A.1]{bar:cou:mol:tei21} show that the set of model implied optimal choices is a measurable correspondence given by the $J-\kappa+1$ best alternatives in $\cJ$, so that 
$\G(\eu|x;\theta)=\cup_{K\subseteq\mathcal{J}:|K|=\kappa}\left\{\arg\max_{j \in K}  \pi(x_{j},\eu;\theta)\right\}$.

We depict it in Panel (a) of Figure~\ref{fig:G_sets}, for $|\mathcal{J}|=3$ and $\pi(x_{j},\eu;\theta)=\pi(\ex_{j};\theta)+\eu_j$ with $\eu=(\eu_j,j\in\mathcal{J})$, as a function of $(u_1-u_3,u_2-u_3)$.
\begin{figure}
\centering
\includegraphics[scale=1.2]{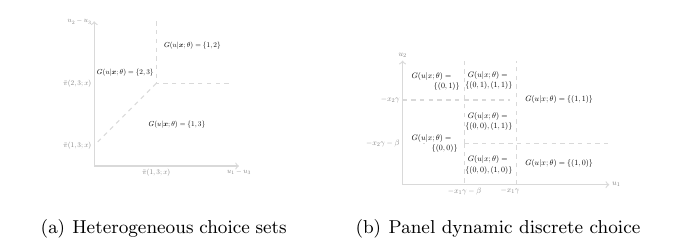}
\caption{\footnotesize{Stylized depictions of $\G(\cdot|x;\theta)$ in Example \ref{example:BCMT} (Panel (a), with $\pi(\ex_{j},\eu;\theta)=\pi(\ex_{j};\theta)+\eu_j$, $\mathcal{J}=\{1,2,3\}$, $\kappa=2$, and $\bar{\pi}(j,k;x) \equiv \pi(x_{k};\theta)-\pi(x_{j};\theta)$) and Example \ref{example:panel} (Panel (b), with $\beta\ge 0$).}}
		\label{fig:G_sets}
\end{figure}
In this example, if $\eu$ has full support on $\R^J$, Assumption \ref{ass:finite_support}-(e)(f) are immediately satisfied with $\cA_G=\{K\subset \cY:|K|\ge J-\kappa+1\}$ because each set of alternatives in $\cJ$ of size $J-\kappa+1$ can realize as the $J-\kappa+1$ best and $\cA^{(*e)}=\{(y,y'):y\neq y',~y,y'\in\cY\}$, with no more than two inequalities binding at $\qs{\theta,y|x}^*$ and one of $\qs{\theta,y|x}^*(y)$ a linear combination of the others using the simplex equality.
Assumption \ref{ass:finite_support}(b)(d) can be verified similarly to how they are verified for Example \ref{example:CT}.

It is also instructive to think about whether the introduction of a selection mechanism can allow for application of the method proposed in \citet{Chen_2018} to this example.\footnote{In the case of the entry game in Example \ref{example:CT}, the selection mechanism in \eqref{eq:exampleCT_selection_mech} can be integrated out against the distribution of $\eu$ to obtain a function that plays the role of the nuisance parameter in \citet{Chen_2018}.}
Let $\PP(\ey_i=j|\ex_i;\theta,R)$ denote the model-implied conditional probability that alternative $j\in\cJ$ is chosen given $\ex_i$ and $(\theta,R)$, where $R(\cdot;\ex_i,\eu_i)$ denotes the conditional probability mass function of $C_i$ given $(\ex_i,\eu_i)$. For all $j\in\cJ$,
\begin{equation*}
\PP(\ey_i=j|\ex_i;\theta,R)=\int\sum_{K\subseteq\cJ}\one\left(\argmax_{k\in K} \left(\pi(\ex_k;\theta)+u_k\right)=j\right)R(K;\ex_i,u)d\f.
\end{equation*}
To be able to apply \citeposs{Chen_2018} method, one needs to further restrict the model and assume that $R$ does not depend on $\eu$, in which case $R(\cdot;\ex_i)$ can come out of the integral.
Doing so, however, severely restricts the class of models to which the procedure is applicable, since it requires the distribution of choice sets to be independent of preferences.
Important examples of choice set formation mechanisms that violate this requirement include sequential search, rational inattention, and elimination by aspects (when the aspect with respect to which elimination occurs is the unobserved characteristic $\eu_j$). \hfill $\square$
\end{exampleA}
\smallskip

\noindent\textbf{Example C.1} (Continued--Geometric Properties of $\Theta^*(\ps{0})$).\label{example:BCMT:geometry} We now specialize Example~\ref{example:BCMT} to an instance where the mapping $\ps{}\mapsto\Theta^*(\ps{})$ is continuous, as discussed in Section~\ref{subsec:geometry_Theta_star}, but $\Theta^*(\ps{0})$ shrinks to a singleton as the amount of misspecification increases.
Let $\mathcal{J}=\{1,2,3\}$, $\kappa=2$, $\ex_j=[c_j~~b_j\ez^*]$, $\ez^*\in\{z_H,z_L\}$ and $z_L<z_H$, with $c_1<c_2<c_3$ and $b_1>b_2>b_3$ known constants (so that $\ex_j$ varies stochastically across decision makers, but non-stochastically across alternatives). 
Let $\pi(\ex_j,\eu;\theta)=-((1-\eu)b_j\ez^*+\eu(c_j+b_j\ez^*))$ and $\f(u)=u^\theta$.
This example is inspired by the study in \citet[Figure 1-a, Section 6.1.1]{bar:cou:mol:tei21} of decision making under risk.\footnote{Here, $c_j$ represents the deductible associated with insurance product $j$ and $b_j\ez$ the product's price.} For analytic tractability, we use a utility function linear in $\ez^*$.
It follows that, given $\ez^*=z$ alternative $1$ is preferred to $3$ if and only if $\eu>\tfrac{b_1-b_3}{c_3-c_1}z$; henceforth, let $A\equiv\tfrac{b_1-b_3}{c_3-c_1}>0$. Here, $\G\in\{\{1,2\},\{2,3\}\}$ and  $\Ps_0(\G=\{1,2\}|\ez^*=z)=1-(Az)^\theta\equiv\eta(\theta;z)$. We assume $(b_j,c_j),j\in\{1,2,3\}$, are such that $Az\in(0,1)$ for both $z=z_L,z_H$; and that $\Ps_0(C=\{1,3\}|\ez^*=z,\eu)=1$ if $\eu>Az$, $\Ps_0(C=\{1,3\}|\ez^*=z,\eu)=0.5$ if $\eu\le Az$. Hence, the true data generating process satisfies $\ps{0,y|z}(1|z)=1-(Az)^\theta$, $\ps{0,y|z}(2|z)=\ps{0,y|z}(3|z)=\tfrac{1}{2}(Az)^\theta$.
We generate misspecification similarly to what we did in Example~\ref{example:CT} on p.\pageref{example:entry_uniform}: the researcher observes a misclassified covariate $\ez$ with $\Ps_0(\ez^*=z_H|\ez=z)=\kappa(z,\gamma)\equiv(1-\gamma)\one(z=z_H)+\gamma\one(z=z_L)$. 

We then have that for $y=2,3$, $\ps{\gamma,y|z}(y|z)=\tfrac{1}{2}[(1-\gamma)(Az_H)^\theta+\gamma(Az_L)^\theta]$ if $z=z_H$ and $\ps{\gamma,y|z}(y|z)=\tfrac{1}{2}[\gamma(Az_H)^\theta+(1-\gamma)(Az_L)^\theta]$ if $z=z_L$; and $\ps{\gamma,y|z}(1|z)=1-\ps{\gamma,2|x}(y|z)-\ps{\gamma,3|x}(y|z)$.
If $\gamma=0$, the model is correctly specified and 
\begin{align}
    \Theta_I(\ps{0})=\left[\max_{\small z\in\{z_L,z_H\}}\tfrac{\ln(1-\ps{0,y|z}(1|z))}{\ln(Az)},\min_{\small z\in\{z_L,z_H\}}\tfrac{\ln(1-\ps{0,y|z}(1|z)-\ps{0,y|z}(2|z))}{\ln(Az)}\right].\label{eq:ThetaI:BCMT}
\end{align}
The bound in \eqref{eq:ThetaI:BCMT} has the familiar form of an intersection of intervals.
Intuitively, for $\gamma>0$ this bound remains non-empty provided the distributions $\ps{\gamma,y|z}(\cdot|z_L)$ and $\ps{\gamma,y|z}(\cdot|z_H)$ are sufficiently compatible with each other, and in that case $\Theta^*(\ps{\gamma})=\Theta_I(\ps{\gamma})$ because the misspecification is non-detectable: it is possible that $\ps{\gamma}$ belongs to the model $\fQ$ as in Definition~\ref{def:correct_misspec}.
But for $\gamma$ sufficiently large, the distributions $\ps{\gamma,y|z}(\cdot|z_L)$ and $\ps{\gamma,y|z}(\cdot|z_H)$ become incompatible, $\ps{\gamma}\notin\fQ$, $\Theta_I(\ps{\gamma})=\emptyset$, and $\Theta^*(\ps{\gamma})$ becomes a singleton.
In what follows we more formally explain these facts.

In this example, $\fq{\theta,z}=\{q\in\Delta:q(1)\le\eta(\theta;z),~q_3\le 1-(\theta;z)\}$.
By similar steps to those used for Example~\ref{example:CT}, one can show that in this example, $\qs{\theta,y|z}^*(y|z)=\ps{\gamma,y|z}(y|z)$ for $y=1,2,3$ if $\theta\in\Theta_1(z;\ps{\gamma,y|z})\equiv\{\theta:\ps{\gamma,y|z}(1|z)\le\eta(\theta;z)\le \ps{\gamma,y|z}(1|z)+\ps{\gamma,y|z}(2|z)\}$; $\qs{\theta,y|z}^*(y|z)=\tfrac{\ps{\gamma,y|z}(y|z)}{\ps{\gamma,y|z}(1|z)+\ps{\gamma,y|z}(2|z)}\eta(\theta;z)$ for $y=1,2$ and $\qs{\theta,y|z}^*(3|z)=1-\eta(\theta;z)$ if $\theta\in\Theta_2(z;\ps{\gamma,y|z})\equiv\{\theta:\eta(\theta;z)> \ps{\gamma,y|z}(1|z)+\ps{\gamma,y|z}(2|z)\}$; and $\qs{\theta,y|z}^*(y|z)=\tfrac{\ps{\gamma,y|z}(y|z)}{\ps{\gamma,y|z}(2|z)+\ps{\gamma,y|z}(3|z)}(1-\eta(\theta;z))$ for $y=2,3$ and $\qs{\theta,y|z}^*(1|z)=\eta(\theta;z)$ if $\theta\in\Theta_3(z;\ps{\gamma,y|z})\equiv\{\theta:\eta(\theta;z)<\ps{\gamma,y|z}(1|z)\}$.
Concerning the score function, we have $\score{\theta}{y|z;\ps{\gamma,y|z}}=0$ if $\theta\in\Theta_1(z;\ps{\gamma,y|z})$; 
$\score{\theta}{y|z;\ps{\gamma,y|z}}=\tfrac{-\ln(Az)(Az)^\theta}{1-(Az)^\theta}$ for $y=1,2$ and $\score{\theta}{3|z;\ps{\gamma,y|z}}=\ln Az$ if $\theta\in\Theta_2(z;\ps{\gamma,y|z})$; 
$\score{\theta}{1|z;\ps{\gamma,y|z}}=\tfrac{-\ln(Az)(Az)^\theta}{1-(Az)^\theta}$ and $\score{\theta}{y|z;\ps{\gamma,y|z}}=\ln Az$ for $y=2,3$
if $\theta\in\Theta_3(z;\ps{\gamma,y|z})$.

Next, let $\bar{\gamma}$ be defined so that $\ps{\bar\gamma,y|z}(1|z_H)=\ps{\bar\gamma,y|z}(1|z_L)+\ps{\bar\gamma,y|z}(2|z_L)$.
Algebraic manipulations show that if $\gamma\le\bar\gamma$, then $\Theta^*(\ps{\gamma})=\Theta_I(\ps{\gamma})$ as in \eqref{eq:ThetaI:BCMT} with $\ps{0}$ replaced by $\ps{\gamma}$.
However, $\Theta_I(\ps{\gamma})=\emptyset$ for all $\gamma\in(\bar\gamma,1]$.
Yet, recalling that at any $\theta^*\in\Theta^*(\ps{\gamma})$ it must hold that $\E[\score{\theta}{\ey|\ez;\ps{\gamma,y|z}}]=0$, one can show that  $\E[\score{\theta}{\ey|\ez;\ps{\gamma,y|z}}]=0$ if and only if either
\begin{align*}
\tfrac{\ln(Az_L)\left[\ps{\gamma,y|z}(3|z_L)-(Az_L)^{\theta}\right]}{1-(Az_L)^\theta}
+\tfrac{\ln(Az_H)\left[1-(Az_H)^{\theta}-\ps{\gamma,y|z}(1|z_H)\right]}{1-(Az_H)^\theta}
&=0,\quad\quad\text{or}\\
\tfrac{\ln(Az_H)\left[\ps{\gamma,y|z}(3|z_H)-(Az_H)^{\theta}\right]}{1-(Az_H)^\theta}
+\tfrac{\ln(Az_L)\left[1-(Az_L)^{\theta}-\ps{\gamma,y|z}(1|z_L)\right]}{1-(Az_L)^\theta}&=0,
\end{align*}
but not both. Then, $\Theta^*(\ps{\gamma})$ equals the singleton value of $\theta$ that solves the one that holds.$\,$\hfill $\square$
\vspace{.2cm}

\begin{exampleA}[Panel dynamic discrete choice]
\label{example:panel}
Decision maker $i$ chooses between actions $y=0$ and $y=1$ across multiple time periods, according to
\begin{align*}
\ey_{it}=1\{\ex_{it}\gamma+\ey_{it-1}\beta+\alpha_i+\epsilon_{it}\ge 0\},~i=1,\dots, n, ~t=1,\dots,T,
\end{align*}
with $\ey_{it}$ their decision in period $t$, $\ex_{it}$ a vector of observed covariates in period $t$, $\alpha_i$ an individual-specific unobserved effect that is fixed over time, and $\epsilon_{it}$ an idiosyncratic unobserved effect that varies over time.
When $\beta\neq 0$,  period $t$'s choice depends on previous periods' choices, introducing state dependence. 
The researcher observes $(\ey_{it},\ex_{it})$ for $i=1,\dots,n$ and $t=1,\dots,T$, but does not observe $\ey_{i0}$, so that $\{\ey_{i1},\dots,\ey_{iT}\}$ is not fully determined and the model is incomplete \citep{heckman78,hon:tam06}.
Nonetheless, for given $(\ex_{it},\alpha_i,\epsilon_{it}),~t=1,\dots,T$, the model constrains the possible values that $(\ey_{it},t=1,\dots,T)$ can take.
For example, with $T=2$, one has:
\begin{align*}
 \ey_{i1}&=\left\{\begin{array}{ll}
1\{\ex_{i1}^\top\gamma +\alpha_i+\epsilon_{i1}\ge 0\} & \text{~~~if } \ey_{i0}=0,\\
1\{\ex_{i1}^\top\gamma+\beta+\alpha_i+\epsilon_{i1}\ge 0\} & \text{~~~if } \ey_{i0}=1,
\end{array}  \right. \\
\ey_{i2}&= 1\{\ex_{i2}^\top\gamma+\ey_{i1}\beta+\alpha_i+\epsilon_{i2}\ge 0\}~~~ \text{if } \ey_{i0}=0~\text{or}~\ey_{i0}=1.
\end{align*}

Denoting the unobservables as $\eu_{it}\equiv\alpha_i+\epsilon_{it}$, for given $\theta=(\gamma,\beta)\in\Theta$ and $x\in\cX$, \citet{che:kai21} derive the measurable \citep[e.g.,][Example 1.5]{mol:mol18} correspondence $\G(\cdot|x;\theta)$ as the set of values $(y_1,y_2)\in\{0,1\}^2$ that satisfy the above equations.
The correspondence is depicted in Panel (b) of Figure~\ref{fig:G_sets} as a function of $(u_1,u_2)$ for the case that $\beta\ge 0$.
Similar examples arise in nonparametric models of state dependence \citep[e.g.,][]{tor19}.
In this example, Assumption \ref{ass:finite_support}(b)(d) can be verified similarly to how they are verified for Example \ref{example:CT}.
If the parameter space for $\beta$ is a subset of $\R_{++}$, Assumption~\ref{ass:finite_support}-(e) is satisfied: for all $\theta\in\Theta$, $\cA_\G=\{\{(0,0)\},\{(0,1)\},\{(1,0)\},\{(1,1)\},\{(0,1),(1,1)\},\{(0,0),(1,1)\},\{(0,0),(1,0)\}\}$.
To verify Assumption~\ref{ass:finite_support}-(f), let $\eu$ have positive density on $\R^2$. Denote $a_1\equiv \f(S_{\{(0,0)\}|x;\theta})$, $a_2\equiv\f(S_{\{(0,1)\}|x;\theta}), a_3\equiv\f(S_{\{(0,0),(1,0)\}|x;\theta})$, $a_4\equiv\f(S_{\{(0,0),(1,1)\}|x;\theta})$, $ a_5\equiv\f(S_{\{(0,1),(1,1)\}|x;\theta})$, $a_6\equiv\f(S_{\{(1,0)\}|x;\theta})$, $a_7\equiv\f(S_{\{(1,1)\}|x;\theta})$. Then, $\sum_{j=1}^7 a_j=1$ and $a_1,\dots,a_7>0$ as $\beta>0$. One can show that $\fq{\theta}=\big\{\qs{y|x}\in\Delta:~a_2~{\le}_{(1)}~\qs{y|x}((0,1)|x)~$ ${\le}_{(2)}~a_2+a_5;~a_6~{\le}_{(3)}~\qs{y|x}((1,0)|x)$ ${\le}_{(4)}~a_6+a_3;~\qs{y|x}((0,0)|x)+\qs{y|x}((0,1)|x)~{\ge}_{(5)}~a_1+a_2;~\qs{y|x}((0,1)|x)+\qs{y|x}((1,1)|x)~{\ge}_{(6)} ~a_2+a_5+a_7,~x\in\cX\big\}$.
Here, $\fq{\theta}$ contains $\qs{y|x}$ such that $\qs{y|x}(y|x)>0$ for all $y\in\cY$. By Assumption~\ref{ass:finite_support}(d), these nonnegativity constraints do not bind at $\qs{\theta,y|x}^*$. Hence, we focus on inequalities (1)-(6).
We show that $|\cA^{(*e)}_{\mathrm{bind}}(\qs{y|x})|\le 3$ for all $\qs{\theta,y|x}\in\fq{\theta,x}$. Since $a_3,a_5>0$, an active set cannot contain both inequalities (1) and (2), nor both (3) and (4).
Hence, if an active set included four inequalities, it should include (5) and (6), one of $\{(1),(2)\}$, and one of $\{(3),(4)\}$. Consider $\cA^{(*e)}_{\mathrm{bind}}(\qs{y|x})=\{(1),(3),(5),(6)\}$. Then $\qs{y|x}((0,1)|x)=a_2,\qs{y|x}((1,0)|x)=a_6,\qs{y|x}((0,0)|x)=a_1,\qs{y|x}((1,1)|x)=a_5+a_7$, and using that these probabilities sum up to 1, we get $a_1+a_2+a_5+a_6+a_7=1$. Since also $\sum_{j=1}^7 a_j=1$, this forces $a_3+a_4=0$, yielding a contradiction (the other cases are similar).
Finally, we show that for any $\qs{y|x}\in\fq{\theta,x}$ with $\qs{y|x}(y|x)>0,\forall y\in\cY$, every inequality indexed by $A\in \cA^{(*e)}_{\mathrm{bind}}(\qs{y|x})\cup\cA^{(*e)}_=$ is irredundant.
Write the simplex equality $h(q)\equiv \sum_{y\in\cY}\qs{y|x}(y)-1$ and use the sum-to-one constraint so that each of the six inequalities can be written as $g_j(q)\ge 0,j=1,\dots,6$ with $\nabla_q g_1=(0,1,0,0)$,
$\nabla_q g_2=(1,0,1,1)$, 
$\nabla_q g_3=(0,0,1,0)$,
$\nabla_q g_4=(1,1,0,1)$,
$\nabla_q g_5=(1,1,0,0)$,
$\nabla_q g_6=(0,1,0,1)$.
By the previous argument, no set of binding inequalities can contain both $\{(1),(2)\}$ or $\{(3),(4)\}$. Inspecting the gradients of $g_j$ shows: (i) for any admissible collection of jointly binding inequalities, the associated gradients of $g_j$ are linearly independent, and (ii) $\nabla_q h = (1,1,1,1)$ lies in the span of a set of inequality constraints only if the set of constraints contains both contains $\{(1),(2)\}$ or $\{(3),(4)\}$. As the latter condition cannot hold, $\nabla_q h$ and the gradients of inequalities in $\cA^{(*e)}_{\mathrm{bind}}(\qs{y|x})$ are linearly independent.
Consider the affine hull of the simplex: $\mathsf H\equiv\{q\in\mathbb{R}^4:h(q)=0\}$, which has dimension $3$. Let $L_{\mathcal A_{\mathrm{bind}}(\qs{y|x})} \equiv \mathsf H \cap \bigl(\bigcap_{i\in \mathcal A_{\mathrm{bind}}(\qs{y|x})}\{q:\ g_i(q)=0\}\bigr)$.
The face containing $\qs{y|x}$ is $F_{\mathcal A_{\mathrm{bind}}(\qs{y|x})}=L_{\mathcal A_{\mathrm{bind}}(\qs{y|x})}\cap \fq{\theta,x}$.
Since every inequality not in $\mathcal A_{\mathrm{bind}}(\qs{y|x})$ has strict slack at $\qs{y|x}$, there exists a neighborhood $U$ of $\qs{y|x}$ such that $F_{\mathcal A_{\mathrm{bind}}(\qs{y|x})}\cap U=L_{\mathcal A_{\mathrm{bind}}(\qs{y|x})}\cap U$. On this neighborhood, the feasible set is cut out only by $\mathsf H$ and the hyperplanes $\{q:\ g_i(q)=0\},i\in \mathcal A_{\mathrm{bind}}(\qs{y|x})$.
Using that $\nabla_q h$ together with the active $\nabla_q g_i$ are linearly independent, $F_{\mathcal A_{\mathrm{bind}}(\qs{y|x})}$ has dimension $3-|\mathcal A_{\mathrm{bind}}(\qs{y|x})|$.
Removing any single active constraint $g_i\ge 0$ (with $i\in \mathcal A_{\mathrm{bind}}(\qs{y|x})$) 
increases the face dimension by $1$, and the face is strictly enlarged.\hfill $\square$
\end{exampleA}

\section{Allowing $\supp(\G(\cdot|x;\theta))$ to Depend on $\theta$}\label{app:clarke}
Assumption~\ref{ass:finite_support}(e)(f) are restrictive. These conditions can be dispensed with through the use of Clarke’s subdifferentials \citep{clarke1990optimization}, at the cost of a more cumbersome procedure as we show here. 
We illustrate the approach for the two player entry game with payoffs given in Example~\ref{example:CT} and let $(\delta_1,\delta_2)$ belong to parameter space $[\underline{\delta},\bar{\delta}]^2\ni 0$. This is an instructive case because $\supp(\G(\cdot|x;\theta))$ changes depending on which quadrant $(\delta_1,\delta_2)$ belong to, while being invariant to $\theta$ for $(\delta_1,\delta_2)$ in the interior of a quadrant, and for $\delta_1\cdot\delta_2=0$ the region of multiplicity in Figure~\ref{fig:G_entry} disappears and Assumption~\ref{ass:finite_support}(e)(f) fail. When $\mathrm{sign}(\delta_1\cdot\delta_2)<0$, a PSNE does not exist for certain values of $\theta$ \citep[see, e.g.][Figure 2]{tam03} and following \citet[Appendix D]{ber:mol:mol11} we let $\G(\cdot|\theta)=\cY$ for these values of $\theta$. Throughout, we assume that $\Theta$ only includes values for $\mathrm{Corr}(\eu_1,\eu_2)$ bounded away from $\pm1$.

The value function is $V(\theta|x)=\sup_{\qs{}\in\fq{\theta,x}}\sum_{y\in\cY} \ps{0,y|x}(y|x)\ln\qs{}(y)$, with $\fq{\theta,x}=\fq{\theta,x}^0$ if $\delta_1\cdot\delta_2=0$, $\fq{\theta,x}=\fq{\theta,x}^\mathrm{I}$ if $\delta_1>0,\delta_2>0$, $\fq{\theta,x}=\fq{\theta,x}^\mathrm{II}$ if $\mathrm{sign}(\delta_1\cdot\delta_2)<0$,  
and $\fq{\theta,x}=\fq{\theta,x}^\mathrm{III}$ if $\delta_1<0,\delta_2<0$. 
The set $\fq{\theta,x}^\mathrm{III}$ corresponds to the one in \eqref{eq:frak_q_CT}, while
\begin{align*}
    \fq{\theta,x}^0= \Big\{\qs{}\in\Delta:~\qs{}((0,0))&=\f((-\infty,-x_1\beta_1),(-\infty,-x_2\beta_2));\\
        \qs{}((1,1))&=\f([-x_1\beta_1,\infty),[-x_2\beta_2,\infty));\\
	  \qs{}((1,0))&=\f([-x_1\beta_1,\infty),(-\infty,-x_2\beta_2));\\
      \qs{}((0,1))&= \f((-\infty,-x_1\beta_1),[-x_2\beta_2,\infty))\Big\},\\
    \fq{\theta,x}^\mathrm{I}= 
        \Big\{\qs{}\in\Delta:~\qs{}((0,0))&\ge \f((-\infty,-x_1\beta_1-\delta_1),(-\infty,-x_2\beta_2))\\
      &\quad\quad+\f([-x_1\beta_1-\delta_1,-x_1\beta_1),(-\infty,-x_2\beta_2-\delta_2))\\
      \qs{}((0,0))&\leq \f((-\infty,-x_1\beta_1),(-\infty,-x_2\beta_2));\\
      \qs{}((1,0))&=\f([-x_1\beta_1,\infty),(-\infty,-x_2\beta_2-\delta_2));\\
        \qs{}((0,1))&=\f((-\infty,-x_1\beta_1-\delta_1),[-x_2\beta_2,\infty))\Big\},\\
    \fq{\theta,x}^\mathrm{II}= 
        \Big\{\qs{}\in\Delta:~\qs{}((0,0))&\ge\f((-\infty,-x_1\beta_1),(-\infty,-x_2\beta_2));\\
        \qs{}((1,1))&\ge\f([-x_1\beta_1-\delta_1,\infty),[-x_2\beta_2-\delta_2,\infty));\\
	  \qs{}((1,0))&\ge \f([-x_1\beta_1,\infty),(-\infty,-x_2\beta_2-\delta_2))\\
       \qs{}((0,1))&\geq \f((-\infty,-x_1\beta_1-\delta_1),(-x_2\beta_2,\infty))\Big\},
\end{align*}
with $\Delta$ the unit simplex in $\R^4$. 
We assume that $\f((-\infty,t_1],(-\infty,t_2])$ is jointly continuous in $(t_1,t_2,\theta)$ and for all $t_1,t_2\in\R$, $\f((-\infty,t_1],(-\infty,t_2])$, $\f([t_1,\infty),[t_2,\infty))$, $\f([t_1,\infty),(-\infty,t_2])$, $\f((-\infty,t_1],[t_2,\infty))\in\R_{++}$.
By Theorem~\ref{thm:score_characterization}-(i), for $\theta\in\Theta$ such that $(\delta_1,\delta_2)$ is in the interior of a quadrant, $V(\theta|x)$ is continuously differentiable in $\theta$, and provided $\qs{\theta,y|x}^*(y|x)>c$, $\Ps_0-a.s.$ the gradient is bounded.
Let $V^{\mathrm{I}}$ denote $V(\theta|x)$ for $\delta_1,\delta_2>0$, and use similar notation for the other cases and for the gradient of $\nabla V(\theta|x)$.
In our entry game example, when $\delta_1<0,\delta_2<0$ we proved that $\nabla V^{\mathrm{III}}$ equals the inner product of $\ps{0,y|x}$ with the score vector in \eqref{eq:eg_score1}-\eqref{eq:eg_score4}.
A similar derivation can be carried out to obtain $\nabla V^{\mathrm{I}}$ when $\delta_1>0,\delta_2>0$ and $\nabla V^{\mathrm{II}}$ when $\delta_1\cdot\delta_2<0$.
Moreover, one can show that whenever $\delta_1\cdot\delta_2=0$, the formulas for $V$ coming from the adjacent quadrants coincide. Hence $V$ is continuous on $\Theta$.
Because each gradient $\nabla V^{\mathrm{I}},\nabla V^{\mathrm{II}},\nabla V^{\mathrm{III}}$ is continuous and $\Theta$ is compact, the gradients are bounded on each quadrant. 
As there are finitely many quadrants, $V(\theta|x)$ is locally Lipschitz on $\theta$ \citep[][p.25]{clarke1990optimization}; however, $V(\theta|x)$ fails to be differentiable at $\theta:\delta_1\delta_2=0$.
By Theorem 2.5.1 in \citet[p.63]{clarke1990optimization}, for $T$ any set of Lebesgue mesaure zero in $\R^{d_\theta}$, the Clarke \emph{generalized gradient} of $V$ at $\theta$ is
\begin{align*}
    \partial V^\circ(\theta|x)=\mathrm{conv}\left\{\lim\nabla V(\theta_i):\theta_i\to\theta,\theta_i\notin T,\theta_i:\delta_{i1}\delta_{i2}\neq 0\right\}.
\end{align*}
Let $\overline{\nabla V}^{\mathrm{J}},~\mathrm{J}\in\{\mathrm{I},\mathrm{II},\mathrm{III}\}$ denote the continuous extensions of $\nabla V^{\mathrm{J}}$ to the closure of the quadrant associated with $\fq{\theta,x}^{\mathrm{J}}$, denoted $\Theta^{\mathrm{J}}$.
Then
\begin{align*}
    \partial V^\circ(\theta|x)=
    \begin{cases}
    \{\nabla V^{\mathrm{J}}(\theta|x)\}, & \theta\in\Theta^{\mathrm{J}},~\text{with}~\delta_1\cdot\delta_2\neq 0,\\
    \mathrm{conv}\{\overline{\nabla V}^{\mathrm{I}}(\theta|x),\overline{\nabla V}^{\mathrm{II}}(\theta|x)\}, & \delta_1=0, \delta_2>0,~~\text{or}~~\delta_1>0, \delta_2=0,\\
    \mathrm{conv}\{\overline{\nabla V}^{\mathrm{II}}(\theta|x),\overline{\nabla V}^{\mathrm{III}}(\theta|x)\}, & \delta_1=0, \delta_2<0,~~\text{or}~~\delta_1<0, \delta_2=0,\\
    \mathrm{conv}\{\overline{\nabla V}^{\mathrm{I}}(\theta|x),\overline{\nabla V}^{\mathrm{II}}(\theta|x),\overline{\nabla V}^{\mathrm{III}}(\theta|x)\}, & \delta_1=\delta_2=0.
\end{cases}
\end{align*}

Next, Theorem~\ref{thm:score_characterization}-(ii) shows that $\E[\score{\theta}{\ey|\ex;\ps{0,y|x}}]=0$ for all $\theta\in\Theta^*(\ps{0})$.
Correspondingly, we first note that when vectors $\theta:\delta_1\cdot\delta_2=0$ belong to $\Theta$, $\partial V^\circ(\theta)$ is set valued and no longer a singleton.
Correspondingly, we use the Aumann expectation of $\partial V^\circ(\theta|\ex)$ \citep[Chapter 3]{mol:mol18} to collect all the possible values of the expected score.
Armed with this set-valued expectation, we invoke Proposition 2.3.2 in \citet[p.38]{clarke1990optimization}, by which if $\crit(\theta)$ attains a local maximum at $\theta^*$, then
\begin{align*}
    0\in \E\left[\partial^\circ V(\theta^*|\ex)\right].
\end{align*}

For each $\theta\in\Theta$, by \citet[Proposition 2.1.2]{clarke1990optimization} and \citet[Theorem 3.11]{mol:mol18}, the set $\E\left[\partial^\circ V(\theta|\ex)\right]$ is closed and convex $\Ps_0$-a.s., hence we can represent it using its support function.
Moreover, its Clarke subdifferential is equal to the gradient for $\theta$ such that $V(\theta|\ex)$ is differentiable $\Ps_0$-a.s.
As the only non-differentiability points are $\theta:\delta_1\cdot\delta_2=0$, and this set of values is known, we can build a confidence set for $\theta$ that dispenses with Assumption~\ref{ass:finite_support}-(e)(f). Fix $\epsilon>0$, say $\epsilon=10^{-6}$, and let
\begin{align}
    \hspace{-.4cm}CS_n=\left\{\theta\in\Theta_{\epsilon}^c:T_n(\theta)\le c_{d_\theta,\alpha}\right\}\cup\big\{\theta\in\Theta_{\epsilon}:\sqrt{n}\mathrm{d}_H(0,\hat{\E}_n\left[\partial^\circ V(\theta|\ex_i)\right])\le c_{H,\alpha}(\theta)\big\},\label{eq:CS_subdifferential}
\end{align}
with, for any nonempty and compact sets $A,B$, $\mathrm{d}_H(A,B)=\sup_{a\in A}\inf_{b\in B}\|a-b\|$, and $\Theta_{\epsilon}=\left\{\theta\in\Theta:\vert\delta_1\cdot\delta_2\vert<\epsilon\right\}$, $\Theta_{\epsilon}^c=\Theta\setminus\Theta_{\epsilon}$, and $c_{H,\alpha}(\theta)$ the $1-\alpha$ quantile of the limit distribution of $\mathrm{d}_H(\E\left[\partial^\circ V(\theta|\ex)\right],\hat{\E}_n\left[\partial^\circ V(\theta|\ex_i)\right])$ (using that $\mathrm{d}_H(0,\hat{\E}_n\left[\partial^\circ V(\theta|\ex_i)\right])\le\mathrm{d}_H(\E\left[\partial^\circ V(\theta|\ex)\right],\hat{\E}_n\left[\partial^\circ V(\theta|\ex_i)\right])$ under the null that $0\in\E\left[\partial^\circ V(\theta|\ex)\right]$).
If $\partial^\circ V(\theta|\ex_i)$ were observed, we could adapt \citet{ber:mol08}'s results and consistent bootstrap procedure to estimate the critical values of the limit distribution, to establish asymptotic validity of $CS_n$ in \eqref{eq:CS_subdifferential}. 
However, in our application $\partial^\circ V(\theta|\ex_i)$ depends on $\ps{0,y|x}$, which is unknown and nonparametrically estimated in first stage. 
\citet{sem23} and \citet{liu:mol25v2} derive the limit distribution of support function processes similar to the one associated with $\mathrm{d}_H(\E\left[\partial^\circ V(\theta|\ex)\right],\hat{\E}_n\left[\partial^\circ V(\theta|\ex_i)\right])$, accounting for first-step nonparametric estimation, and they put forward consistent bootstrap procedures to estimate the critical values. Due to space constraints and the fact that our main proposal trades computationally tractability for the stronger Assumption~\ref{ass:finite_support}(e)(f), we leave extending their results to our context to future research. 

\newpage
\bibliographystyle{ecta-fullname} 
\bibliography{MisspRef,EwPI_biblio3}   
\end{document}